\documentclass[10pt]{article}
\usepackage{natbib,setspace,lscape,longtable}
\usepackage{natbib,epsfig,graphicx}
\usepackage{mathrsfs,amsmath,amsthm,amssymb,color, verbatim}
\allowdisplaybreaks
\newcommand{\keywords}[1]{\vspace{1ex}\noindent\textbf{Keywords:} #1}
\usepackage{multirow}
\usepackage{bm}
\usepackage{hyperref}
\usepackage{bbm}

\usepackage{dsfont}
\usepackage{tabularx}

\usepackage{dsfont}
\usepackage{IEEEtrantools}
\usepackage{xr} 

\usepackage{geometry}
\usepackage{float}
\usepackage{array}
\usepackage{enumerate}
\usepackage{comment}
\usepackage{setspace}
\usepackage{amsfonts}
\usepackage{indentfirst}
\usepackage{booktabs}
\usepackage{appendix}
\usepackage{caption}
\usepackage{graphicx, subfig}
\usepackage{algorithm}
\usepackage{algpseudocode}

\numberwithin{equation}{section}
\newtheorem{theorem}{Theorem}
\newtheorem{lemma}{Lemma}
\newtheorem{proposition}{Proposition}
\newtheorem{corollary}{Corollary}

\theoremstyle{definition}

\theoremstyle{remark}

\newtheorem{remark}{Remark}
\newtheorem{assumption}{Assumption}

\theoremstyle{plain}

\newtheorem{lemmaA}{Lemma}
\newtheorem{propositionA}{Proposition}

\theoremstyle{definition}

\theoremstyle{remark}

\def\1{\mathbf{1}}

\def\hat{\widehat}

\def\beq{\begin{equation}}
	\def\eeq{\end{equation}}
\def\beqs{\begin{equation*}}
	\def\eeqs{\end{equation*}}
\def\beqr{\begin{eqnarray}}
	\def\eeqr{\end{eqnarray}}
\def\beqrs{\begin{eqnarray*}}
	\def\eeqrs{\end{eqnarray*}}
\def\bet{\begin{theorem}}
	\def\eet{\end{theorem}}
\def\bel{\begin{lemma}}
	\def\eel{\end{lemma}}
\def\bep{\begin{proposition}}
	\def\eep{\end{proposition}}
\def\bg{\begin{figure}[tbph]\begin{center}}
		\def\eg{\end{center}\end{figure}}

\def\bc{\begin{center}}
	\def\ec{\end{center}}

\newcommand{\nmin}{\min\{n_1,n_2\}}
\newcommand{\maxrhon}{\max\{1,K\rho_{n_1,n_2}\}}
\newcommand{\maxsqrtrho}{\max\left\{1,\sqrt{K\rho}\right\}}
\newcommand{\maxrho}{\max\left\{1,K\rho\right\}}
\newcommand{\maxbeta}{\max\left\{1,\frac{K}{\beta K_2}\right\}}
\newcommand{\maxbarrho}{\max\left\{1,\bar{K}\rho\right\}}

\hypersetup{
	colorlinks=true,
	linkcolor={blue},
	filecolor={Maroon},
	citecolor={blue},
	urlcolor={blue}}

\title{Cross-Validation in Bipartite Networks}
\author{Bokai Yang\thanks{
		The first two authors contributed equally to this work. }\hspace{.2cm}\\
	Qiuzhen College, Tsinghua University\\
	and \\
	Yuanxing Chen \\
	Yau Mathematical Sciences Center, Tsinghua University\\
	and\\
	Yuhong Yang\thanks{
		Corresponding author. Email: yyangsc@mail.tsinghua.edu.cn.
	}\\
	Yau Mathematical Sciences Center, Tsinghua University}
	
\begin{document}
	\maketitle
	\begin{abstract}
		Bipartite networks, which encode interactions between two distinct types of entities, arise widely in applications and exhibit inherent asymmetry across node sets. Despite a growing literature on bipartite community detection, estimating community numbers $(K_1, K_2)$, a critical issue for bipartite network analysis, remains theoretically underdeveloped without any model selection consistency established, to our knowledge. Indeed, the inherent asymmetry and the two-dimensional parameter space with possibly drastically different $K_1$ and $K_2$ pose unique challenges that differ from unipartite cases. In particular, the candidate models may simultaneously overfit one node set while underfitting the other. To address these challenges, we propose Bipartite Cross-Validation (BCV), a penalized cross-validation framework that jointly selects $(K_1,K_2)$ in a fully data-driven manner. We establish the first model selection consistency for bipartite networks, notably accommodating the regime where the numbers of communities scale with the network size, revealing the intricate interplay between sparsity and model complexity. Simulations and real-data applications demonstrate strong finite-sample performance of BCV.
	\end{abstract}
	
	\section{Introduction}\label{introduction}
	
	Network data have attracted increasing attention as a powerful framework for representing complex relational structures in recent years. 
	Among them, bipartite networks, which characterize interactions between two distinct types of entities, arise naturally in many applications, including examples such as author-paper relationships in citation networks \citep{newman2001structure}, user-item interactions in recommendation systems \citep{zhang2011tagaware}, and legislator-bill connections in political science \citep{fowler2007cosponsor}. 
	A key feature of bipartite networks is the inherent asymmetry between the two node sets, which often exhibit drastically different structural patterns.
	As a result, collapsing them into a single homogeneous network may obscure important structural information.
	Instead, joint modeling of both node types is essential to capture their interdependent bipartite structure and uncover latent heterogeneity and functional organization.

	The study of bipartite community detection dates back to the seminal work of \citet{barber2007bimod}, who extended modularity maximization to two-mode networks through a ``bi-modularity'' objective. A large body of subsequent work has focused on developing algorithms for uncovering community structure, including methods based on modularity-type criteria \citep{Guimera2007module,liu2010community,pesantez2017efficient,ZHOU2018novel}, 
	one-mode projections of bipartite networks \citep{Alzahrani2016}, and hierarchical or Bayesian-MCMC formulations \citep{Larremore2014infer, larremore2020bayes,tamarit2020hierarchical}. 
	While these methods have demonstrated good empirical performance in certain settings, their treatment of the numbers of communities as tuning parameters still lacks a theoretical understanding. With the numbers of communities known, recent advances have established impressive theoretical guarantees (e.g., on membership recovery) for specific algorithms under the bipartite stochastic block model, such as spectral clustering \citep{Zhou2019AnalysisOS,braun2024strong} and pseudo-likelihood methods \citep{zhou2020optimal}. 
	This calls for a method to consistently select the numbers of communities for bipartite networks.
	
	However, establishing a consistent model selection framework for bipartite networks presents unique structural hurdles, which are further amplified in modern large-scale applications where the latent complexity, or, the numbers of communities, are modeled as growing with the network size. Firstly, due to the inherent asymmetry between node sets, the rank of the associated probability matrix is strictly capped by the smaller number of communities, rendering the side with more communities inherently rank-deficient. Secondly, unlike the one-dimensional complexity in unipartite networks, bipartite community counts occupy a two-dimensional grid $(K_1, K_2)$. This introduces the risk of simultaneously overfitting one side while underfitting the other, fundamentally complicating the bias–variance trade-off. Thirdly, possible substantial disparity between $K_1$ and $K_2$ poses severe challenges to identify the community structure on the rank-deficient side, subsequently hinders the model selection process by diminishing the theoretical separation between the ground-truth and misspecified models. Finally, these networks often exhibit highly imbalanced node sizes (e.g., 99 senators versus over 2000 bills in a later data example), which imposes stringent sparsity requirements: if the network is too sparse, the smaller node set may not accumulate enough edges to provide a sufficient signal for community detection.

	Model selection for unipartite networks has been extensively studied using likelihood-based methods \citep{wang2017likelihood}, spectral approaches \citep{le2022estimating}, and hypothesis testing procedures \citep{jin2023optimal,Bhadra2025unified}. More recent work allows for a diverging number of communities via exploiting a one-dimensional sequence of nested models \citep{lei2016goodness,wu2024eigengap}. 
	However, these procedures are difficult to extend to handle bipartite networks, due to the issues mentioned above.

	Cross-validation provides a natural and principled approach to model selection, with a well-established theoretical foundation in statistical learning \citep{stone1974cv,geisser1975pred,shao1993cv,van2003unified,yang2007consistency,arlot2010survey,lei2025moderncv}.
	In network analysis, recent works have successfully adapted cross-validation to unipartite stochastic block models, establishing consistency guarantees for selecting the number of communities \citep{chen2018network,li2020network,Chakrabarty2025subsampling,yang2025pnncv}. 
	However, these approaches face two hurdles in the bipartite setting. First, existing results are mostly confined to the fixed-$K$ regime and do not account for a diverging number of communities on a two-dimensional grid, which demands a uniform control on the large set of candidate models. Second, standard CV lacks a mechanism to handle the asymmetry between the two node sets. Without a tailored penalty, the procedure can easily overfit one side of the network while underfitting the other. This makes it difficult to correctly identify the true $(K_1, K_2)$ for both node sets simultaneously, rendering existing cross-validation procedures inadequate for bipartite structures.

	In this work, we develop Bipartite Cross-Validation (BCV), a new procedure that decouples and penalizes asymmetric errors across the two node sets. 
	Within a penalized cross-validation framework, BCV incorporates a tailored complexity penalty and estimation strategy to jointly control underfitting and overfitting across both sides of the network. Our contributions are twofold.
	
	First, we establish the first consistency framework for bipartite model selection with a possibly diverging number of communities $(K_1, K_2)$.
	More specifically, we establish consistency of BCV under explicit scaling conditions and balance the trade-off between network sparsity and the growth in community numbers across both balanced and highly imbalanced regimes.
	
	Secondly, we address the intrinsic asymmetry of bipartite data in terms of the numbers of communities and node sizes, absent in unipartite settings. We introduce a penalty that explicitly accounts for this imbalance and develop a refined analysis that partitions the two-dimensional parameter space into overfitting and underfitting regimes. Notably, our framework employs a stochastic penalty and give a concrete data-driven construction, which is a feature rarely considered in standard penalization methods, allowing the control of community sizes to adapt to the specific sparsity level of the observed network. This ensures that with high probability, the true $(K_1, K_2)$ is the unique global minimizer, regardless of the relative growth rates of the numbers of communities across the two sides. 
	Extensive simulations and real-data analyses, including the ``Southern Women'' dataset and the U.S. Senate cosponsorship network, further demonstrate strong finite-sample performance.

	The rest of the article is organized as follows.
	Section \ref{subsec:model} introduces the bipartite stochastic block model (biSBM).
	Section \ref{subsec:bcvalg} presents the BCV algorithm.
	Section \ref{sec:analysis} provides theoretical results, including growth rates of $(K_1,K_2)$ and the construction of the penalty term.
	Section \ref{sec:simulation} reports simulation results under different growth regimes.
	Section \ref{sec:realdata} applies BCV to two real datasets, yielding new insights beyond existing studies.
	Section \ref{sec:discuss} concludes with a brief discussion and future directions. Proofs of the main theorem and corollaries are deferred to the supplementary material.
	
	\section{Methodology}\label{sec:methodology}
	\subsection{Notation}
	For any positive integer $n$, we denote $[n]=\{1,\dots, n\}$.
	For any matrix $M=[M_{ij}]$, we denote $\|M\|$ as its spectral norm. We denote $I_r$ as the identity matrix of order $r$. For a finite set $\mathcal{A}$, we denote $|\mathcal{A}|$ as its cardinality. Besides, let $\text{I}$ represent the indicator function. 
	For positive number, we write $a\gg b$ or $a=\omega(b)$ as $a/b\to\infty$, and $a\ll b$ or $a=o(b)$ as $a/b\to 0$. Moreover, $a=\Omega(b)$ indicates that there exists a constant $C>0$ such that $a\geq Cb$. Finally, for random variables, $a_n=o_{\mathrm{P}}(b_n)$ if $a_n/b_n\overset{p}{\to}0$.
	
	\subsection{The Bipartite Network Model}\label{subsec:model}
	To model latent community structure in bipartite networks, we adopt the bipartite Stochastic Block Model (biSBM), a probabilistic generative model that captures group-dependent interactions between two distinct node sets. 
	
	Let the two sets have sizes $n_1$ and $n_2$, indexed by $[n_1]$ and $[n_2]$, respectively. Let $A\in\{0,1\}^{n_1\times n_2}$ be the bi-adjacency matrix, where $A_{ij}=1$ if there is an edge between node $i\in[n_1]$ and node $j\in[n_2]$, and 0 otherwise. In the biSBM, we assume that nodes on each side are partitioned into \( K_1 \) and \( K_2 \) latent communities, respectively. Let \( Z_r \in \mathbb{H}^{n_r \times K_r} \) $(r=1,2)$ denote the membership matrix for the \( r \)-th node set, where \( \mathbb{H}^{n_r \times K_r} \) denotes the set of binary indicator matrices with exactly one 1 per row. The corresponding community label vector is denoted by \( c_r \in [K_r]^{n_r} \), with \( (c_r)_i = k \) indicating that node \( i \) belongs to the \( k \)-th community. Each element of $A$ is an independent draw from a Bernoulli distribution, and
	\[
	P := \mathbb{E}[A] = Z_1 B^{(n_1,n_2)} Z_2^T,
	\]
	where $B^{(n_1,n_2)}\in [0,1]^{K_1 \times K_2}$ is the block-wise connection probability matrix. Denote $K=\min\{K_1,K_2\}$, which is the largest possible rank for the mean matrix $P$.
	For notational simplicity, and when no confusion arises, we will omit the superscript in $B^{(n_1,n_2)}$ and refer to the two node sets as side 1 and side 2, respectively, throughout the remainder of this paper. We emphasize that due to the nature of biSBM, both $(n_1,n_2)$ and $(K_1,K_2)$ should be allowed to be at different orders.

	\subsection{The BCV Algorithm}\label{subsec:bcvalg}
	This subsection describes a procedure for selecting the numbers of communities \( K_1 \) and \( K_2 \) on the two sides of a bipartite network, referred to as the Bipartite Cross-Validation (BCV) algorithm.

	Let $0<w_{n_1,n_2}<1$ be the training proportion. Each edge pair $(i,j)$ is independently assigned to the training set with probability $w_{n_1,n_2}$ and to the evaluation set with probability $1-w_{n_1,n_2}$. 
	Since the bi-adjacency matrix $A$ is asymmetric, this edge-splitting scheme differs from those used in symmetric SBM cases (e.g., \cite{li2020network,yang2025pnncv}), where only the upper triangular entries are split due to symmetry. 
	For notational simplicity, we write $w$ for $w_{n_1,n_2}$ throughout the paper. 
	Let $\mathcal{E}\subset [n_1]\times[n_2]$ be the training edge set, and let $\mathcal E^c$ denotes its complement, corresponding to the evaluation set.

	Given a bi-adjacency matrix $A$ and a training edge set $\mathcal E$, let $Y$ be a partially observed bi-adjacency matrix such that $Y_{ij}=A_{ij}$ if $(i,j)\in\mathcal E$ and $Y_{ij}=0$ otherwise. To proceed, we need to recover the bi-adjacency matrix from the partially observed matrix $Y$.
	Leveraging the low-rank structure inherent in the bipartite SBM, where the rank corresponds to $\min\{K_1,K_2\}$, we employ a low-rank matrix approximation method. Specifically, for a pre-specified rank $k$, the complete bi-adjacency matrix can be estimated via a singular value decomposition (SVD) with thresholding:
	\begin{equation} \label{recover}
		\hat{A}^{(k)}=\frac{1}{w}S_H(Y, k),
	\end{equation}
	where $S_H(Y, k)$ denotes the rank-$k$ truncated SVD of $Y$. To elaborate, if the SVD of $Y$ is written as
	$
	Y= UDV^{\top},
	$
	with $ D = \operatorname{diag}(\sigma_1, \ldots, \sigma_{\min\{n_1,n_2\}}) $ and  
	$\sigma_1 \geq \sigma_2 \geq \cdots \geq \sigma_{\min\{n_1,n_2\}} \geq 0$, then  
	$
	S_H(Y, k) = UD_{k}V^{\top},
	$
	where  
	$D_{k} = \operatorname{diag}(\sigma_1, \ldots, \sigma_{k}, 0, \ldots, 0)$ is the diagonal matrix retaining only the top $k$ singular values while setting the remaining singular values to zero.

	For any candidate pair $(K_1',K_2')$, we apply the thresholded SVD with $k=\min\{K_1',K_2'\}$ to obtain the recovered bi-adjacency matrix $\hat{A}^{(K_1',K_2')}$. We then estimate the two community labels $\hat{c}_1^{(K_1',K_2')}$ and $\hat{c}_2^{(K_1',K_2')}$ via spectral clustering. 
	In detail, for side 1, we extract the top-$k$ left singular vectors $\hat{U}^{(K_1',K_2')}$ of $\hat{A}^{(K_1',K_2')}$, and apply $k$-means algorithm with $K_1'$ clusters to obtain $\hat{c}_1^{(K_1',K_2')}$. 
	For side 2, we use the top-$k$ right singular vectors $\hat{V}^{(K_1',K_2')}$ and apply $k$-means with $K_2'$ clusters to obtain $\hat{c}_2^{(K_1',K_2')}$. When no confusion arises, we write the estimated label as $\hat{c}_1$ and $\hat{c}_2$ for simplicity.

	Given the estimated labels, we estimate the block probability matrix by
	$$\hat{B}^{(K_1',K_2')}_{k_1k_2}=\frac{\sum_{(i,j)\in\mathcal{E}}A_{ij}\text{I}\{\hat{c}_{1,i}=k_1,\hat{c}_{2,j}=k_2\}}{|(i,j)\in\mathcal{E}:\hat{c}_{1,i}=k_1,\hat{c}_{2,j}=k_2|},\quad k_1\in[K_1'],k_2\in[K_2'],$$
	and the corresponding estimated connection probability matrix results in the form
	$$\hat{P}^{(K_1',K_2')}_{ij}=\hat{B}^{(K_1',K_2')}_{\hat{c}_{1,i}\hat{c}_{2,j}}.$$

	Finally, we determine the optimal model by minimizing the penalized loss
	$$
	L_{K_1',K_2'}(A,\mathcal{E}^c)=\frac{1}{|\mathcal E^c|}\sum_{(i,j)\in\mathcal{E}^c}(A_{ij}-\hat{P}^{(K_1',K_2')}_{ij})^2+d_{K_1',K_2'}\lambda_{n_1,n_2}.
	$$
	Here, $d_{K_1',K_2'}$ implies the model complexity for the candidate $(K_1',K_2')$, which we prefer to set $d_{K_1',K_2'}=K_1'K_2'$ naturally as it is the number of parameters in the block-wise connection probability matrix. 
	
	The penalty factor $\lambda_{n_1,n_2}$, which is allowed to be a random variable instead of a constant, plays a crucial role in balancing goodness-of-fit and model complexity. In bipartite networks, where overfitting on one side and underfitting on the other may occur simultaneously, this regularization helps in ruling out such incorrect solutions and ensuring consistent model selection under suitable conditions. Specifically, when the overfitting side is too large, the penalty dominates any marginal reduction in prediction error, while conversely when the overfitting side is only mildly inflated, the increase in empirical penalized objective due to underfitting rules out such candidates. The full BCV algorithm is summarized in Algorithm 1 below.
	
	\begin{algorithm}
		\caption{Bipartite Cross-Validation (BCV)}\label{alg:bcv}
		\begin{algorithmic}[1]
			\State \textbf{Input:} Adjacency matrix $A$, training proportion $w$, candidate sets $\mathcal{K}_1$ and $\mathcal{K}_2$, number of replications $S$, complexity parameters $d_{K_1',K_2'}$, and penalty order $\lambda_{n_1,n_2}$.
			
			\For{$s = 1$ \textbf{to} $S$}
			\State Randomly sample a training subset of node pairs $\mathcal{E}_s$ with probability $w$.
			\For{\textbf{each} $(K_1', K_2') \in \mathcal{K}_1 \times \mathcal{K}_2$}
			\State Set $k = \min\{K_1', K_2'\}$.
			\State Compute the rank-$k$ truncated SVD on the adjusted partially observed
			\Statex \hspace{2.7em} matrix $Y/w$ obtain its top $k$ left singular vectors $\hat{U}_s^{(k)}$ and right singular
			\Statex \hspace{2.7em} vectors $\hat{V}_s^{(k)}$. 
			\State Run the $k$-means algorithm on $\hat{U}^{(k)}_s$ with no more than $K_1'$ clusters and
			\Statex \hspace{2.7em} on $\hat{V}^{(k)}_s$ with no more than $K_2'$ clusters to get the estimated label $\hat{c}^{(K_1',K_2')}_{1,s}$
			\Statex \hspace{2.7em} and $\hat{c}^{(K_1',K_2')}_{2,s}$.
			\State Obtain the corresponding estimator $\hat{P}^{(K_1',K_2')}_{s}$.
			\State Evaluate the penalized loss on the test set $\mathcal{E}_s^c$:
			\Statex \hspace{3.7em} $L_{K_1',K_2'}(A,\mathcal{E}_s^c) = \frac{1}{|\mathcal{E}_s^c|} \sum_{(i,j) \in \mathcal{E}_s^c} (A_{ij} - \hat{P}^{(K_1',K_2')}_{s,ij})^2 + d_{K_1',K_2'} \lambda_{n_1,n_2}$.
			\EndFor
			\EndFor
			
			\State Determine the best model: 
			\Statex \hspace{1em} $(\hat{K}_1, \hat{K}_2) = \arg\min_{K_1',K_2'} S^{-1} \sum_{s=1}^S L_{K_1',K_2'}(A,\mathcal{E}_s^c)$.
			
			\State \textbf{Output:} Estimated community numbers $\hat{K}_1$ and $\hat{K}_2$.
		\end{algorithmic}
	\end{algorithm}
	\begin{remark}
		Unlike the results as in \cite{chen2018network} and \cite{li2020network} for unipartite SBM,  since the data analyst typically does not know sensible upper bounds on $K_1$ and $K_2$, we assume that the candidate sets $\mathcal{K}_1$ and $\mathcal{K}_2$ can be unbounded, defaulting to the largest possible choice, $[n_1]$ and $[n_2]$.  
	\end{remark}

	\section{Theoretical Properties of the BCV algorithm}\label{sec:analysis}
	\subsection{The Main Consistency Theorem}\label{subsec:theo}
	
	We first introduce some assumptions on the community structrue and the underlying probability matrix. Recall that $P=\mathbb{E}[A]= Z_1 B^{(n_1,n_2)} Z_2^T$, where $Z_1$ and $Z_2$ are membership matrices and $B^{(n_1,n_2)}$ is the block-wise probability matrix. Denote $n=\max\{n_1,n_2\}$. Without loss of generality, we assume $K_1\leq K_2$, that is, the side 2 has more number of communities.

	\begin{assumption}[Balanced community structure]\label{asm1}
		Let \( n_{rk} = |\{i \in [n_r] : (c_r)_i = k\}| \) denote the size of the \(k\)-th community on side \(r\). There exists a constant \( \pi_0 \in(0,1) \) such that, for all \(r\) and \(k \in [K_r]\), \( n_{rk} \geq \pi_0 n_r/K_r \).
	\end{assumption}

	Assumption \ref{asm1} is analogous to the standard balanced community structure condition in symmetric SBM models \citep{chen2018network, li2020network}.
	
	In theoretical analysis of the bipartite setting, the gap between the number of clusters $K_2$ and the rank $K$ of the underlying mean matrix $P$ brings difficulties in establishing consistency properties for the recovered community label on the second side. Therefore, to address this concern, we adopt the following incoherence condition proposed in \cite{Zhou2019AnalysisOS}.
	Specifically, denote $N_r=\text{diag}(n_{r1},\ldots,n_{rK_r})$, where $n_{rj}$ is the size of the $j$-th community on side $r$. Define 
	$$\bar{B}:=N_1^{1/2}BN_2^{1/2}\qquad \text{ and }\qquad\bar{Z}_r := Z_r N_r^{-1/2},\quad r\in\{1,2\}.$$

	\begin{assumption}[Incoherence condition]\label{asm2}
		Let $V\in\mathbb{O}^{K_2\times K}$ be the matrix of right singular vectors of $\bar{B}$. For a $K_2\times K_2$ matrix $M$ and an index set $\mathcal{I} \subset [K_2 ]$, let $M_{\mathcal{I}}$ be the principal sub-matrix of M on indices $\mathcal{I} \times \mathcal{I}$. Then $V$ satisfies that for some $\beta_{n_1,n_2}\in(0,1]$,
		$$\max_{\mathcal{I}\subset[K_2]:|\mathcal{I}|=2}\|(VV^{\top}-I_{K_2})_{\mathcal{I}}\|\leq1-\beta_{n_1,n_2}.$$
	\end{assumption}

	This assumption ensures sufficient separation among the rows of $\bar{Z}_2 V$, which is the right singular vector matrix of $P$ as shown in \cite{Zhou2019AnalysisOS}. Specifically, it provides a lower bound on pairwise distances between distinct rows of $\bar{Z}_2 V$, enabling consistent recovery of communities on side 2 via spectral clustering with rank $K_2$.

	\begin{remark}\label{remark2}
		Assumption \ref{asm2} is mild in many practical settings, particularly when $K_1$ and $K_2$ are fixed. To see this, suppose $c_1$ and $c_2$ are independently drawn from multinomial distributions $\text{Multinomial}(g_1)$ and $\text{Multinomial}(g_2)$, where $g_r = (g_{r,i})_{i=1}^{K_r}$ satisfies $\sum_i g_{r,i} = 1$. Then, the normalized block-wise probability matrix $\bar{B}/\sqrt{n_1 n_2}$ converges in probability to the limit
		$
		B_{\mathrm{limit}} := \mathrm{diag}(g_1)\, B\, \mathrm{diag}(g_2).
		$
		Consequently, the spectral structure of $\bar{B}$ stabilizes as $n_1, n_2 \to \infty$, and its singular vectors concentrate around those of $B_{\mathrm{limit}}$. 
		This implies that the incoherence parameter $\beta_{n_1,n_2}$ is bounded away from zero with high probability.
		
		However, in high-dimensional regimes where $K_1, K_2 \to \infty$, the behavior of $\beta_{n_1,n_2}$ depends on their relative growth:
		\begin{enumerate}
			\item \textit{Balanced case} ($K_1 \asymp K_2$): if $B$ is well-conditioned, the rows of $V$ remain nearly orthonormal, and $\beta_{n_1,n_2}=\Theta(1)$.
			\item \textit{Asymmetric case} ($K_1 \ll K_2$): $\beta_{n_1,n_2}= O(K_1/K_2)$, since $VV^{\top}$ is a rank $K_1$ projection matrix and its leverage values $(VV^{\top})_{ii}$ is around $K_1/K_2$.
			Thus, by inequality $\|M\|_2\geq M_{ii}$ one can get the rate. This reflects the geometric fact that the singular vectors must spread their mass across a larger number of nodes.
		\end{enumerate}
		In all above scenarios, Assumption \ref{asm4} and Theorem \ref{theo1} explicitly account for this dependency on $\beta_{n_1,n_2}$, as specified in the rate conditions.
	\end{remark}

	\begin{assumption}[Structural conditions on the block probability matrix]\label{asm3}
		Assume $B^{(n_1,n_2)} = \rho_{n_1,n_2} B_0^{(n_1,n_2)}$, where the scaled matrix $B_0^{(n_1,n_2)} \in [0,1]^{K_1 \times K_2}$ satisfies:
		\begin{enumerate}
			\item Normalization: $\max_{i,j} B_{0,ij}^{(n_1,n_2)} = 1$, $\min_{i,j}B_{0,ij}^{(n_1,n_2)}\geq C$ for some constant $C>0$.
			\item Spectral strength: There exists a constant $C_\sigma > 0$ such that the $K$-th singular value satisfies $\sigma_K(B_0^{(n_1,n_2)}) \geq C_\sigma \sqrt{K_2}$.
			\item Separation distance: There exists $d_{n_1,n_2} \in (0,1)$ such that the minimum separation distance between distinct rows and columns satisfies:
			\[
			\min \left\{ \min_{k_1 \neq k_2} \left\| B_{0,k_1 \cdot}^{(n_1,n_2)} - B_{0,k_2 \cdot}^{(n_1,n_2)} \right\|_\infty, \, \min_{l_1 \neq l_2} \left\| B_{0, \cdot l_1}^{(n_1,n_2)} - B_{0, \cdot l_2}^{(n_1,n_2)} \right\|_\infty \right\} \geq d_{n_1,n_2}.
			\]
		\end{enumerate}
	\end{assumption}
	
	\begin{remark}
		Unlike classical SBMs with fixed $K_1$, $K_2$, we allow them to grow with network size, so $B_0$ is defined as a sequence of matrices with increasing dimensions. Condition (ii) ensures that the signal remains identifiable as the parameter space expands. 
		This is justified by the fact that $\|B_0^{(n_1,n_2)}\|_F^2 = \sum_{k=1}^K \sigma_k^2(B_0^{(n_1,n_2)})$. Under the common setting where entries of $B_0$ are of constant order, we have $\|B_0^{(n_1,n_2)}\|_F^2 \asymp K_1 K_2$., then assuming all $K$ singular values are of the same order of magnitude yields $\sigma_K \asymp \sqrt{K_2}$. Finally, Condition (iii) provides the necessary geometric separation to resolve individual clusters, where $d_{n_1,n_2}$ is often treated as a constant in standard regimes.
	\end{remark}

	\begin{assumption}[Scaling conditions for model selection consistency]\label{asm4}
		Assume $d_{n_1,n_2}$ and $w_{n_1,n_2}$ are bounded away from zero by absolute constants. Let $n = \max\{n_1, n_2\}$. For BCV to consistently identify $(K_1, K_2)$, the following scaling limits must hold:
		\begin{enumerate}
			\item Spectral identifiability: $n \rho_{n_1,n_2}=\Omega( (K_2/K) \log n)$.
			\item Underfitting signal dominance: $n_1 n_2 \rho_{n_1,n_2}^2 \cdot \min\{1/\log n, 1/(K_2 \rho_{n_1,n_2})\} \gg K^4 K_2^6.$
			\item Model separation: $\min\{n_1,n_2\}\rho_{n_1,n_2} \beta_{n_1,n_2} \cdot \min\{1 , 1/(K\rho_{n_1,n_2}) \} \gg K^4 K_2^4$.
		\end{enumerate}
	\end{assumption}
	
	\begin{remark}[Intuition for the scaling conditions]
		The conditions in Assumption \ref{asm4} delineate the feasible regime for consistent bipartite community selection. The fully general conditions accommodating vanishing $w_{n_1,n_2}$ and $d_{n_1,n_2}$ are provided in Section A.1 of the Supplementary Material.
		
		Condition 1 ensures the spectral structure is identifiable. Conditions 2 and 3 address underfitting: if  $K'_1 < K_1$, distinct communities must be merged, causing systematic inflation in the cross-validation loss. 
		Condition 2 ensures this inflation dominates estimation noise, as well as this ``merging penalty'' is not masked by the complexity of the bipartite structure on the other node set when no severe overfitting occurs on that side. Condition 3 guarantees that the penalized loss of the correctly specified model remains strictly lower than underfitted alternatives.
	\end{remark}
	
	We now state the consistency result for the BCV estimator.	
	
	\begin{theorem}\label{theo1}
		Suppose Assumptions \ref{asm1}--\ref{asm4} hold and let $d_{K_1',K_2'}=K_1'K_2'$. Assume the penalty factor $\lambda_{n_1,n_2}$ satisfies: 
		\begin{enumerate}
			\setlength{\itemsep}{0.05cm}
			\item $K^2K_2\Pr[\lambda_{n_1,n_2}>\rho_{n_1,n_2}^2/K^3K_2^4]\to 0.$
			\item For any constant $C_1>0$, $$\Pr\left[\lambda_{n_1,n_2}<C_1\frac{\rho_{n_1,n_2} K}{\beta_{n_1,n_2}\nmin}\maxrhon\right]\ll \frac{1}{n^2}$$ \item For any constant $C_2>0$, $$\Pr\left[\lambda_{n_1,n_2}<C_2\frac{KK_2^2\log n}{n_1n_2}\right]\ll \frac{1}{n^2}.$$
		\end{enumerate}
		Then BCV is consistent in the sense that 
		$$
		\Pr\left((\hat{K}_1,\hat{K}_2)=(K_1,K_2)\right)\to1\quad\mathrm{as}\quad (n_1,n_2)\to\infty.
		$$
	\end{theorem}
	
	Here, the penalty factor $\lambda_n$ is allowed to be a random variable, enabling our method to handle unbounded candidate sets $\mathcal{K}_1$ and $\mathcal{K}_2$. Condition 1 controls the upper tail, while Conditions 2-3 impose lower bounds to ensure consistency.

	While Theorem \ref{theo1} provides a general framework for consistency, the roles of network sparsity $\rho_{n_1,n_2}$, size asymmetry $(n_1,n_2)$, and community counts $(K_1, K_2)$ are implicit. To clarify the ``identifiable region'', we now consider two representative regimes: the balanced case ($n_1 \asymp n_2$) and the highly imbalanced case ($n_1 = n_2^a, a > 1$). The following corollary makes the resulting sparsity thresholds and corresponding admissible growth rates of $K$ explicit.
	
	\begin{corollary}\label{cor1}
		Suppose $n_1, n_2 \to \infty$, and $\beta_{n_1,n_2}=\Omega(K/K_2)$. Under the conditions of Theorem \ref{theo1}, the following regimes arise:
		\begin{enumerate}
			\item Balanced case ($n_1 \asymp n_2$): the condition for $\rho_{n_1,n_2}$ says $\rho_{n_1,n_2}=\Omega(\log n/n)$. One the one hand, at the sparsity threshold $\rho_{n_1,n_2}\asymp\log n/n$, BCV is consistent provided $K_1\asymp K_2$ and $K = o((\log n)^{1/10})$. 
			\item Highly imbalanced case ($n_1 \sim n_2^a, a > 1$): the condition for $\rho_{n_1,n_2}$ says $\rho_{n_1,n_2}=\Omega(1/n_2)$. Moreover, under the sparsity threshold $\rho_{n_1,n_2}\asymp1/n_2$, the consistency of BCV requires $K_1$ and $K_2$ to be of constant order.  
		\end{enumerate}
	\end{corollary}
	
	The assumption $\beta_{n_1, n_2}$ agrees with the upper bound discussed in Remark \ref{remark2}. 
	The contrast between the balanced and asymmetric regimes highlights the unique challenges of bipartite model selection. In the balanced case, the sparsity condition for $\rho_{n_1,n_2}$ matches the optimal rate (up to constants) known for unipartite SBMs and allows $(K_1,K_2)$ to grows slowly with $n$, in a logarithmic scale. 
	
	By comparison, existing unipartite results that accommodate a growing number of communities typically require stronger sparsity conditions. For instance, \cite{lei2016goodness} and \cite{hu2021using} assume relatively dense regimes with $\rho$ bounded away from zero, while \cite{wu2024eigengap} requires $\rho \gg n^{-1/3}$. Under a comparable regime $\rho \gg n^{-1/3}$, our result allows $K = O(n^{1/12})$, which is less restrictive than $K=O(n^{1/12-\psi})$ for $\psi>0$ in \cite{wu2024eigengap}.
	
	In the highly imbalanced regime $n_1\sim n_2^a$, Corollary 1 implies
	$$\rho_{n_1,n_2}\gg n_2^{-1}\asymp n_1^{-1/a}.$$
	This requirement can be interpreted as ensuring that the smaller side of the bipartite graph receives sufficient signal strength to distinguish its true communities. Compared with the balanced regime where the sparsity threshold is typically of order \(\log n / n\), the condition here is stricter, reflecting the fact that the side with fewer nodes would otherwise be overwhelmed by the scale of the larger partition. Moreover, due to the new difficulties brought by imbalance data sizes, $(K_1,K_2)$ need to be of constant order to compensate for the information loss on the smaller side.
	
	\subsection{Data-Driven Choice of Penalty Factor}\label{subsec:penalty}
	
	To make use of the theoretical results for practice, we now provide a concrete construction of a data-driven stochastic penalty $\lambda_{n_1, n_2}$.  
	Theorem \ref{theo1} characterizes admissible ranges of $\lambda_{n_1,n_2}$ ensuring model selection consistency, but these depend on unknown quantities such as the sparsity level $\rho_{n_1,n_2}$, the incoherence parameter $\beta_{n_1,n_2}$, and $(K_1,K_2)$, and are therefore not directly implementable. In many applications, researchers can specify conservative upper bounds, $(\bar{K}_1, \bar{K}_2)$, which represent the maximum latent resolution of interest. In this subsection, we focus on the case that such upper bounds $\bar{K}_1$ and $\bar{K}_2$ are relatively small compared to $n_1$ and $n_2$ such that they satisfies the scaling conditions in Assumption \ref{asm4}, when choosing the data-driven $\lambda_{n_1, n_2}$.
	
	In this case, the admissible range of $\lambda_{n_1,n_2}$ can be expressed in terms of the observable quantities up to constant factors. 
	Moreover, as discussed in Remark \ref{remark2}, we shall assume that $\bar{K}_1$ and $\bar{K}_2$ are of the same growth order, so that $\beta_{n_1,n_2}$ is of constant order. 
	
	Compared with the unipartite setting, an additional difficulty arises in the bipartite case, as the number of candidate models now grows over a two–dimensional grid $(K_1,K_2)$ rather than a one–dimensional sequence. 
	As a result, a substantially larger proportion of the candidate models are high–complexity models with large values of $K_1K_2$.
	This amplifies sensitivity to the choice of the penalty level: overly large penalties induce systematic underfitting, while insufficient penalties fail to control model complexity.
	Therefore, accurate calibrating of $\lambda_{n_1,n_2}$ is particularly crucial in the bipartite case. 
	
	This leads to a practical guideline for choosing the penalty level.
	
	\begin{corollary}
		Assume that the true community numbers satisfy $K_1 \leq \bar K_1$ and $K_2 \leq \bar K_2$, where the known upper bounds $\bar K_1$ and $\bar K_2$ may grow with the network size and are of the same order. 
		Suppose further that Assumptions \ref{asm1}--\ref{asm3} hold, the scaling conditions in Assumption \ref{asm4} are satisfied with $(K_1,K_2)$ replaced by $(\bar K_1,\bar K_2)$, and that $\beta_{n_1,n_2}=\Theta(1)$. 
		Without loss of generality, assume $\bar K_1 \leq \bar K_2$. Let $\hat{\rho}_{n_1,n_2}=\sum_{i,j}A_{ij}/n_1n_2$. Then, the choice
		\begin{equation}
			\lambda_{n_1, n_2}=C\frac{\hat{\rho}_{n_1,n_2}}{\bar{K}_1^{3/2}\bar{K}_2^2}\sqrt{\max\left\{\frac{\hat{\rho}_{n_1,n_2}\bar{K}_1}{\nmin}\max\{1,\bar{K}_1\hat{\rho}_{n_1,n_2}\},\frac{\bar{K}_1\bar{K}_2^2\log n}{n_1n_2}\right\}}\label{eqn:penalty}
		\end{equation}
		satisfies the penalty conditions in Theorem \ref{theo1}. In particular, if the two upper bounds $\bar{K}_1$ and $\bar{K}_2$ are bounded, then
		\begin{equation*}
			\lambda_{n_1, n_2}=C\frac{\hat{\rho}_{n_1,n_2}^{3/2}}{\sqrt{\nmin}}
		\end{equation*}
		is a proper choice of the penalty factor for consistent model selection.
	\end{corollary}
	
	The choice of $\lambda_{n_1, n_2}$ balances the upper and lower bounds in Theorem \ref{theo1}. Specifically, the proposed form corresponds to the geometric mean of the upper bound and the leading rate of the lower bound, ensuring sufficient penalization of overfitting while preserving the true signal.
	
	While this corollary provides an implementable rule under a pre-specified bound $(\bar{K}_1, \bar{K}_2)$, extending fully data-driven calibration to the general growing-$K$ regime remains a non-trivial challenge. 
	In that setting, admissible range of $\lambda_{n_1, n_2}$ depends on latent community sizes and incoherence, suggesting that a fully adaptive tuning rule might necessitate iterative pilot estimations or multi-stage calibration procedures. We leave the development of such adaptive schemes for future work.
	
	\section{Simulation Studies}\label{sec:simulation}

	As discussed before, the relative growth of the two node sets affects the sparsity required for consistent model selection. In practice, bipartite network datasets exhibit diverse scaling patterns. For instance, the well-known Southern women dataset \citep{davis1941deep} has comparable sizes on both sides, whereas the Cosponsorship dataset studied in \cite{lo2025senate} (99 senators and 2631 bills) is highly imbalanced. These differences in scale motivate an evaluation of our method across varying regimes of $(n_1,n_2)$.
	In this section, we mainly focus on two scenarios: (i) $n_1$ and $n_2$ are of the same order, and (ii) $n_2$ grows polynomially with $n_1$. The training proportion is set to $w = 0.9$, and BCV is implemented via a single 10-fold split. For the penalty, empirical performance is stable with $C=0.01$, and we fix $\lambda_{n_1, n_2}=0.01\hat{\rho}_{n_1,n_2}^{3/2}/\sqrt{\nmin}$ in all numerical experiments.

	\subsection{Case I: Balanced Growth}\label{subsec:balance}
	
	In this subsection, we study the balanced growth regime and evaluate BCV method under varying network structures and community configurations. Specifically, we consider three settings with different community sizes and imbalance levels. 
	In all case, the community-wise edge probability matrix is set to $B = r B_0$ with $r \in \{0.05, 0.1\}$, and the numbers of nodes on the two sides are $n_l = K_l n_0$, where $n_0 \in \{100,200,300,400,500\}$ under the balanced setting and $n_0 \in \{200,300,400,500,600\}$ under the imbalanced setting.	
	\vspace{5pt}\\
	\textit{Setting 1:} $(K_1,K_2)=(3,3)$. The diagonal entries of $B_0$ are independently drawn from uniform distribution $\mathcal{U}(0.7,1)$, and the off-diagonal entries from $\mathcal{U}(0.1,0.3)$. Community labels $c_1$ and $c_2$ are generated from multinomial distributions $\mathcal{M}(n, \Pi_1)$ and $\mathcal{M}(n, \Pi_2)$, with (i) balanced proportions $\Pi_1=\Pi_2=(1/3,1/3,1/3)$, and (ii) imbalanced proportions $\Pi_1=\Pi_2=(1/6,1/3,1/2)$.\vspace{5pt}\\
	\textit{Setting 2:} $(K_1,K_2)=(3,4)$. The construction of $B_0$ and the sampling mechanism are identical to Setting~1, except that (i) for the balanced case, $\Pi_1=(1/K_1,\dots,1/K_1)$ and $\Pi_2=(1/K_2,\dots,1/K_2)$; (ii) for the imbalanced case, $\Pi_1=(1/2,1/3,1/6)$ and $\Pi_2=(3/8,1/4,1/4,1/8)$.\vspace{5pt}\\
	\textit{Setting 3:} $(K_1,K_2)=(10,14)$. This setting examines performance in larger models. The diagonal entries of $B_0$ are set to 1; the first three entries of the 12th column, the 4th--6th of the 13th column, and the last four of the 14th column are set to~0.75; and the remaining entries to~0.25. For the balanced case, $\Pi_1=(1/K_1,\dots,1/K_1)$ and $\Pi_2=(1/K_2,\dots,1/K_2)$; for the imbalanced case, 
	$\Pi_1=(1/15,1/15,1/15,1/15,1/10,1/10,2/15,2/15,2/15,2/15)$ and \\
	$\Pi_2=(1/21,1/21,1/21,1/21,1/21,1/14,1/14,1/14,1/14,2/21,2/21,2/21,2/21,2/21)$.\vspace{5pt}

	To alleviate the computational burden of exhaustive search over all candidate pairs $(K_1,K_2)$, we adopt an adaptive strategy. Starting from $(1,1)$, we expand the candidate set to include all pairs satisfying $K_1 \leq k$ and $K_2 \leq k$ at step $k$, and evaluate the corresponding penalized loss. 
	The search stops if no improvement is observed for a pre-specified number of consecutive steps (which is called the patience parameter).
	This adaptive procedure greatly reduces the number of models evaluated while ensuring that the true model $(K_1,K_2)$ is visited with high probability, when the patience parameter is sufficiently large (Theorem \ref{theo1}). 
	Thus, this strategy offers an efficient and reliable alternative to full two-dimensional grid search. In practice, a patience value of 3 performs well.

	We compare BCV with two well-established alternatives.
	The first is the bimodularity method proposed by \citet{barber2007bimod}, which is the earliest community detection approach specifically designed for bipartite networks, jointly clustering both node sets.
	For a fair comparison under our setting, we separate the nodes into the two corresponding parts and count the number of communities detected in each part.
	The second is the projection-based method \citep{Alzahrani2016}, which projects the bipartite graph onto each side to obtain two unipartite graphs and then applies modularity-based clustering.
	We do not include Bayesian approaches such as \citet{larremore2020bayes} due to their computational cost at our data scale. 
	All results are averaged over 100 repetitions and reported in Table \ref{simulation:balanced_growth}. 
	Each entry reports the recovery rates on the left and right node sets, respectively.
	
	\begin{table}[htbp]
		\centering\renewcommand\arraystretch{1}{\scriptsize
			\caption{Simulation result under balanced growth.}\label{simulation:balanced_growth}
			\setlength{\tabcolsep}{5pt}{
				\begin{tabular}{ccccccc}
					\hline
					\cline{1-7}
					
					$(K_1,K_2)$ & Balance & $r$ & $n_0$ & BCV & Projection & Bimodularity\\
					\hline
					(3,3) & Balanced & 0.05 & 100 & (0,0) & (0,0) & (0,0)\\
					& & & 200 & (0.80,0.80) & (0.94,0.94) & (0.31,0.31)\\
					& & & 300 & (1.00,1.00) & (0.99,0.99) & (0.94,0.94)\\
					& & & 400 & (1.00,1.00) & (1.00,1.00) & (0.98,0.98)\\
					& & & 500 & (1.00,1.00) & (1.00,1.00) & (1.00,1,00)\\
					& & 0.1 & 100 & (0.79,0.79) & (0.94,0.96) & (0.42,0.42)\\
					& & & 200 & (1.00,1.00) & (1.00,1.00) & (1.00,1.00)\\
					& & & 300 & (1.00,1.00) & (1.00,1.00) & (1.00,1.00)\\
					& & & 400 & (1.00,1.00) & (1.00,1.00) & (1.00,1.00)\\
					& & & 500 & (1.00,1.00) & (1.00,1.00) & (1.00,1.00)\\
					& Imbalanced & 0.05 & 200 & (0.02,0.01) & (0.43,0.39) & (0,0)\\
					& & & 300 & (0.26,0.13) & (0.70,0.63) & (0.09,0.09)\\
					& & & 400 & (0.45,0.30) & (0.63,0.68) & (0.52,0.52)\\
					& & & 500 & (0.62,0.51) & (0.53,0.56) & (0.88,0.88)\\
					& & & 600 & (0.63,0.69) & (0.54,0.55) & (0.97,0.97)\\
					& & 0.1 & 200 & (0.44,0.40) & (0.58,0.58) & (0.61,0.61)\\
					& & & 300 & (0.68,0.68) & (0.53,0.50) & (0.95,0.95)\\
					& & & 400 & (0.87,0.82) & (0.50,0.53) & (0.99,0.99)\\
					& & & 500 & (0.92,0.92) & (0.52,0.54) & (0.99,0.99)\\
					& & & 600 & (0.99,0.99) & (0.52,0.52) & (0.99,0.99)\\
					\hline
					(3,4) & Balanced & 0.05 & 100 & (0,0) & (0.02,0) & (0,0)\\
					& & & 200 & (0.80,0.16) & (0.93,0.30) & (0.11,0.31)\\
					& & & 300 & (1.00,0.68) & (1.00,0.08) & (0.49,0.32)\\
					& & & 400 & (1.00,0.90) & (1.00,0.02) & (0.69,0.25)\\
					& & & 500 & (1.00,0.98) & (1.00,0) & (0.82,0.14)\\
					& & 0.1 & 100 & (0.81,0.26) & (0.96,0.28) & (0.20,0,49)\\
					& & & 200 & (1.00,0.85) & (1.00,0.02) & (0.78,0.19)\\
					& & & 300 & (1.00,0.99) & (1.00,0) & (0.96,0.04)\\
					& & & 400 & (1.00,1.00) & (0.99,0)  & (1.00,0)\\
					& & & 500 & (1.00,1.00) & (1.00,0) & (0.99,0.01)\\
					& Imbalanced & 0.05 & 200 & (0.08,0.10) & (0.63,0.25) & (0,0.16)\\
					& & & 300 & (0.49,0.06) & (0.76,0.04) & (0.20,0.51)\\
					& & & 400 & (0.96,0.42) & (0.66,0.01) & (0.59,0.35)\\
					& & & 500 & (0.93,0.64) & (0.69,0) & (0.80,0.13)\\
					& & & 600 & (0.99,0.79) & (0.68,0) & (0.94,0.05)\\
					& & 0.1 & 200 & (0.84,0.38) & (0.69,0) &  (0.80,0.17)\\
					& & & 300 & (0.99,0.78) & (0.65,0) & (0.98,0.01)\\
					& & & 400 & (1.00,0.92) & (0.68,0) & (0.98,0.01)\\
					& & & 500 & (1.00,0.99) & (0.67,0) & (0.99,0)\\
					& & & 600 & (1.00,1.00) & (0.69,0) & (0.99,0)\\
					\hline
					(10,14) & Balanced & 0.05 & 100 & (0,0) & (0.01,0) & (0,0)\\
					& & & 200 & (0,0) & (0,0) & (0,0)\\
					& & & 300 & (0,0) & (0,0) & (0,0)\\
					& & & 400 & (0.88,0.88) & (0,0) & (0,0)\\
					& & & 500 & (0.99,0.99) & (0,0) & (0,0)\\
					& & 0.1 & 100 & (0,0) & (0,0) & (0,0)\\
					& & & 200 & (0.89,0.89) & (0,0) & (0,0)\\
					& & & 300 & (1.00,1.00) & (0,0) & (0,0)\\
					& & & 400 & (1.00,1.00) & (0,0) & (0,0)\\
					& & & 500 & (1.00,1.00) & (0,0) & (0,0)\\
					& Imbalanced & 0.05 & 200 & (0,0) & (0,0) & (0,0)\\
					& & & 300 & (0,0) & (0,0) & (0,0)\\
					& & & 400 & (0,0) & (0,0) & (0,0)\\
					& & & 500 & (0,0) & (0,0) & (0,0)\\
					& & & 600 & (0,0) & (0,0) & (0,0)\\
					& & 0.1 & 200 & (0,0) & (0,0) & (0,0)\\
					& & & 300 & (0,0) & (0,0) & (0,0)\\
					& & & 400 & (0.42,0.42) & (0,0) & (0,0)\\
					& & & 500 & (1.00,0.99) & (0,0) & (0,0)\\
					& & & 600 & (1.00,1.00) & (0,0) & (0,0)\\
					\hline
					\cline{1-7}	
				\end{tabular}
		}}
	\end{table}
	Table \ref{simulation:balanced_growth} shows that the proposed BCV method consistently outperforms the competing methods across all scenarios.
	In Setting 1, all three methods converge under balanced conditions, whereas the projection method fails in the imbalanced case, possibly due to information loss from projection.
	In Setting 2, despite asymmetric community numbers, our BCV method remains stable. Both the projection and bimodularity methods recover the smaller-side community structure reasonably well in the balanced case, but perform poorly on the other side.
	In the more challenging Setting 3, the BCV method maintains strong performance except under sparse and imbalanced conditions. 
	By contrast, both alternatives consistently fail to identify the correct numbers of communities.
	
	\subsection{Case II: Polynomial growth}\label{subsec:poly}
	
	To empirically validate the discussion in Section \ref{sec:analysis}, we evaluate the BCV method in a polynomially imbalanced regime.
	We focus on the representative case $n_{2}=n_{1}^{3/2}$, which introduces moderate imbalance. 
	We consider two size regimes: $n_{1}\in\{100,200,300,400\}$ under balanced community proportions, and $n_{1}\in\{200,400,600,800\}$ otherwise. Edges are generated from a bipartite SBM with $B = rB_{0}$, where $r$ controls the sparsity level. 
	Guided by theory, we adopt lower sparsity levels than in the balanced-growth setting to ensure sufficient signal strength. 
	We compare BCV with the projection and bimodularity methods under two settings.
	Results are averaged over 100 repetitions and reported in Table~\ref{simulation:polynomial_growth}.\vspace{5pt}\\

	\noindent\textit{Setting 1:} We adopt the same configuration of $(K_{1},K_{2})$, $B_{0}$, and community proportions as in Setting 1 of Section~\ref{subsec:balance}, but restrict $r \in \{0.1,0.2\}$ to reflect the stricter sparsity requirements under imbalance, as discussed in the preceding analysis.\vspace{5pt}\\
	\textit{Setting 2:} We consider a more heterogeneous case with $(K_{1},K_{2})=(3,6)$ and
	\[
	B_0 =
	\begin{bmatrix}
		1 & 0.25 & 0.25 & 0.75 & 0.75 & 0.25 \\
		0.25 & 1 & 0.25 & 0.75 & 0.25 & 0.75 \\
		0.25 & 0.25 & 1 & 0.25 & 0.75 & 0.75 
	\end{bmatrix}.
	\]
	Community labels $c_{1}$ and $c_{2}$ are drawn from multinomial distributions
	$\mathcal{M}(n_{1},\Pi_{1})$ and $\mathcal{M}(n_{2},\Pi_{2})$, with balanced proportions 
	$\Pi_{1}=(1/3,1/3,1/3)$, $\Pi_{2}=(1/6,\dots,1/6)$ and imbalanced proportions 
	$\Pi_{1}=(1/2,1/3,1/6)$, $\Pi_{2}=(1/4,1/4,1/6,1/6,1/12,1/12)$. 
	We again use $r \in \{0.1,0.2\}$.
	
	\begin{table}[htbp]
		\centering\renewcommand\arraystretch{1}{\scriptsize
			\caption{Simulation result under polynomial growth.}\label{simulation:polynomial_growth}
			\setlength{\tabcolsep}{5pt}{
				\begin{tabular}{ccccccc}
					\hline
					\cline{1-7}
					
					$(K_1,K_2)$ & Balance & $r$ & $n_1$ & BCV & Projection & Bimodularity\\
					\hline
					(3,3) & Balanced & 0.1 & 100 & (0.46,0.45) &(0.99,0) &(0.05,0.05) \\
					& & & 200 &(1.00,1.00) &(1.00,0.46) &(1.00,1.00) \\
					& & & 300 &(1.00,1.00) &(1.00,0.97) &(1.00,1.00) \\
					& & & 400 &(1.00,1.00) &(1.00,0.98) &(1.00,1.00) \\
					& & 0.2 & 100 &(1.00,1.00) &(1.00,0.77) &(1.00,1.00) \\
					& & & 200 &(1.00,1.00) &(1.00,1.00) &(1.00,1.00) \\
					& & & 300 &(1.00,1.00) &(1.00,1.00) &(1.00,1.00) \\
					& & & 400 &(1.00,1.00) &(1.00,1.00) &(1.00,1.00) \\
					& Unbalanced & 0.1 & 200 &(0.88,0.16) &(0.57,0.31) &(0.32,0.32) \\
					& & & 400 &(0.98,0.86) &(0.50,0.69) &(0.99,0.99) \\
					& & & 600 &(1.00,1.00) &(0.50,0.57) &(0.99,0.99) \\
					& & & 800 &(1.00,1.00) &(0.49,0.53) &(0.99,0.99) \\
					& & 0.2 & 200 &(0.95,0.79) &(0.55,0.60) &(0.99,0.99) \\
					& & & 400 &(1.00,1.00) &(0.51,0.53) &(0.99,0.99) \\
					& & & 600 &(1.00,1.00) &(0.50,0.51) &(0.99,0.99) \\
					& & & 800 &(1.00,1.00) &(0.49,0.49) &(0.99,0.99) \\
					\hline
					(3,6) & Balanced & 0.1 & 100 &(0,0) &(1.00,0) &(0,0) \\
					& & & 200 &(1.00,0) &(1.00,0) &(0.97,0) \\
					& & & 300 &(1.00,0.53) &(1.00,0) &(1.00,0) \\
					& & & 400 &(1.00,1.00) &(1.00,0) &(1.00,0) \\
					& & 0.2 & 100 &(1.00,0.02) &(1.00,0) &(0.97,0) \\
					& & & 200 &(1.00,0.99) &(1.00,0) &(1.00,0) \\
					& & & 300 &(1.00,1.00) &(1.00,0) &(1.00,0) \\
					& & & 400 &(1.00,1.00) &(1.00,0) &(1.00,0) \\
					& Unbalanced & 0.1 & 200 &(0.85,0) &(0.18,0) & (0.09,0)\\
					& & & 400 &(1.00,0.03) &(0.16,0) &(1.00,0) \\
					& & & 600 &(1.00,0.44) &(0.08,0) &(1.00,0) \\
					& & & 800 &(1.00,1.00) &(0.05,0) &(1.00,0) \\
					& & 0.2 & 200 &(0.99,0.05) &(0.18,0) &(1.00,0) \\
					& & & 400 &(1.00,1.00) &(0.17,0) &(1.00,0) \\
					& & & 600 &(1.00,1.00) &(0.08,0) &(0.99,0) \\
					& & & 800 &(1.00,1.00) &(0.06,0) &(0.97,0) \\
					\hline
					\cline{1-7}
				\end{tabular}
		}}
	\end{table}
	
	Table~\ref{simulation:polynomial_growth} shows that BCV remains robust under polynomially imbalanced. In the baseline configuration $(K_{1},K_{2})=(3,3)$, BCV achieves accurate recovery under both balanced and imbalanced community proportions, with performance improving steadily as $n_1$ increases. 
	In the more challenging setting $(K_{1},K_{2})=(3,6)$, BCV continues to improve with increasing sample size, even under imbalance.
	By contrast, the Projection method degrades markedly in the imbalanced regime and is sensitive to sparsity. 
	The Bimodularity method performs reasonably well only when both sides have equal community sizes, but deteriorates substantially in the heterogeneous case.
	
	These findings align with our theoretical analysis: imbalance weakens effective signal strength, yet BCV remains stable when sparsity is properly controlled.
	
	\section{Real Data Analysis}\label{sec:realdata}
	In this section, we apply BCV to the ``Southern women'' network \citep{davis1941deep} and the U.S. senate cosponsorship network \citep{fowler2007cosponsor,lo2025senate} to assess empirical performance.

	\subsection{``Southern Women'' Network}
	The ``Southern women'' dataset is a standard benchmark for bipartite community detection \citep{Alzahrani2016}. It consists of 18 women and 14 social events, with edges indicating attendance at a particular event.
	Various clustering results have been reported in the literature. For example, using modularity-based approaches, \citet{barber2007bimod} and \citet{liu2010community} identify four modules, while \citet{Guimera2007module} finds two. From a Bayesian perspective, \citet{larremore2020bayes} suggests a single community per side.

	We apply BCV with tuning parameters selected as in Section~\ref{subsec:balance}. Figure~\ref{fig:south_women_mse} reports the penalized CV loss (also referred to MSE) over candidate models. Specifically, panel (a) shows the full two-dimensional search, and panel (b) fixes two communities on the women side. The penalized loss heatmap indicates a clear preference for two clusters among women, whereas the event side attains its minimum at 3, but in a less pronounced fashion.
	
	Figure \ref{fig:south_women}(a) visualizes the detected structure from spectral clustering with the numbers of clusters selected by our method, where triangles denote women, circles represent events, and colors indicate community membership on each side. For comparison, Figure \ref{fig:south_women}(b) displays the result of \citet{barber2007bimod}, in which black nodes denote events and white nodes denote women, and shapes indicate different modules.

	\begin{figure}[htbp]
		\centering
		\subfloat[Heatmap of penalized CV loss.]{
			\includegraphics[width =7.7cm]{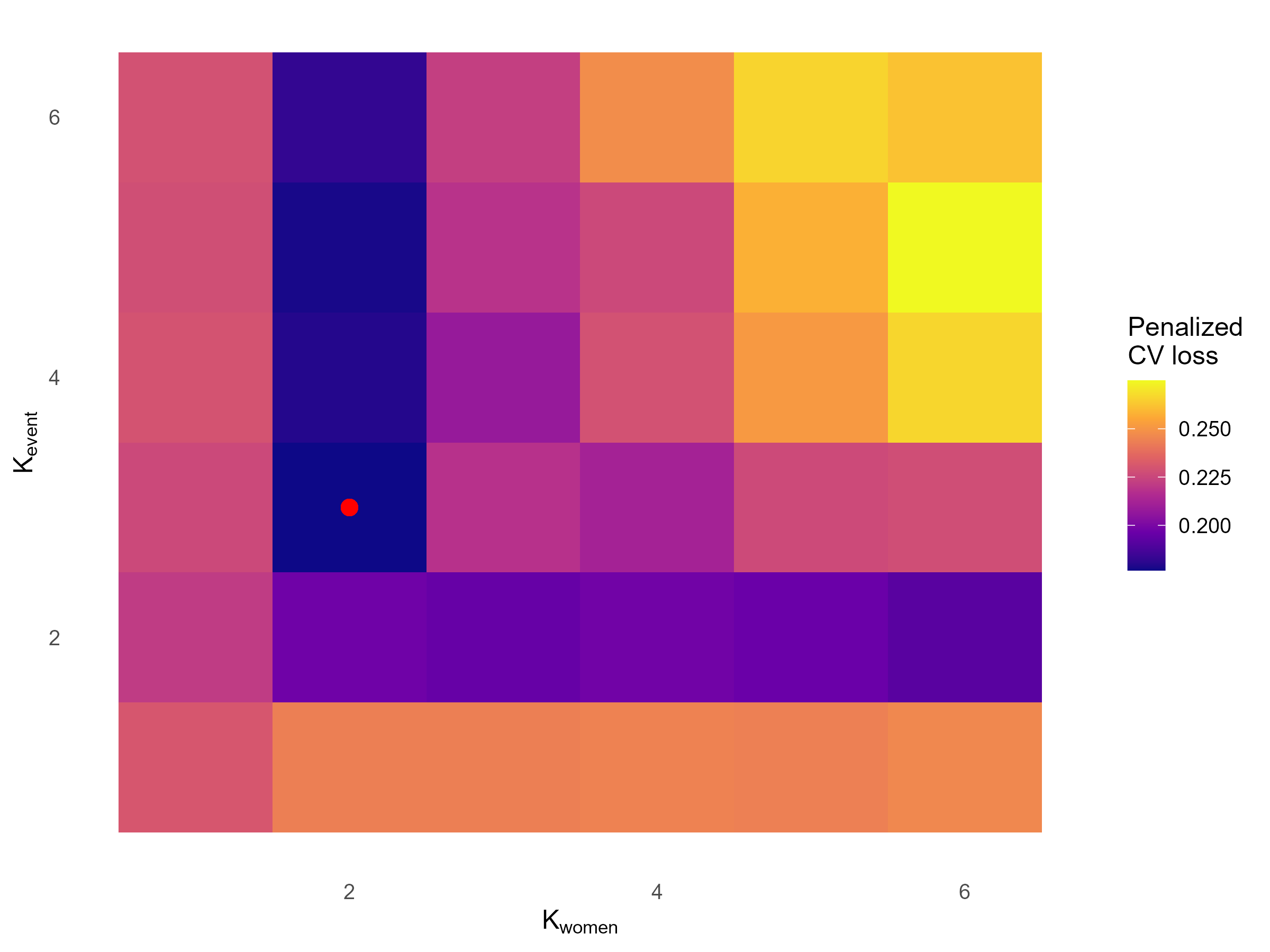}}
		\subfloat[Penalized CV loss on the events side conditioned on $K_{\mathrm{women}}=2$.]{
			\includegraphics[width =7.3cm]{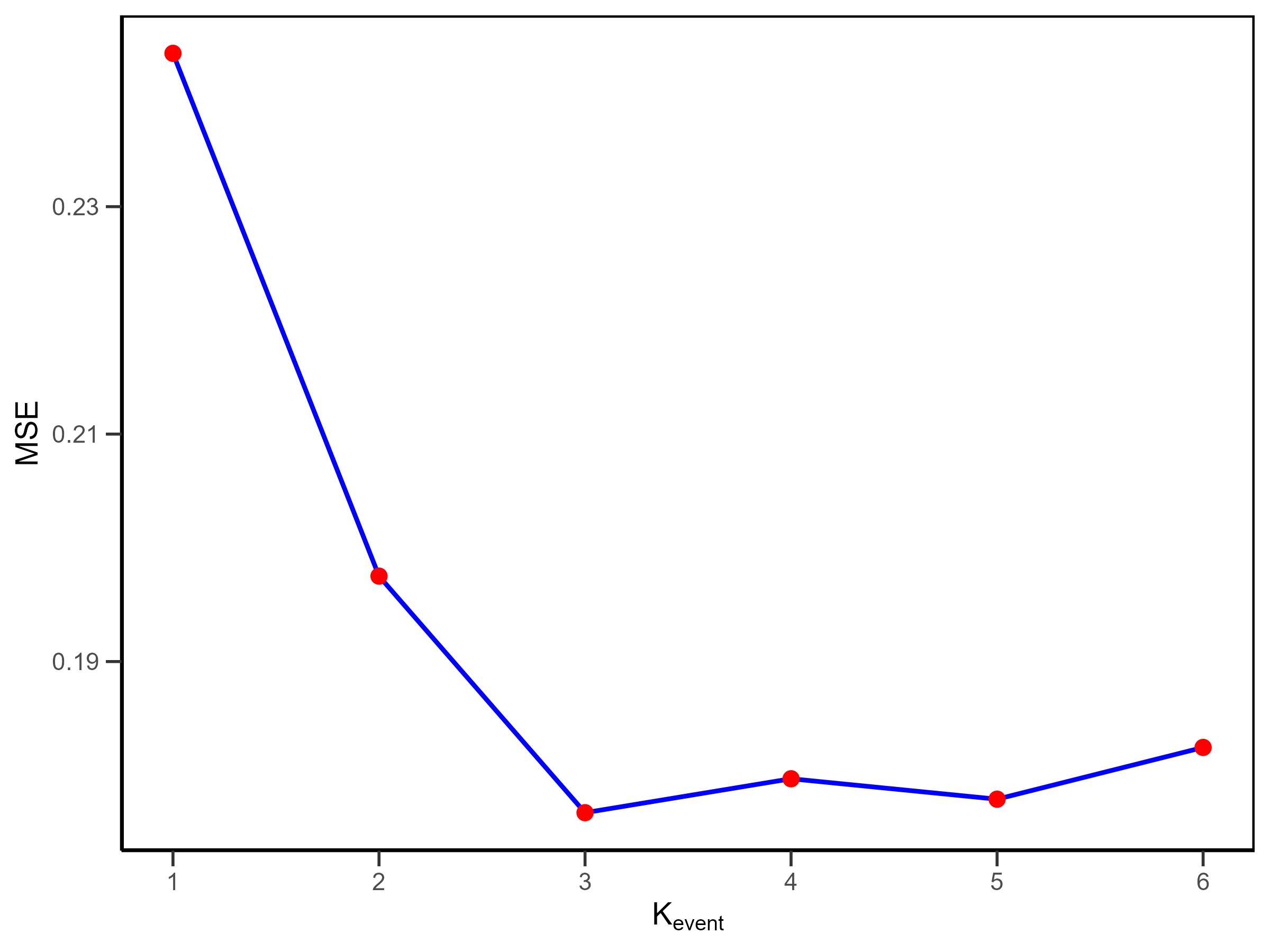}}\\
		\centering
		\caption{Results of Southern Women network by 10-fold BCV algorithm.}\label{fig:south_women_mse}
	\end{figure}

	\begin{figure}[htbp]
		\centering
		\subfloat[Clustering result of BCV.]{
			\includegraphics[width =7.5cm]{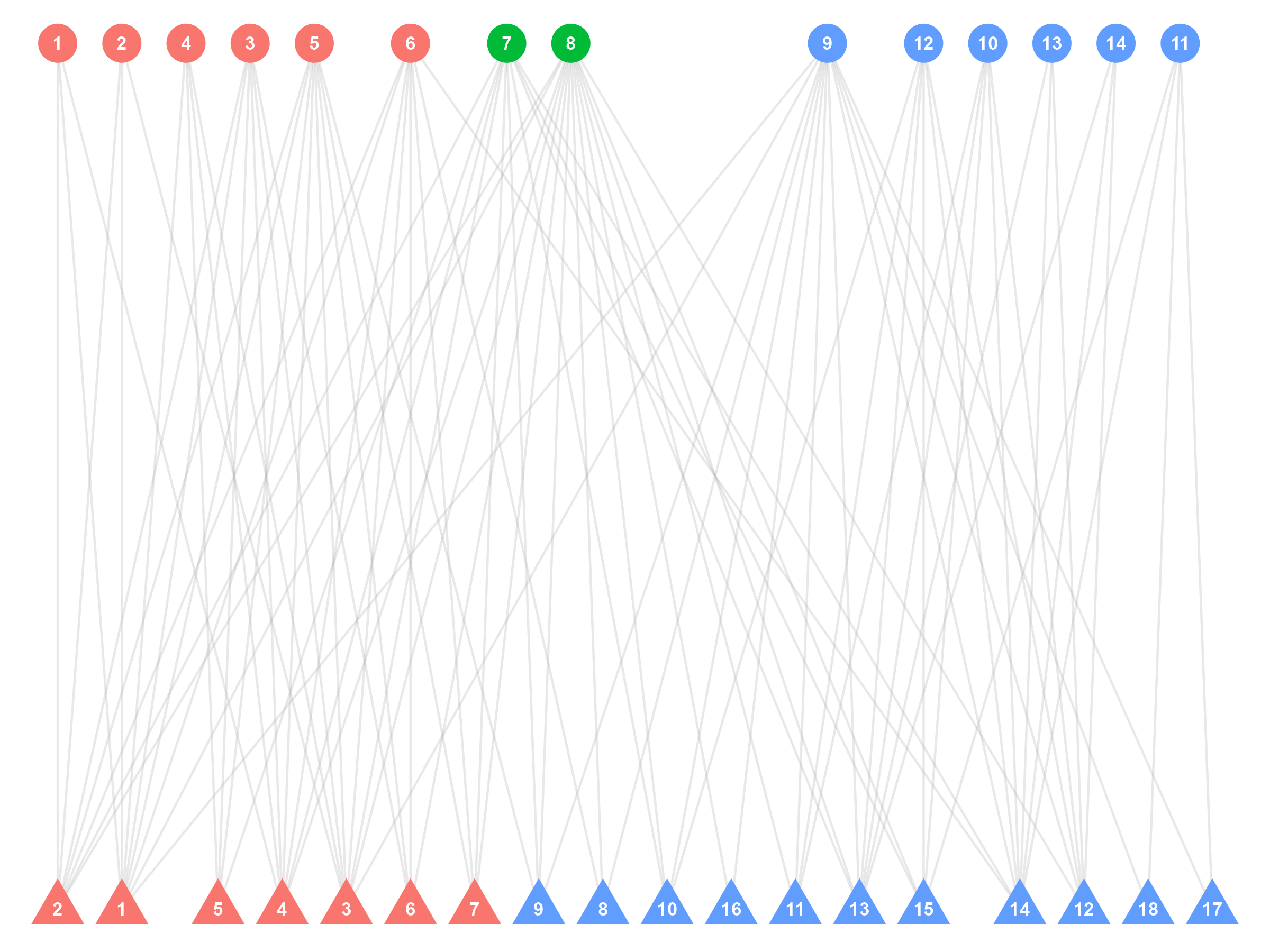}}\label{fig:south_women_1}
		\subfloat[Clustering result of bimodularity method.]{
			\includegraphics[width =7.5cm]{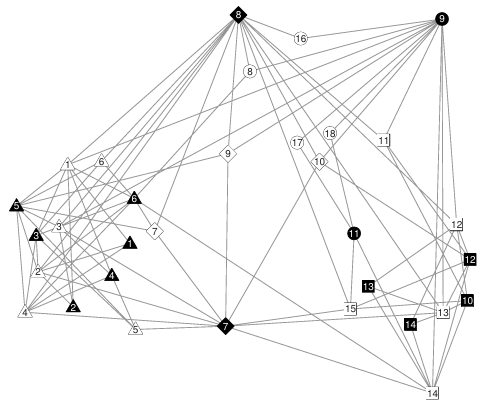}}\label{fig:south_women_2}\\
		\centering
		\caption{Estimated community structure for Southern women network.}\label{fig:south_women}
	\end{figure}
	
	In Figure~\ref{fig:south_women}(a), each community exhibits distinct structure. Most red circles (events) connect primarily to red triangles (women), and most blue events to blue women, while two green events link both groups of women. 
	The two-cluster division among women closely matches the ethnographic findings of \citet{davis1941deep}, and the three-way split among events—where one group bridges the two women communities—is consistent with the observation in \citet{DOREIAN200429}. 
	These ``bridging events'' are sociologically meaningful, as they connect otherwise separate circles. Unlike modularity-based methods (e.g., Figure~\ref{fig:south_women}(b)), which tend to absorb such events into a major module, our bipartite treatment highlights them as a distinct group, revealing their integrative role within the social structure.

	\subsection{Cosponsorship Network}
	We study the cosponsorship network between legislators and bills in the 107th U.S. Senate. The data were originally compiled by \citet{fowler2007cosponsor} and later processed and analyzed by \citet{lo2025senate}. This widely used dataset is publicly available on the Harvard Dataverse and CodeOcean. We use the CodeOcean version, which contains both dyadic sponsorship data and committee-level bill annotations.
	After removing the single senator with no (co)sponsorship activity, the resulting bipartite network consists of $99$ senators and $2631$ bills, with edges indicating sponsorship or cosponsorship. Party affiliation (Democrat or Republican) is available for all senators and serves as a natural external benchmark on the legislator side.

	Applying BCV with the same tuning procedure as in previous experiments, we jointly select the numbers of communities on both sides of the bipartite network. The resulting penalized loss surface over the grid of $(K_{\text{sen}}, K_{\text{bill}})$ is visualized in Figure~\ref{fig:cosponsorship_mse}(a) as a heatmap.

	On the senator side, the loss surface exhibits a relatively clear minimum at $K_{\text{sen}} = 2$. In particular, when compared across different values of $K_{\text{sen}}$, the row corresponding to $K_{\text{sen}} = 2$ consistently attains lower loss values than neighboring choices such as $K_{\text{sen}} = 3$, indicating a well-separated choice.
	In contrast, conditional on $K_{\text{sen}} = 2$, the loss surface along the bill dimension is relatively flat for $K_{\text{bill}} \ge 10$. 
	Although the minimum occurs at $K_{\text{bill}} = 13$, nearby values yield rather similar losses, suggesting a less decisive choice. 
	To further illustrate this behavior, we additionally plot the loss curve given $K_{\text{sen}} = 2$, which confirms that the improvement beyond $K_{\text{bill}} \approx 10$ is gradual rather than abrupt. This weaker separation aligns with prior findings that factors such as ideology, seniority, region, and gender shape cosponsorship behavior \citep{lo2025senate}, leading to a more complex network structure.
	\begin{figure}[htbp]
		\centering
		\subfloat[Heatmap of penalized CV loss.]{
			\includegraphics[width =7.7cm]{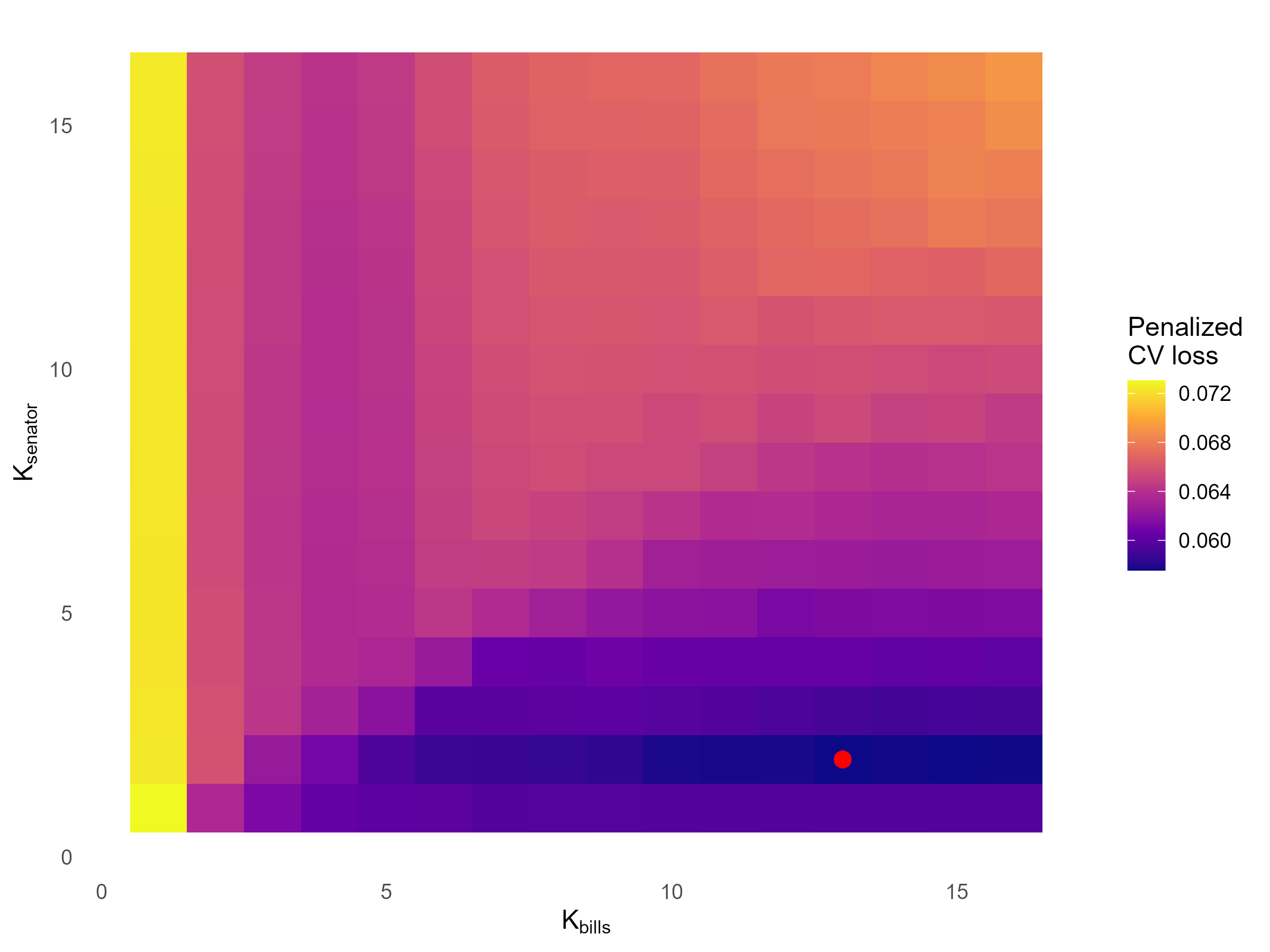}}
		\subfloat[Penalized CV loss on the bills side conditioned on $K_{\mathrm{senator}}=2$.]{
			\includegraphics[width =7.3cm]{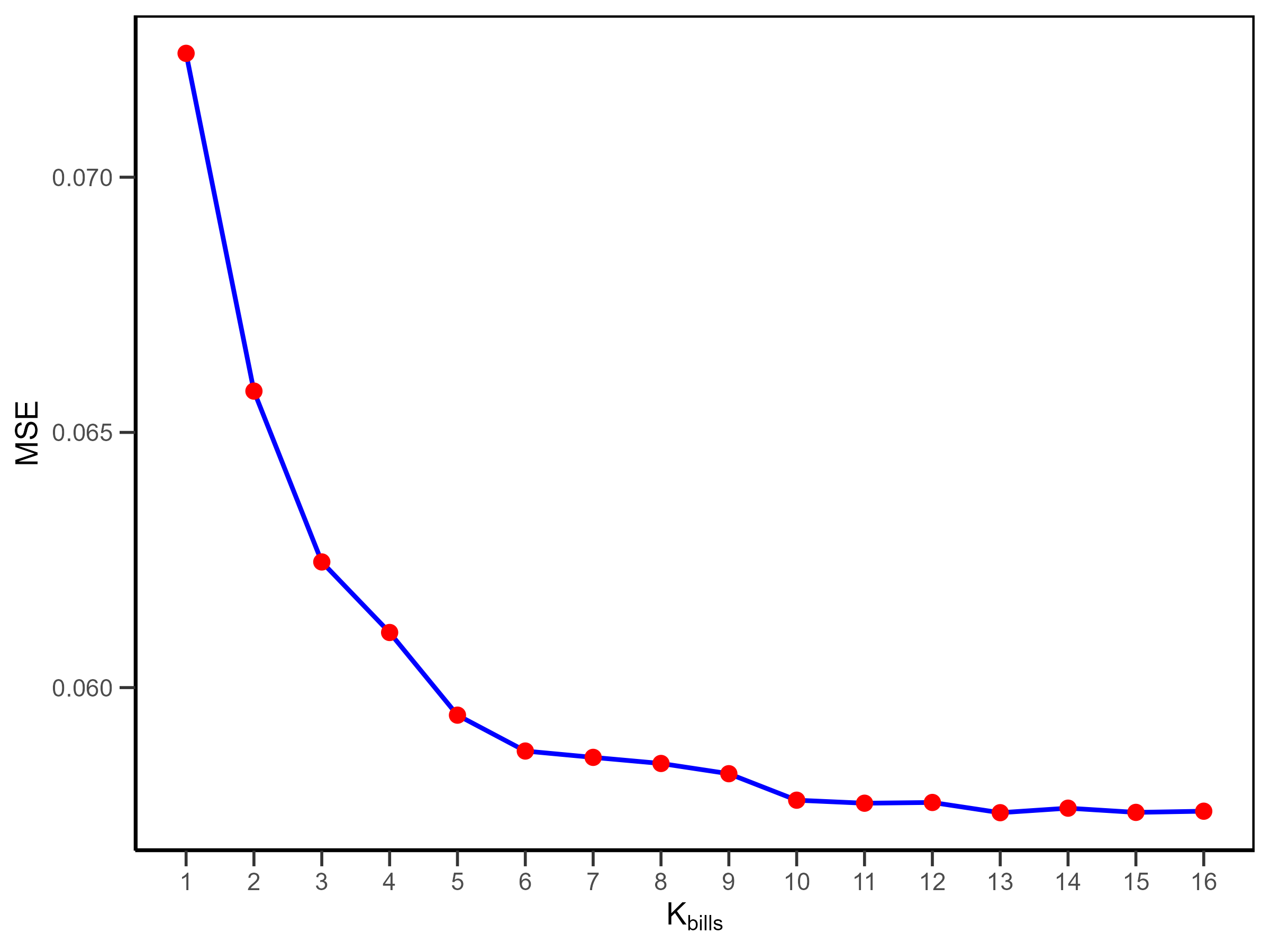}}\\
		\centering
		\caption{Results of cosponsorship network by 10-fold BCV algorithm.}\label{fig:cosponsorship_mse}
	\end{figure}
	
	For the community structure, the estimated two-way clustering among senators aligns well with party affiliation: 89 of 99 senators are assigned consistently, and the adjusted Rand index (ARI) is $0.633$. This confirms that partisan alignment is the strongest organizing principle in the network, while the small discrepancies can plausibly be attributed to other political covariates, as mentioned earlier.
	
	On the bill side, our method identifies 13 communities. Following the logic of \cite{Peel2017truth}, we interpret this community-level structure by comparing the empirical distribution of node metadata within each community to its overall distribution. 
	Specifically, for each community, we compute the relative proportion of each committee, indicating how strongly it is overrepresented within the group. Figure~\ref{fig:cosponsor} displays, for 6 representatives bill communities, the two most enriched committees, linking the detected clusters to substantive legislative themes. 
	The results show clear differentiation across communities, suggesting that our identified 13 clusters capture meaningful heterogeneity in legislative focus.
	
	\begin{figure}[htbp]
		\centering
		\includegraphics[width=0.6\linewidth]{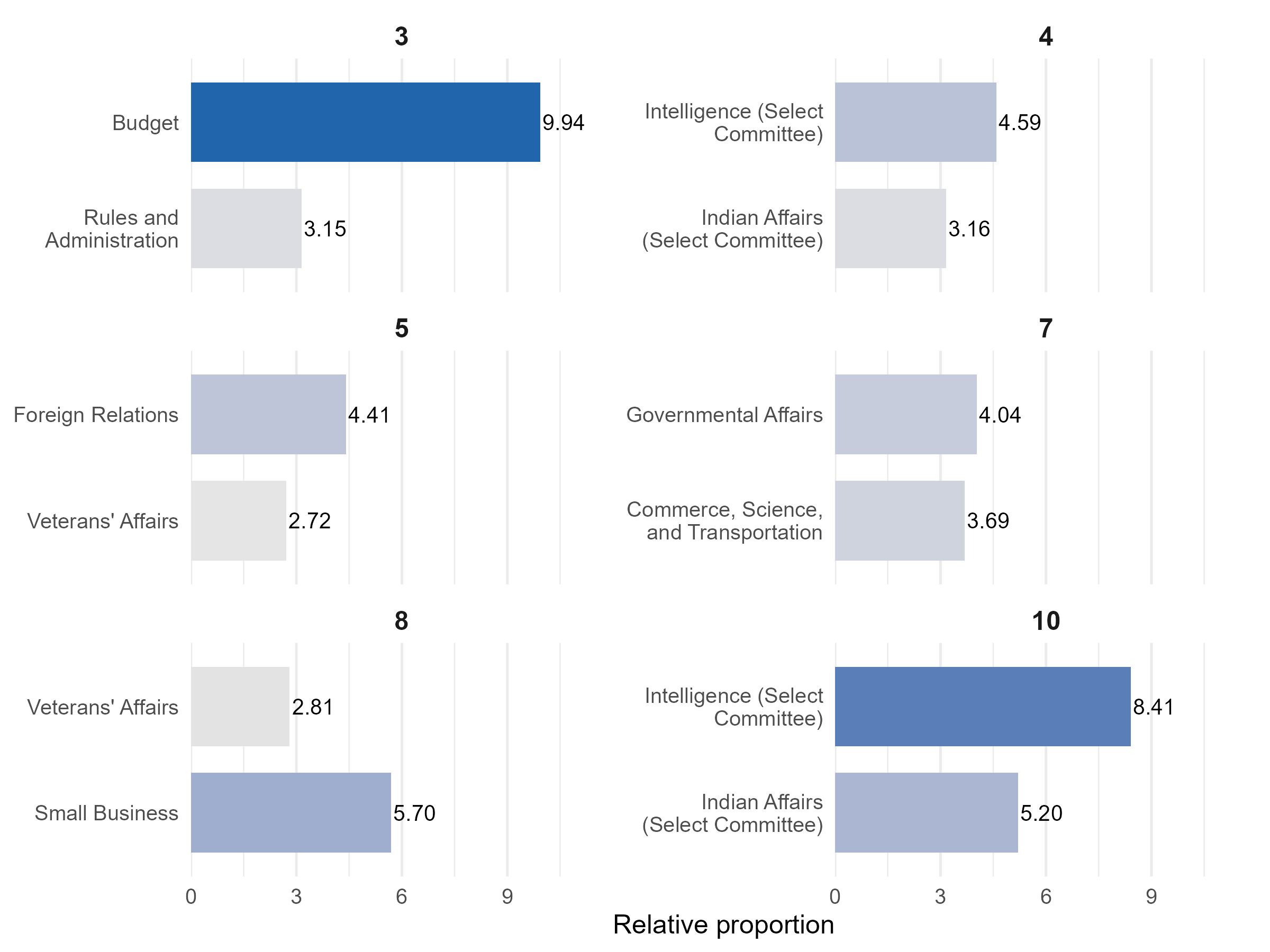}
		\caption{Top enriched committees in 6 representative bill communities. Each panel shows the two most overrepresented committees in the corresponding estimated cluster. Note that the average relative proportion of committee involved in the bills in each community is 0.727.}
		\label{fig:cosponsor}
	\end{figure}
	
	\section{Discussion}\label{sec:discuss}
	In this paper, we establish the first model selection consistency framework for bipartite stochastic block models, several important directions remain open.
	
	Firstly, our theoretical results characterize consistency under a growing-community regime and provide explicit conditions on the admissible growth rates of $(K_1,K_2)$. It would be interesting to further sharpen these results, for example by investigating optimal growth rates or minimax limits, as well as understanding whether similar phenomena arise in other latent structure models beyond the bipartite setting.
	
	Secondly, as in the unipartite network literature, nodes on both sides of a bipartite network are likely to exhibit degree heterogeneity. Incorporating such heterogeneity and developing corresponding model selection procedures under these more general settings represents an important and promising direction for future research.
	
	
	Moreover, since our approach requires the simultaneous selection of $(K_1,K_2)$, a grid search is to be conducted. While this is manageable for small sample sizes or a modest number of communities, the computational cost grows quadratically as the number of communities increases, in contrast to the unipartite case. Developing more computationally efficient strategies for large-scale bipartite networks thus represents another important avenue for future work.
	
	\bibliographystyle{asa}
	\bibliography{ref}

@article{Le2017concentration,
  title={Concentration and Regularization of Random Graphs},
  author={Le, Cam M. and Levina, Elizaveta and Vershynin, Roman},
  journal={Random Structures \& Algorithms},
  volume={51},
  number={3},
  pages={538--561},
  year={2017},
  publisher={Wiley}
}

@article{Zhou2019AnalysisOS,
  title={Analysis of spectral clustering algorithms for community detection: the general bipartite setting},
  author={Zhou, Zhixin and Amini, Arash A.},
  journal={Journal of Machine Learning Research},
  volume={20},
  number={47},
  pages={1--47},
  year={2019}
}

@article{lei2015consistency,
  title={Consistency of spectral clustering in stochastic block models},
  author={Lei, Jing and Rinaldo, Alessandro},
  journal={The Annals of Statistics},
  volume={43},
  number={1},
  pages={215--237},
  year={2015}
}

@article{lei2016goodness,
  title={A goodness-of-fit test for stochastic block models},
  author={Lei, Jing},
  journal={The Annals of Statistics},
  volume={44},
  number={1},
  pages={401},
  year={2016},
  publisher={Institute of Mathematical Statistics}
}

@article{hu2021using,
  title={Using maximum entry-wise deviation to test the goodness of fit for stochastic block models},
  author={Hu, Jianwei and Zhang, Jingfei and Qin, Hong and Yan, Ting and Zhu, Ji},
  journal={Journal of the American Statistical Association},
  volume={116},
  number={535},
  pages={1373--1382},
  year={2021},
  publisher={Taylor \& Francis}
}

@article{wang2017likelihood,
  title={Likelihood-based model selection for stochastic block models},
  author={Wang, YX Rachel and Bickel, Peter J},
  journal={The Annals of Statistics},
  volume={45},
  number={2},
  pages={500--528},
  year={2017}
}

@article{li2020network,
  title={Network cross-validation by edge sampling},
  author={Li, Tianxi and Levina, Elizaveta and Zhu, Ji},
  journal={Biometrika},
  volume={107},
  number={2},
  pages={257--276},
  year={2020},
  publisher={Oxford University Press}
}

@article{chen2018network,
  title={Network cross-validation for determining the number of communities in network data},
  author={Chen, Kehui and Lei, Jing},
  journal={Journal of the American Statistical Association},
  volume={113},
  number={521},
  pages={241--251},
  year={2018},
  publisher={Taylor \& Francis}
}

@article{le2022estimating,
  title={Estimating the number of communities by spectral methods},
  author={Le, Can M and Levina, Elizaveta},
  journal={Electronic Journal of Statistics},
  volume={16},
  number={1},
  pages={3315--3342},
  year={2022},
  publisher={The Institute of Mathematical Statistics and the Bernoulli Society}
}

@article{jin2023optimal,
  title={Optimal estimation of the number of network communities},
  author={Jin, Jiashun and Ke, Zheng Tracy and Luo, Shengming and Wang, Minzhe},
  journal={Journal of the American Statistical Association},
  volume={118},
  number={543},
  pages={2101--2116},
  year={2023},
  publisher={Taylor \& Francis}
}

@article{wu2024eigengap,
  title={An Eigengap Ratio Test for Determining the Number of Communities in Network Data},
  author={Wu, Yujia and Zhang, Jingfei and Lan, Wei and Tsai, Chih-Ling},
  journal={arXiv preprint arXiv:2409.05276},
  year={2024}
}

@article{yang2007consistency,
  title={Consistency of cross validation for comparing regression procedures},
  author={Yang, Yuhong},
  journal={Annals of Statistics},
  volume={35},
  number={6},
  pages={2450--2473},
  year={2007},
  publisher={Institute of Mathematical Statistics}
}

@article{barber2007bimod,
  title = {Modularity and community detection in bipartite networks},
  author = {Barber, Michael J.},
  journal = {Phys. Rev. E},
  volume = {76},
  issue = {6},
  pages = {066102},
  numpages = {9},
  year = {2007},
  month = {Dec},
  publisher = {American Physical Society}
}

@InCollection{Alzahrani2016,
  author    = {Taher Alzahrani and K. J. Horadam},
  editor    = {Jinhu Lü and Xinghuo Yu and Guanrong Chen and Wenwu Yu},
  title     = {Community Detection in Bipartite Networks: Algorithms and Case Studies},
  booktitle = {Complex Systems and Networks: Dynamics, Controls and Applications},
  year      = {2016},
  publisher = {Springer Berlin Heidelberg},
  pages     = {25--50}
}

@article{larremore2020bayes,
  title = {Community detection in bipartite networks with stochastic block models},
  author = {Yen, Tzu-Chi and Larremore, Daniel B.},
  journal = {Phys. Rev. E},
  volume = {102},
  issue = {3},
  pages = {032309},
  numpages = {16},
  year = {2020},
  month = {Sep},
  publisher = {American Physical Society}
}

@book{davis1941deep,
  title={Deep South: A Social Anthropological Study of Caste and Class},
  author={Davis, Allison and Gardner, Burleigh B. and Gardner, Mary R.},
  year={1941},
  publisher={University of Chicago Press},
  address={Chicago}
}

@article{liu2010community,
  title={Community detection in large-scale bipartite networks},
  author={Liu, Xin and Murata, Tsuyoshi},
  journal={Transactions of the Japanese Society for Artificial Intelligence},
  volume={25},
  number={1},
  pages={16--24},
  year={2010},
  publisher={The Japanese Society for Artificial Intelligence}
}

@article{Guimera2007module,
  title = {Module identification in bipartite and directed networks},
  author = {Guimer\`a, Roger and Sales-Pardo, Marta and Amaral, Lu\'{\i}s A. Nunes},
  journal = {Phys. Rev. E},
  volume = {76},
  issue = {3},
  pages = {036102},
  numpages = {8},
  year = {2007},
  month = {Sep},
  publisher = {American Physical Society}
}

@article{DOREIAN200429,
title = {Generalized blockmodeling of two-mode network data},
journal = {Social Networks},
volume = {26},
number = {1},
pages = {29-53},
year = {2004},
author = {Patrick Doreian and Vladimir Batagelj and Anuška Ferligoj}
}

@misc{fowler2007cosponsor,
  author    = {James H. Fowler},
  title     = {Replication data for: Legislative Cosponsorship Networks in the U.S. House and Senate},
  year      = {2007},
  publisher = {Harvard Dataverse},
  version   = {V1},
  doi       = {10.7910/DVN/O22JMY},
  url       = {https://doi.org/10.7910/DVN/O22JMY}
}

@article{lo2025senate, 
title={A Statistical Model of Bipartite Networks: Application to Cosponsorship in the United States Senate}, 
journal={Political Analysis}, 
author={Lo, Adeline and Olivella, Santiago and Imai, Kosuke}, 
year={2025},
pages={1–20}}

@article{Peel2017truth,
author = {Leto Peel and Daniel B. Larremore and Aaron Clauset },
title = {The ground truth about metadata and community detection in networks},
journal = {Science Advances},
volume = {3},
number = {5},
pages = {e1602548},
year = {2017}}

@article{yang2025pnncv,
    author = {Bokai Yang and Yuanxing Chen and Yuhong Yang},
    title = {Network Cross-Validation for Nested Models by Edge-Sampling},
    journal = {arXiv preprint arXiv: 2506.14244},
    year = {2025}
}

@article{Larremore2014infer,
  title = {Efficiently inferring community structure in bipartite networks},
  author = {Larremore, Daniel B. and Clauset, Aaron and Jacobs, Abigail Z.},
  journal = {Phys. Rev. E},
  volume = {90},
  issue = {1},
  pages = {012805},
  numpages = {12},
  year = {2014},
  month = {Jul},
  publisher = {American Physical Society}
}

@article{ZHOU2018novel,
title = {A novel community detection method in bipartite networks},
journal = {Physica A: Statistical Mechanics and its Applications},
volume = {492},
pages = {1679-1693},
year = {2018},
author = {Cangqi Zhou and Liang Feng and Qianchuan Zhao}
}

@article{newman2001structure,
  title={The structure of scientific collaboration networks},
  author={Newman, Mark EJ},
  journal={Proceedings of the National Academy of Sciences},
  volume={98},
  number={2},
  pages={404--409},
  year={2001},
  publisher={National Acad Sciences}
}

@article{zhang2011tagaware,
  title={Tag-aware recommender systems: A state-of-the-art survey},
  author={Zhang, Zi-Ke and Zhou, Tao and Zhang, Yi-Cheng},
  journal={Journal of Computer Science and Technology},
  volume={26},
  number={5},
  pages={767--777},
  year={2011},
  publisher={Springer}
}

@article{pesantez2017efficient,
  title={Efficient detection of communities in biological bipartite networks},
  author={Pes{\'a}ntez-Cabrera, Paola and Kalyanaraman, Ananth},
  journal={IEEE/ACM transactions on computational biology and bioinformatics},
  volume={16},
  number={1},
  pages={258--271},
  year={2017},
  publisher={IEEE}
}

@article{tamarit2020hierarchical,
  title={Hierarchical clustering of bipartite data sets based on the statistical significance of coincidences},
  author={Tamarit, Ignacio and Pereda, Mar{\'\i}a and Cuesta, Jos{\'e} A},
  journal={Physical Review E},
  volume={102},
  number={4},
  pages={042304},
  year={2020},
  publisher={APS}
}

@article{braun2024strong,
  title={Strong consistency guarantees for clustering high-dimensional bipartite graphs with the spectral method},
  author={Braun, Guillaume},
  journal={Electronic Journal of Statistics},
  volume={18},
  number={2},
  pages={2798--2823},
  year={2024},
  publisher={The Institute of Mathematical Statistics and the Bernoulli Society}
}

@article{zhou2020optimal,
  title={Optimal bipartite network clustering},
  author={Zhou, Zhixin and Amini, Arash A},
  journal={Journal of Machine Learning Research},
  volume={21},
  number={40},
  pages={1--68},
  year={2020}
}

@article{stone1974cv,
    author = {Stone, Mervyn},
    title = {Cross-validatory choice and assessment of statistical predictions},
    journal = {Journal of the Royal Statistical Society: Series B (Methodological)},
    volume    = {36},
    number    = {2},
    pages     = {111--147},
    year      = {1974}
}

@article{shao1993cv,
    author = {Shao, Jun},
    title = {Linear model selection by cross-validation},
    journal = {Journal of the American Statistical Association},
    volume  = {88},
    number = {422},
    pages = {486--494},
    year =  {1993}
}

@article{arlot2010survey,
  author    = {Sylvain Arlot and Alain Celisse},
  title     = {A survey of cross-validation procedures for model selection},
  journal   = {Statistics Surveys},
  volume    = {4},
  pages     = {40--79},
  year      = {2010}
}

@article{lei2025moderncv,
    author = {Lei, Jing},
    title = {A Modern Theory of Cross-Validation through the Lens of Stability},
    journal = {Foundations and Trends in Statistics},
    year={2025},
    volume={1},
    number={3--4},
    pages={391--540}
}

@article{geisser1975pred,
    author = {Geisser, Seymour},
    title = {The Predictive Sample Reuse Method with Applications},
    journal={Journal of the American Statistical association},
    volume={70},
    number={350},
    pages={320--328},
    year={1975},
    publisher={Taylor \& Francis}
}

@misc{van2003unified,
  title={Unified cross-validation methodology for selection among estimators and a general cross-validated adaptive epsilon-net estimator: Finite sample oracle inequalities and examples},
  author={Van Der Laan, Mark J and Dudoit, Sandrine},
  year={2003},
  publisher={bepress}
}

@article{Chakrabarty2025subsampling,
	title={Network Cross-Validation and Model Selection via Subsampling},
	author={Chakrabarty, Sayan and Sengupta, Srijan and Chen, Yuguo},
	journal={arXiv:2504.06903},
	year={2025}
}

@article{Bhadra2025unified,
	title={A Unified Framework for Community Detection and Model Selection in Blockmodels},
	ISSN={1537-2715},
	DOI={10.1080/10618600.2025.2590073},
	journal={Journal of Computational and Graphical Statistics},
	publisher={Informa UK Limited},
	author={Bhadra, Subhankar and Tang, Minh and Sengupta, Srijan},
	year={2025},
	pages={1–27} }
	
	\newpage
	\appendix
	
	\begin{center}
		\Large \textbf{Supplementary Materials for ``Cross-Validation in Bipartite Networks''}
	\end{center}
	
	In this supplementary file, we provide the proofs of the theoretical results given in the main article. In Section A.1, we introduce additional notations and present the whole scaling conditions addressing diminishing $w_{n_1,n_2}$ and $d_{n_1,n_2}$, which is the scenario we considered in the subsequent proofs. In Section A.2 we address the community recovery results on two sides under our framework. In Section A.3 we give the proof of Theorem 1, dividing the whole candidate space into five subspaces. In Section A.4 and A.5, we give the proofs of the two corollaries.
	
	\section{Proofs}
	\subsection{Additional Notations and Full Assumptions}
	We focus on a single split in the penalized edge split cross-validation procedure, since if a single split has the consistency result, then multiple splits hold automatically.
	
	We start with additional notations. For any matrix $M=[M_{ij}]$, we denote $\|M\|_F$ as its Forbenius norm, and let $\|M\|_{\max}:=\max_{i,j}|M_{ij}|$ denote the maximum magnitude of the entries in the matrix. For random variables $\{a_n\}$ and $\{b_n\}$ with $b_n>0$, we write $a_n=O_{\mathrm{P}}(b_n)$ if $a_n/b_n$ is bounded in probability and write $a_n=\omega_{\mathrm{P}}(b_n)$ if $a_n/b_n\overset{p}{\to}\infty$. For two symmetric matrices $A$ and $B$, we write $A\succeq B$ if $A-B$ is positive definite. Finally, we write $\mathbb{O}^{r\times r}$ for all r-dimensional orthogonal matrices.
	
	For notation simplicity, we write $B^{(n_1,n_2)}$ as $B$, $B_0^{(n_1,n_2)}$ as $B_0$, $\rho_{n_1,n_2}$ as $\rho$, $\beta_{n_1,n_2}$ as $\beta$, $d_{n_1,n_2}$ as $d$ and denote $n=\max\{n_1,n_2\}$, $K=\min\{K_1,K_2\}$. Denote $\mathcal{G}_k$ as the $k$-th community on the first side, and denote $\mathcal{H}_l$ as the $l$-th community on the second side. Without loss of generalization, we will assume that $K_1=K\leq K_2$ throughout the proof.
	
	\subsection{Label Recovery Results}
	The next lemma is a direct result of Theorem 1.1 in \cite{Le2017concentration} and the remark there, by padding $A$ with rows (or columns) of zeros to get a square matrix.
	\begin{lemmaA}\label{lemma_con1}
		Assume that $A \in \mathbb{R}^{n_1\times n_2}$ is generated from the SBM model. Assume that $\rho\geq C_0\log n/n$ for some constant $C_0$, then for any $r\geq1$, there exists a constant $C=C(r,C_0)$, such that
		$$\|A-P\|_2\leq C\sqrt{n\rho_n}$$
		with probability at least $1-n^{-r}$.
	\end{lemmaA}
	
	The next lemma is a direct result of Lemma 2 in \cite{li2020network}, by padding $Z$ with rows (or columns) of zeros to get a square matrix, and correspondingly enlarge $G$ to be an adjacency matrix of an $n\times n$ Erd\H{o}s-Renyi graph with the probability of edge $w$.
	
	\begin{lemmaA}\label{lemma_con2}
		Let $G$ be an adjacency matrix of an $n_1\times n_2$ Erd\H{o}s-Renyi graph with the probability of edge $w \geq C_1 \log n / n$ for a constant $C_1$. That is, for $i\in [n_1]$ and $j\in[n_2]$, $G_{ij}\sim \text{Bernoulli}(w)$ independently. Then for any $\gamma > 0$, with probability at least $1 - n^{-\gamma}$, the following relationship holds for any $Z \in \mathbb{R}^{n_1 \times n_2}$ with $\mathrm{rank}(Z) \leq K$
		\[
		\left\| \frac{1}{w} Z \circ G - Z \right\|_2 \leq 2C \sqrt{\frac{nK}{w}} \, \|Z\|_{\max}
		\]
		where $C = C(\gamma, C_1)$ is the constant from Lemma 1 that only depends on $\gamma$ and $C_1$.
	\end{lemmaA}
	
	The next lemma is a direct result of Lemma 3 in \cite{li2020network}, by padding $X$ with rows (or columns) of zeros to get a square matrix.
	
	\begin{lemmaA}\label{lemma_con3}
		Let $X$ be an $n_1 \times n_2$ matrix with each entry $X_{ij}$ being independent and bounded random variables, such that $\max_{i,j} |X_{ij}| \leq \sigma$ with probability 1. Then for any $\gamma > 0$, with probability at least $1 - n^{-\gamma}$,
		\[
		\|X\|_2 \leq C' \max\left( \sigma_1, \sigma_2, \sqrt{\log n} \right)
		\]
		in which $C' = C'(\sigma, \gamma)$ is a constant that only depends on $\gamma$ and $\sigma$,
		\[
		\sigma_1 = \max_i \sqrt{ E \sum_j X_{ij}^2 }, \quad \sigma_2 = \max_j \sqrt{ E \sum_i X_{ij}^2 }.
		\]
	\end{lemmaA}
	
	Now we can get a concentration result for the partially observed matrix in the bipartite setting.
	\begin{propositionA}\label{prop_con}
		Assume that $\rho\geq C_0\log n/n$ and the training proportion $w\geq C_1\log n/n$. Then the adjusted partially observed matrix satisfies for any $\gamma>0$, with probability at least $1 - 3n^{-\gamma}$,
		$$
		\| Y/w - P \|_2 \leq \tilde{C} \max \left\{ \left( \frac{K n\rho^2}{w} \right)^{1/2}, \left( \frac{n\rho}{w} \right)^{1/2}, \left( \frac{\log n}{w} \right)^{1/2} \right\},
		$$
		where $\tilde{C} = \tilde{C}(\gamma, C_0, C_1)$ is a constant that depends only on $C_0$, $C_1$ and $\gamma$.
	\end{propositionA}
	
	\begin{proof}[Proof of Proposition \ref{prop_con}]
		Just follow the lines of the proof of Theorem 1 in \cite{li2020network}, substituting the lemmas there by Lemma 2--4 here.
	\end{proof}
	
	In our setting, the rank $K=\min\{K_1,K_2\}$ may increase as $n$ grows, so unlike the traditional settings the bound in Proposition 1 shall be written more elaborated as 
	\begin{equation}
		\| Y/w - P \|_2 \leq \tilde{C}\sqrt{\frac{n\rho}{w}}\maxsqrtrho. \label{adjusted_concen}
	\end{equation}
	Without generality, we assume that $K_1=K=\min\{K_1,K_2\}$. Now we analyze the consistency of the spectral clustering when $k$ coincides with the true $K_1$. 
	
	Recall we denote $N_r=\text{diag}(n_{r1},\ldots,n_{rK_r})$, where $n_{rj}$ is the size of the $j$-th cluster of side $r$, and denote 
	$$\bar{B}:=N_1^{1/2}BN_2^{1/2}\qquad \text{ and }\qquad\bar{Z}_r := Z_r N_r^{-1/2}\quad\text{for }r\in\{1,2\}.$$ 
	The following lemma reveals key spectral properties for bipartite SBMs:
	
	\begin{lemmaA}[Lemma 1 of \cite{Zhou2019AnalysisOS}]\label{lemma_svd}
		Assume that $\bar{B} = U \Sigma V^{\top}$ is the reduced SVD of $\bar{B}$, where $U\in \mathbb{O}^{K_1 \times K}$, $V\in \mathbb{O}^{K_2 \times K}$, $\Sigma = \mathrm{diag}(\sigma_1, \dots, \sigma_K)>0$, and $K = \min\{K_1, K_2\}$. Then,
		\[
		P = (\bar{Z}_1 U) \, \Sigma \, (\bar{Z}_2 V)^T
		\]
		is the reduced SVD of $P$ where $\bar{Z}_r := Z_r N_r^{-1/2}$ is itself a column orthogonal matrix.
	\end{lemmaA}
	
	\begin{lemmaA}\label{lemma_svdcons}
		Assume the rank $K$-truncated SVD of $Y/w$ is $\hat{U}\hat{\Sigma}\hat{V}^{\top}$, then for any $\gamma>0$, with probability at least $1-3n^{-\gamma}$,
		$$\|\hat{U}-\bar{Z}_1UQ\|_F,\,\|\hat{V}-\bar{Z}_2VQ\|_F\leq \frac{C_2}{\sigma_K}\sqrt{\frac{Kn\rho}{w}}\maxsqrtrho$$
		for some $Q\in \mathbb{O}^{K\times K}$, where $C_2=C_2(\tilde{C})$ is a constant that depends only on $\tilde{C}$ in Proposition \ref{prop_con}, and $\sigma_K$ is the smallest nonzero singular value in Lemma 1.
	\end{lemmaA}
	
	\begin{proof}[Proof of Lemma \ref{lemma_svdcons}]
		Following the lines of the proof of Lemma 3 in \cite{Zhou2019AnalysisOS}, substituting $A_{\mathrm{re}}$ there by $Y/w$, denote
		$$\bar{W}_1=\frac{1}{\sqrt{2}}\begin{pmatrix}
			\bar{Z}_1U\\
			\bar{Z}_2V
		\end{pmatrix},\ \hat{W}_1=\frac{1}{\sqrt{2}}\begin{pmatrix}
			\hat{U}\\
			\hat{V}
		\end{pmatrix},$$
		then 
		$$\min_{Q\in\mathbb{O}^{K\times K}}\|\hat{W}_1-\bar{W}_1Q\|_F\leq\frac{\sqrt{2K}}{\sigma_K}\|Y/w-P\|_2.$$
		Since 
		$$2\|\hat{W}_1-\bar{W}_1Q\|_F^2=\|\hat{U}-\bar{Z}_1UQ\|_F^2+\|\hat{V}-\bar{Z}_2VQ\|_F^2,$$
		we obtain the desired result by \eqref{adjusted_concen}.
	\end{proof}
	\begin{lemmaA}[Lemma 5.3 of \cite{lei2015consistency}]\label{lemma_labelcons}
		Let $\hat{U}$ and $U$ be two $n \times K$ matrices such that $U$ contains $K$ distinct rows. Denote $n_k$ as the amount of appearance of the $k$-th distinct row of $U$. Let $\delta$ be the minimum distance between two distinct rows of $U$, and $g$ be the membership vector given by clustering the rows of $U$. Let $\hat{g}$ be the output of a $k$-means clustering algorithm on $\hat{U}$, with objective value no larger than a constant factor of the global optimum. Assume that $\|\hat{U} - U\|_F^2 \leq c n_k\delta^2$ for some small enough constant $c$ and for all $k\in[K]$. Then $\hat{g}$ agrees with $g$ on all but $c^{-1} \|\widehat{U} - U\|^2_F \delta^{-2}$ nodes, after an appropriate label permutation.
	\end{lemmaA}
	
	From now on, for notation simplicity we may refer $C$ as a constant, and different $C$ may refer to different values.
	\begin{propositionA}\label{prop_labelcons}
		Assume $\nmin\rho w\gg K^2$ and $\min\{n_1,n_2\}w\gg K^3$. The estimated label $\hat{c}^{(K_1,K_2)}_1$ obtained by applying procedures 1--3 in step 1 of Algorithm 1 with $(K_1',K_2')$ coincides with $(K_1,K_2)$ satisfies that for all $\gamma>0$, there exists a constant $C=C(\tilde{C})$ depending on $\tilde{C}$ in Proposition \ref{prop_con} such that it agrees with $c_1$ on all but $\leq C(nK\maxrho/{n_2\rho w})$ nodes with probability larger than $1-n^{-\gamma}$, after an appropriate label permutation.
	\end{propositionA}
	
	\begin{proof}[Proof of Proposition \ref{prop_labelcons}]
		We apply Lemma \ref{lemma_labelcons}, with $U$ substituted by $\bar{Z}_1UQ$. Notice that $U,Q\in\mathbb{O}^{K\times K}$, therefore the minimum distance between two distinct rows of $\bar{Z}_1UQ$ equals to the same distance for $\bar{Z}_1$, which is of order $\sqrt{K}/\sqrt{n_1}$. Also the $K$ different rows of $\bar{Z}_1UQ$ corresponds to the $K$ different rows of $\bar{Z}_1$, thus the result of clustering rows of $\bar{Z}_1UQ$ is just $c_1$.
		
		Notice $\sigma_K$ is the minimal nonzero singular value of $\bar{B}=\rho N_1^{1/2}B_0N_2^{1/2}$, thus it is the square root of the minimal eigenvalue of $$\bar{B}\bar{B}^T=\rho^2N_1^{1/2}B_0N_2B_0^TN_1^{1/2}\succeq\rho^2\pi_0\frac{n_2}{K_2}N_1^{1/2}B_0B_0^TN_1^{1/2}\succeq\rho^2\pi_0^2\frac{n_1n_2}{KK_2}B_0B_0^T.$$
		Denote the smallest nonzero singular value of $B_0$ as $\sigma_0$, then from the above formula we see that $$\sigma_K\geq\rho\pi_0\sqrt{\frac{n_1n_2}{KK_2}}\sigma_0\geq C\rho\sqrt{\frac{n_1n_2}{K}},$$
		by our assumption for $B_0$.
		
		Thus by Lemma \ref{lemma_svdcons}, we have
		$$\|\hat{U}-\bar{Z}_1UQ\|_F^2\leq \tilde{C}\frac{nK^2}{n_1n_2\rho w}\maxrho,$$
		for some constant $\tilde{C}$. Then as long as $\nmin\rho w\gg K^2$ and $\min\{n_1,n_2\}w\gg K^3$, the condition of Lemma \ref{lemma_labelcons} is satisfied, and the result follows from the order of $\delta$. 
	\end{proof}
	
	Now we prove the consistency result of the estimated label on the more clusters side.
	
	\begin{propositionA}\label{prop_labelcons2}
		Assume $\nmin\rho w\beta\gg K^2$ and $\min\{n_1,n_2\}w\beta\gg K^3$. The estimated label $c^{(K_1,K_2)}_2$ obtained by applying procedures 1--3 in step 1 of Algorithm 1 with $(K_1',K_2')$ coincides with $(K_1,K_2)$ satisfies that for all $\gamma>0$, there exists a constant $C=C(\tilde{C})$ depending on $\tilde{C}$ in Proposition \ref{prop_con} such that it agrees with $c_1$ on all but $\leq C(nK^2\maxrho/(n_1\rho w\beta K_2))$ nodes with probability larger than $1-n^{-\gamma}$, after an appropriate label permutation.
	\end{propositionA}
	
	\begin{proof}[Proof of Proposition \ref{prop_labelcons2}]
		The only difference here from Proposition \ref{prop_labelcons} is that the minimum distance between two distinct rows of $\bar{Z}_2VQ$ is no longer simply $1/\sqrt{n_2}$. In fact, follow the lines in the proof of Theorem 7 in \cite{Zhou2019AnalysisOS}, one can prove that the minimum distance $\delta$ here satisfies
		$$\delta\geq\sqrt{\frac{\beta K_2}{n_2}}.$$
		Therefore, by Lemma \ref{lemma_svdcons}, since $$\|\hat{V}-\bar{Z}_2VQ\|_F^2\leq C\frac{nK^2}{n_1n_2\rho w}\maxrho,$$ as long as $\nmin\rho w\beta\gg K^2$ and $\min\{n_1,n_2\}w\beta\gg K^3$, we have $\|\hat{V}-\bar{Z}_2VQ\|_F^2\ll n_2\delta^2/K_2=\beta$, and thus by Lemma \ref{lemma_labelcons} the conclusion follows.
	\end{proof}
	
	\subsection{Proof of Theorem 1}
	We first give the upper bound of the loss function $\ell$ when $(K_1',K_2')$ coincides with $(K_1,K_2)$, where we denote
	$$\ell_{K_1',K_2'}(A,\mathcal{E}^c)=\sum_{(i,j)\in\mathcal{E}^c}(A_{ij}-\hat{P}^{(K_1',K_2')}_{ij})^2,$$
	$$\ell_{0}(A,\mathcal{E}^c)=\sum_{(i,j)\in\mathcal{E}^c}(A_{ij}-P_{ij})^2.$$
	
	\begin{propositionA}\label{prop_sbmupper}
		As long as $n\rho K=\Omega(K_2\log n)$, $\nmin\rho w\beta\gg K^2$ and $\min\{n_1,n_2\}w\beta\gg K^3$, we have with probability larger than $1-3K^2K_2^2n^{-\gamma}$,
		$$\ell_{K_1,K_2}(A,\mathcal{E}^c)-\ell_0(A,\mathcal{E}^c)\leq C(\gamma)\left(\frac{nK^2\rho(1-w)}{w\beta}\maxrho\right).$$
	\end{propositionA}
	
	\begin{proof}[Proof of Proposition \ref{prop_sbmupper}]
		For notation simplicity, we will simply denote $\hat{c}^{(K_1',K_2')}_1$ as $\hat{c}_1$ and $\hat{c}^{(K_1',K_2')}_2$ as $\hat{c}_2$ when there is no notation confusion. We define several set of entries:
		$$\mathcal{Q}_{k_1,k_2,l_1,l_2}=\{(i,j):\hat{c}_{1,i}=k_1,c_{1,i}=l_1,\hat{c}_{2,j}=k_2,c_{2,j}=l_2\},$$
		$$\mathcal{U}_{k_1,k_2,l_1,l_2}=\{(i,j)\in\mathcal{E}:\hat{c}_{1,i}=k_1,c_{1,i}=l_1,\hat{c}_{2,j}=k_2,c_{2,j}=l_2\},$$
		$$\mathcal{T}_{k_1,k_2,l_1,l_2}=\{(i,j)\in\mathcal{E}^c:\hat{c}_{1,i}=k_1,c_{1,i}=l_1,\hat{c}_{2,j}=k_2,c_{2,j}=l_2\}.$$
		Let $\mathcal{T}_{\cdot,\cdot,l_1,l_2}:=\cup_{k_1,k_2}\mathcal{T}_{k_1,k_2,l_1,l_2}$ be the union taken over the first two entries, and similarly define $\mathcal{T}_{k_1,k_2,\cdot,\cdot}$, $\mathcal{U}_{\cdot,\cdot,l_1,l_2}$, $\mathcal{U}_{k_1,k_2,\cdot,\cdot}$, $\mathcal{Q}_{\cdot,\cdot,l_1,l_2}$ and $\mathcal{Q}_{k_1,k_2,\cdot,\cdot}$. By Proposition \ref{prop_labelcons} and Proposition \ref{prop_labelcons2}, we have with probability larger than $1-n^{-\gamma}$,
		$$|\cup_k(\hat{\mathcal{G}}_k\Delta\mathcal{G}_k)|\leq C\left(\frac{nK}{n_2\rho w}\maxrho\right)$$ and
		$$|\cup_k(\hat{\mathcal{H}}_k\Delta\mathcal{H}_k)|\leq C\left(\frac{nK^2}{n_1\rho w\beta K_2}\maxrho\right).$$ Thus notice
		$$\mathcal{Q}_{k_1,k_2,l_1,l_2}= (\hat{\mathcal{G}}_{k_1}\cap \mathcal{G}_{l_1})\times(\hat{\mathcal{G}}_{k_2}\cap \mathcal{G}_{l_2}),$$
		we have
		$$|\mathcal{Q}_{k_1,k_2,\cdot,\cdot}\backslash\mathcal{Q}_{\cdot,\cdot,k_1,k_2}|\leq|\hat{\mathcal{G}}_{k_1}\backslash \mathcal{G}_{k_1}|\times|\hat{\mathcal{H}}_{k_2}|+|\hat{\mathcal{G}}_{k_1}\cap\mathcal{G}_{k_1}|\times|\hat{\mathcal{H}}_{k_2}\backslash\mathcal{H}_{k_2}|,$$
		which results in with probability larger than $1-n^{-\gamma}$,
		\begin{align*}
			|\mathcal{Q}_{k_1,k_2,\cdot,\cdot}\Delta\mathcal{Q}_{\cdot,\cdot,k_1,k_2}|&\leq|\hat{\mathcal{G}}_{k_1}\Delta \mathcal{G}_{k_1}|\times|\hat{\mathcal{H}}_{k_2}\cup\mathcal{H}_{k_2}|+|\hat{\mathcal{G}}_{k_1}\cap\mathcal{G}_{k_1}|\times|\hat{\mathcal{H}}_{k_2}\Delta\mathcal{H}_{k_2}|\\
			&\leq C\frac{nK}{\rho w\beta K_2}\maxrho,
		\end{align*}
		and similarly
		$$\left|\cup_{k_1,k_2}(\mathcal{Q}_{k_1,k_2,\cdot,\cdot}\Delta\mathcal{Q}_{\cdot,\cdot,k_1,k_2})\right|=O_{\mathrm{P}}\left(\frac{nK}{\rho w}\maxbeta\maxrho\right).$$
		Consequently, since with probability larger than $1-n^{-\gamma}$,
		$$|\mathcal{U}_{\cdot,\cdot,k_1,k_2}|\geq c\frac{n_1n_2}{K_1K_2}w,$$
		as long as $\nmin\rho w\beta\gg K^2$ and $\min\{n_1,n_2\}w\beta\gg K^3$, we have for probability larger than $1-K_1K_2n^{-\gamma}$,
		$$|\mathcal{U}_{k_1,k_2,\cdot,\cdot}|\geq c\frac{n_1n_2}{K_1K_2}w,$$
		$$\left|\mathcal{U}_{k_1,k_2,\cdot,\cdot}\Delta\mathcal{U}_{\cdot,\cdot,k_1,k_2}\right|\leq C\left(\frac{nK}{\rho\beta K_2}\maxrho\right),$$
		and
		$$\left|\cup_{k_1,k_2}(\mathcal{T}_{k_1,k_2,\cdot,\cdot}\Delta\mathcal{T}_{\cdot,\cdot,k_1,k_2})\right|\leq C\left(\frac{nK(1-w)}{\rho w}\maxbeta\maxrho\right).$$
		Denote $\hat{B}^{(K_1,K_2)}$ simply as $\hat{B}$. By Bernstein's inequality, we have
		$$\mathbb{P}\left(\left|\frac{\sum_{U_{\cdot,\cdot,k,l}}A_{ij}}{|U_{\cdot,\cdot,k,l}|}-B_{kl}\right|\geq C\sqrt{\frac{\rho\log n}{|U_{\cdot,\cdot,k,l}|}}\right)\leq2n^{-\gamma}$$
		for some constant $C$ depending on $\gamma$. Moreover, by Hoeffding's inequality, we have for probability larger than $1-n^{-\gamma}$,
		$$\left|\frac{\sum_{\mathcal{U}_{\cdot,\cdot,k,l}\Delta \mathcal{U}_{k,l,\cdot,\cdot}}A_{ij}}{|\mathcal{U}_{\cdot,\cdot,k,l}\Delta \mathcal{U}_{k,l,\cdot,\cdot}|}-\frac{\sum_{\mathcal{U}_{\cdot,\cdot,k,l}\Delta \mathcal{U}_{k,l,\cdot,\cdot}}B_{ij}}{|\mathcal{U}_{\cdot,\cdot,k,l}\Delta \mathcal{U}_{k,l,\cdot,\cdot}|}\right|\leq C\sqrt{\frac{\log n}{|\mathcal{U}_{\cdot,\cdot,k,l}\Delta \mathcal{U}_{k,l,\cdot,\cdot}|}}.$$
		Therefore, we have for probability larger than $1-3n^{-\gamma}$,
		\begin{align}\label{Bbern2}
			|\hat{B}_{kl}-B_{kl}|=&|\frac{\sum_{\mathcal{U}_{k,l,\cdot,\cdot}}A_{ij}}{|\mathcal{U}_{k,l,\cdot,\cdot}|}-B_{kl}|\nonumber\\
			\leq&\frac{|\mathcal{U}_{\cdot,\cdot,k,l}|}{|\mathcal{U}_{k,l,\cdot,\cdot}|}\left|\frac{\sum_{\mathcal{U}_{\cdot,\cdot,k,l}}A_{ij}}{|\mathcal{U}_{\cdot,\cdot,k,l}|}-B_{kl}\right|+\left|1-\frac{|\mathcal{U}_{\cdot,\cdot,k,l}|}{|\mathcal{U}_{k,l,\cdot,\cdot}|}\right|B_{kl}\nonumber\\
			&+\frac{|\mathcal{U}_{\cdot,\cdot,k,l}\Delta \mathcal{U}_{k,l,\cdot,\cdot}|}{|\mathcal{U}_{k,l,\cdot,\cdot}|}\frac{\sum_{\mathcal{U}_{\cdot,\cdot,k,l}\Delta \mathcal{U}_{k,l,\cdot,\cdot}}A_{ij}}{|\mathcal{U}_{\cdot,\cdot,k,l}\Delta \mathcal{U}_{k,l,\cdot,\cdot}|}\nonumber\\
			\leq&C\left(\sqrt{\frac{\rho KK_2\log n}{n_1n_2w}}+\frac{nK^2\maxrho}{n_1n_2\rho w\beta}\rho+\frac{\sqrt{\log n|\mathcal{U}_{\cdot,\cdot,k,l}\Delta \mathcal{U}_{k,l,\cdot,\cdot}|}}{|\mathcal{U}_{k,l,\cdot,\cdot}|}\right)\nonumber\\
			\leq&C\left(\frac{nK^2\maxrho}{n_1n_2w\beta}\right),
		\end{align}
		as long as $K_2\nmin w\rho\beta\ll K^3n$ and $n\rho K=\Omega(K_2\log n)$. Here the first condition is directly followed by $\nmin\rho w\beta\gg K^2$.
		
		Now comparing $\ell_{K_1,K_2}$ and $\ell_0$, we have
		\begin{align*}
			\ell_{K_1,K_2}(A, \mathcal{E}^c) - \ell_0(A, \mathcal{E}^c) 
			=& \sum_{k_1, k_2, l_1, l_2} \sum_{(i,j) \in \mathcal{T}_{k_1, k_2, l_1, l_2}} 
			\left[ (A_{ij}-\hat{B}_{k_1 k_2})^2 -(A_{ij}-B_{l_1 l_2})^2\right] \\
			=& \sum_{k_1, k_2} \sum_{(i,j) \in \mathcal{T}_{k_1, k_2, k_1, k_2}} 
			\left[ (A_{ij}-\hat{B}_{k_1 k_2})^2 -(A_{ij}-B_{k_1 k_2})^2 \right] \\
			&+ \sum_{(k_1, k_2) \ne (l_1, l_2)} \sum_{(i,j) \in \mathcal{T}_{k_1, k_2, l_1, l_2}} 
			\left[(A_{ij}-\hat{B}_{k_1 k_2})^2 -(A_{ij}-B_{l_1 l_2})^2\right] \\
			:=& \mathcal{I} + \mathcal{II}.
		\end{align*}
		Thus we have with probability larger than $1-3K_1^2K_2^2n^{-\gamma}$,
		\begin{align*}
			|\mathcal{I}|\leq& \sum_{k_1, k_2} \sum_{(i,j) \in \mathcal{T}_{k_1, k_2, k_1, k_2}} 
			2|\hat{B}_{k_1 k_2} - B_{k_1 k_2}| \left(A_{ij} + B_{k_1 k_2} + |\hat{B}_{k_1 k_2} - B_{k_1 k_2}| \right) \notag \\
			\leq& C\left( \frac{nK^2\maxrho}{n_1n_2w\beta} n_1n_2(1-w) \rho+ \frac{nK^2\maxrho}{n_1n_2w\beta} n_1n_2(1-w) \rho \right.\\
			&+\left. \left(\frac{nK^2\maxrho}{n_1n_2w\beta} \right)^2 n_1n_2(1-w) \right)\\
			\leq& C\left(\frac{nK^2\rho(1-w)}{w\beta}\maxrho+\frac{n^2K^4(1-w)}{n_1n_2w^2\beta^2}\maxrho^2\right);
		\end{align*}
		\begin{align*}
			|\mathcal{II}| 
			\leq& \sum_{(k_1, k_2) \ne (l_1, l_2)} \sum_{(i,j) \in \mathcal{T}_{k_1, k_2, l_1, l_2}} 
			(2A_{ij} + \hat{B}_{k_1 k_2} + B_{l_1 l_2})(\hat{B}_{k_1 k_2} + B_{l_1 l_2})\\
			=& \sum_{(k_1, k_2) \ne (l_1, l_2)} \sum_{(i,j) \in \mathcal{T}_{k_1, k_2, l_1, l_2}} 
			(2A_{ij} + B_{k_1 k_2} + B_{l_1 l_2} + (\hat{B}_{k_1 k_2} - B_{k_1 k_2}))(B_{k_1 k_2} + B_{l_1 l_2} \\
			&+ (\hat{B}_{k_1 k_2} - B_{k_1 k_2}))\\
			\leq& C \left( \frac{nK(1-w)}{\rho w}\maxbeta\maxrho \rho^2 \right) \\
			=& C\left(\frac{nK\rho(1-w)}{w}\maxbeta\maxrho\right).
		\end{align*}
		Combining the above two formulas, we have with probability larger than $1-3K_1^2K_2^2n^{-\gamma}$,
		$$\ell_{K_1,K_2}(A, \mathcal{E}^c) - \ell_0(A, \mathcal{E}^c)\leq C\left(\frac{nK^2\rho(1-w)}{w\beta}\maxrho\right),$$
		as long as $\nmin\rho w\beta\gg K^2$ and $\min\{n_1,n_2\}w\beta\gg K^3$.
	\end{proof}
	
	We want to prove that 
	$$\sum_{(K_1',K_2')\neq(K_1,K_2)}\mathrm{Pr}\left[L_{K_1',K_2'}(A,\mathcal{E}^c)<L_{K_1,K_2}(A,\mathcal{E}^c)\right]\to0.$$
	We will split the space $\{(K_1',K_2')\neq(K_1,K_2)\}$ into several subspaces and control the sum of probability on each subspace.
	
	\subsubsection{Case $K_1'<K_1$}\label{proof:region1}
	\begin{lemmaA}\label{lemma_underfitindex1}
		Assume $KK_2^2\log K\ll\nmin(1-w)$. Suppose we cluster the first side into $K_1'< K_1$ communities and cluster the second side into $K_2'\leq2K_1K_2$ communities. Define $\mathcal I_{k_1 k_2} = (\mathcal G_{k_1} \times \mathcal H_{k_2})\, \cap\, \mathcal{E}^c$ and $\hat{\mathcal I}_{k_1 k_2} = (\hat{\mathcal G}_{k_1} \times \hat{\mathcal H}_{k_2})\, \cap\, \mathcal{E}^c$. Then, for some $\alpha>0$ depending on $\pi_0$, with probability larger than $1-\exp(-\alpha n_1n_2(1-w)^2/K_1^3K_2^2)$, there must exist $l_1, l_1'\in [K_1]$ and $k_1 \in [K_1']$, with $l_2\in[K_2]$ and $k_2\in[K_2']$ and a constant $c$ that only depends on $\pi_0$, such that
		\begin{enumerate}
			\setlength{\itemsep}{0.05cm}
			\item $|\hat{\mathcal I}_{k_1 k_2} \cap \mathcal I_{l_1l_2}| \geq cn_1n_2(1-w)/K_1^3K_2^2$
			\item $|\hat{\mathcal I}_{k_1k_2} \cap \mathcal I_{l_1'l_2}| \geq cn_1n_2(1-w)/K_1^3K_2^2$.
			\item $|B_{0,l_1l_2}-B_{0,l_1'l_2}|\geq d$.
		\end{enumerate}
	\end{lemmaA}
	\begin{proof}[Proof of Lemma \ref{lemma_underfitindex1}]
		We first prove a uniform bound on the test sample size as follows. Consider two subsets $S_i \subset \mathcal{G}_i$ and $T_j \subset \mathcal{H}_j$ with $|S_i| = n_{S_i} \geq \frac{\pi_0n_1}{K_1^2}$ and $|T_j| = n_{T_j} \geq \frac{\pi_0 n_2}{2K_1K_2^2}$, where $\pi_0$ is the constant in Assumption 1. We know that the cardinality of the test set within $S_i \times T_j$, given by $|(S_i \times T_j) \cap \mathcal{E}^c|$, is Binomial$(n_{S_i} n_{T_j}, 1 - w)$.
		
		Thus by Hoeffding's inequality, we have
		\begin{align*}
			\mathrm{Pr}\left( |(S_i \times T_j) \cap \mathcal{E}^c| \geq \frac{\pi_0^2 n_1n_2 (1 - w)}{4K_1^3K_2^2} \right) 
			&\geq \mathrm{Pr}\left( |(S_i \times T_j) \cap \mathcal{E}^c| \geq \frac{n_{S_i} n_{T_j} (1 - w)}{2} \right) \\
			&\geq 1 - 2 \exp\left( - \frac{1}{2} n_{S_i} n_{T_j} (1 - w)^2 \right) \\
			&\geq 1 - 2 \exp\left( - \frac{ \pi_0^2 n_1n_2 (1 - w)^2}{4K_1^3K_2^2} \right)
		\end{align*}
		Taking $c = \frac{\pi_0^2}{4}$ gives
		\[
		\mathrm{Pr}\left( |(S_i \times T_j) \cap \mathcal{E}^c| \geq c \frac{n_1n_2(1-w)}{K_1^3K_2^2} \right) \geq 1 - 2 \exp\left(-c\frac{n_1n_2(1-w)^2}{K_1^3K_2^2}\right)
		\]
		
		To obtain a uniform bound for all $i\in [K],j\in[K_2]$ and all subsets satisfying $|S_i| = n_{S_i} \geq \frac{\pi_0 n_1}{K_1^2}$ and $|T_j| = n_{T_j} \geq \frac{\pi_0 n_2}{2K_1K_2^2}$, we apply the union bound across all such sets. The number of terms in the sum is bounded above by
		\[
		K_1K_2 \sum_{m_1\in [\frac{\pi_0 n_1}{K_1^2}, (1 - \pi_0)n_1]} \sum_{m_2\in [\frac{\pi_0 n_2}{2K_1K_2^2}, (1 -\pi_0)n_2]}\binom{n_1}{m_1} \binom{n_2}{m_2}.
		\]
		
		By Sterling's approximation, $\log\binom{n_1}{m_1} = \Theta(\ell_{dev}(m_1/n_1)n_1)$ where $\ell_{dev}$ is the binomial deviance. Notice that $\ell_{dev}(\pi_0/K_1^2)\asymp\log K_1/K_1^2$, thus the number of terms can be bounded by 
		$$K_1K_2 n_1n_2 \exp\left(c'\left(\frac{n_1\log K_1}{K_1^2}+\frac{n_2\log K_2}{K_1K_2^2}\right)\right),$$
		where $c'$ is a constant depending on $\pi_0$. Thus by the union bound, as long as $KK_2^2\log K\ll\nmin(1-w)$, with probability at least
		\begin{align*}
			&1 - 2K_1K_2 n_1n_2 \exp\left(c'\left(\frac{n_1\log K_1}{K_1^2}+\frac{n_2\log K_2}{K_1K_2^2}\right)-c\frac{n_1n_2(1-w)}{K_1^3K_2^2}\right)\\
			\leq&1-\exp\left(-\gamma\frac{n_1n_2(1-w)}{K_1^3K_2^2}\right),
		\end{align*}
		for any $i\in [K],j\in[K_2]$ and any $S_i\subset\mathcal{G}_i$ and $T_j\subset\mathcal{H}_j$ with $|S_i| = n_{S_i} \geq \frac{\pi_0 n_1}{K_1^2}$ and $|T_j| = n_{T_j} \geq \frac{\pi_0 n_2}{2K_1K_2^2}$, we have
		$$|(S_i \times T_j) \cap \mathcal{E}^c|\geq c\frac{n_1n_2(1-w)}{K_1^3K_2^2}.$$
		Now, for each $l\in[K_1]$, there must exist $\hat{l}\in[K_1']$ such that
		$$|\hat{\mathcal{G}}_{\hat{l}}\,\cap\,\mathcal{G}_l|\geq\frac{|\mathcal{G}_l|}{K_1'}\geq\frac{\pi_0n_1}{K_1^2}.$$
		By the pigeonhole principle, there must exist $l_1$ and $l_1'$ such that $\hat{l}_1=\hat{l}_1':=k_1$. Since each row of $B_0$ are distinct and the $\ell_{\infty}$ norm is lower bounded by $d$, therefore there exists $l_2\in[K_2]$ such that $|B_{0,l_1l_2}-B_{0,l_1'l_2}|\geq d$. Therefore, there exists $k_2\in[K_2']$ such that
		$$|\hat{\mathcal{H}}_{k_2}\,\cap\,\mathcal{H}_{l_2}|\geq\frac{|\mathcal{H}_{l_2}|}{K_2'}\geq\frac{\pi_0n_2}{2K_1K_2^2}.$$
		Then such $k_1,k_2,l_1,l_1'$ and $l_2$ satisfies the assertion.
	\end{proof}
	
	\begin{propositionA}\label{prop_underfit1}
		For all $K_1'<K_1$ and $K_2'\leq2K_1K_2$, we have with probability larger than $1-n^{-\gamma}-\exp(-\alpha n_1n_2(1-w)^2/K_1^3K_2^2)-2K_1^3K_2^2n^{-6}$, $$\ell_{K_1',K_2'}(A,\mathcal{E}^c)-\ell_{0}(A,\mathcal{E}^c)\geq C\left(\frac{n_1n_2d^2\rho^2(1-w)}{K_1^3K_2^2}\right),$$
		as long as $n_1n_2d^2\rho^2(1-w)\gg K^6K_2^4\log n$ and $n_1n_2d^4\rho(1-w)\gg K^7K_2^4$.
	\end{propositionA}
	\begin{proof}[Proof of Proposition \ref{prop_underfit1}]
		Without loss of generalization, we assume that in Lemma \ref{lemma_underfitindex1}, $l_1=1,l_1'=2,l_2=1$ and $k_1=1$, $k_2=1$. We define several set of entries:
		$$\mathcal{Q}_{k_1,k_2,l_1,l_2}=\{(i,j):\hat{c}_{1,i}=k_1,c_{1,i}=l_1,\hat{c}_{2,j}=k_2,c_{2,j}=l_2\},$$
		$$\mathcal{U}_{k_1,k_2,l_1,l_2}=\{(i,j)\in\mathcal{E}:\hat{c}_{1,i}=k_1,c_{1,i}=l_1,\hat{c}_{2,j}=k_2,c_{2,j}=l_2\},$$
		$$\mathcal{T}_{k_1,k_2,l_1,l_2}=\{(i,j)\in\mathcal{E}^c:\hat{c}_{1,i}=k_1,c_{1,i}=l_1,\hat{c}_{2,j}=k_2,c_{2,j}=l_2\}.$$
		Then we have  
		\begin{align*}
			&\ell_{K_1',K_2'}(A,\mathcal{E}^c)-\ell_{0}(A,\mathcal{E}^c) \\
			\geq& \sum_{(i,j) \in \mathcal{T}_{1,1,1,1}} \left[(A_{ij}-\hat{p})^2- (A_{ij}-B_{11})^2 \right]+\sum_{(i,j) \in \mathcal{T}_{1,1,2,1}} \left[ (A_{ij}-\hat{p})^2-(A_{ij}-B_{21})^2 \right]\\
			&+ \sum_{{(k_1, k_2, l_1, l_2) \notin \{(1,1,1,1), (1,1,2,1)\}}} 
			\sum_{(i,j) \in \mathcal{T}_{k_1, k_2, l_1, l_2}} 
			\left[(A_{ij}-\hat{p}_{k_1,k_2,l_1,l_2})^2-(A_{ij}-B_{l_1 l_2})^2\right]\\
			:=& \mathcal{I} + \mathcal{II} + \mathcal{III},
		\end{align*}
		where $\hat{p}$ is the average of $A_{ij}$ over $\mathcal{T}_{1,1,1,1} \cup \mathcal{T}_{1,1,2,1}$ and $\hat{p}_{k_1,k_2,l_1,l_2}$ is the average of $A_{ij}$ over $\mathcal{T}_{k_1,k_2,l_1,l_2}$. Here $\hat{p}=t\hat{p}_1+(1-t)\hat{p}_2$, where $\hat{p}_1=\hat{p}_{1,1,1,1}$ and $\hat{p}_2=\hat{p}_{1,1,2,1}$, with $t=\frac{|\mathcal{T}_{1,1,1,1}|}{|\mathcal{T}_{1,1,1,1}|+|\mathcal{T}_{1,1,2,1}|}.$ We have $|\mathcal{T}_{1,1,1,1}|\sim|\mathcal{T}_{1,1,2,1}|$, so $t$ and $1-t$ are of constant order. Denote correspondingly $p_1=B_{11}$, $p_2=B_{21}$, and denote $h(p):=\sum_{(i,j)\in\mathcal{T}_{1,1,1,1}}(A_{ij}-p)^2$. We have
		$$\mathcal{I}=h(t\hat{p}_1+(1-t)\hat{p}_2)-h(p_1)\geq|\mathcal{T}_{1,1,1,1}|(1-t)^2|\hat{p}_1-\hat{p}_2|^2+h(\hat{p}_1)-h(p_1).$$
		By Lemma \ref{lemma_underfitindex1}, $|\mathcal{T}_{1,1,1,1}|,|\mathcal{T}_{1,1,2,1}|\geq cn_1n_2(1-w)/K_1^2K_2$, so by Bernstein inequality, with probability larger than $1-n^{-\gamma}$, 
		$$|\hat{p}_1-\hat{p}_2|\geq|p_1-p_2|-|\hat{p}_1-p_1|-|\hat{p}_2-p_2|\geq|p_1-p_2|-2C\left(\sqrt{\frac{\rho K_1^3K_2^2\log n}{n_1n_2(1-w)}}\right)\geq Cd\rho$$
		for some constant $C$ depending on $\gamma$ as long as $n_1n_2d^2\rho(1-w)\gg K^3K_2^2\log n$. Similarly as the proof of Proposition \ref{prop_sbmupper} but with bound of $|\hat{p}_1-p_1|$ becomes $O_{\mathrm{P}}\left(\sqrt{\frac{\rho K^3K_2^2\log n}{n_1n_2(1-w)}}\right)$, also notice that $|\mathcal{T}_{1,1,1,1}|\leq cn_1n_2(1-w)/KK_2$, we get
		$$|h(\hat{p}_1)-h(p_1)|=O_{\mathrm{P}}\left(\sqrt{{n_1n_2(1-w)K}}\rho^{3/2}\right).$$
		Thus combining the above formulas, we have as long as $n_1n_2d^4\rho(1-w)\gg K^7K_2^4$, with probability larger than $1-n^{-\gamma}-\exp(-\alpha n_1n_2(1-w)^2/K_1^3K_2^2)$,
		$$\mathcal{I}\geq C\left(\frac{n_1n_2d^2\rho^2(1-w)}{K_1^3K_2^2}\right).$$
		Similarly with probability larger than $1-n^{-\gamma}-\exp(-\alpha n_1n_2(1-w)^2/K_1^3K_2^2)$, $\mathcal{II}$ has the same lower bound. Moreover,
		\begin{align*}
			\mathcal{III}=&-\sum_{{(k_1, k_2, l_1, l_2) \notin \{(1,1,1,1), (1,1,2,1)\}}}|\mathcal{T}_{k_1, k_2, l_1, l_2}|(\hat{p}_{k_1, k_2, l_1, l_2}-B_{l_1l_2})^2
		\end{align*}
		Here $K_1,K_2$ may diverge to $n$, so we need a more careful control of the probability since this is an infinite sum. We have
		\begin{align*}
			&\mathrm{Pr}\left[\mathcal{III}\leq-3K_1K_2K_1'K_2'\log n\right]\\
			\leq&\mathrm{Pr}\left[\sum_{k_1,k_2,l_1,l_2}|\mathcal{T}_{k_1,k_2,l_1,l_2}|\left(\hat{p}_{k_1,k_2,l_1,l_2}-B_{l_1l_2}\right)^2\geq 3K_1K_2K_1'K_2'\log n\right]
		\end{align*}
		Now we shall estimate the bound of $\sum_{k_1,k_2,l_1,l_2}|\mathcal{T}_{k_1,k_2,l_1,l_2}|\left(\hat{p}_{k_1,k_2,l_1,l_2}-B_{l_1l_2}\right)^2$. For each $\{k_1,k_2,l_1,l_2\}$, by Hoeffding's inequality, we have
		$$\mathrm{Pr}\left[|\hat{p}_{k_1,k_2,l_1,l_2}-B_{l_1l_2}|\geq \sqrt{\frac{3\log n}{|\mathcal{T}_{k_1,k_2,l_1,l_2}|}}\right]
		\leq 2\exp\left(-\frac{6\log n|\mathcal{T}_{k_1,k_2,l_1,l_2}|}{|\mathcal{T}_{k_1,k_2,l_1,l_2}|}\right)\leq 2n^{-6},$$
		i.e.,
		$$\mathrm{Pr}\left[|\mathcal{T}_{k_1,k_2,l_1,l_2}|\left(\hat{p}_{k_1,k_2,l_1,l_2}-B_{l_1l_2}\right)^2\geq 3\log n\right]\leq 2n^{-6}$$
		Therefore, by pigeonhole principle, if $\sum_{k_1,k_2,l_1,l_2}|\mathcal{T}_{k_1,k_2,l_1,l_2}|\left(\hat{p}_{k_1,k_2,l_1,l_2}-B_{l_1l_2}\right)^2\geq 3K_1K_2K_1'K_2'\log n$, then there must exists some $\{k_1,k_2,l_1,l_2\}$, such that $$|\mathcal{T}_{k_1,k_2,l_1,l_2}|\left(\hat{p}_{k_1,k_2,l_1,l_2}-B_{l_1l_2}\right)^2\geq 2\log n.$$ Thus, 
		\begin{align*}
			&\mathrm{Pr}\left[\sum_{k_1,k_2,l_1,l_2}|\mathcal{T}_{k_1,k_2,l_1,l_2}|\left(\hat{p}_{k_1,k_2,l_1,l_2}-B_{l_1l_2}\right)^2\geq3K_1K_2K_1'K_2'\log n\right]\\
			\le & \sum_{k_1,k_2,l_1,l_2}\mathrm{Pr}\left[|\mathcal{T}_{k_1,k_2,l_1,l_2}|\left(\hat{p}_{k_1,k_2,l_1,l_2}-B_{l_1l_2}\right)^2\geq3\log n\right]\\
			\leq & 2K_1K_2K_1'K_2'n^{-6}\leq 2K_1^3K_2^2n^{-6}.
		\end{align*}
		Combining all the results, we have with probability larger than $1-n^{-\gamma}-\exp(-\alpha n_1n_2(1-w)^2/K_1^3K_2^2)-2K_1^3K_2^2n^{-6}$, $$\ell_{K_1',K_2'}(A,\mathcal{E}^c)-\ell_{0}(A,\mathcal{E}^c)\geq C\left(\frac{n_1n_2d^2\rho^2(1-w)}{K_1^3K_2^2}\right),$$
		as long as $n_1n_2d^2\rho^2(1-w)\gg K^6K_2^4\log n$.
	\end{proof}
	\begin{remark}
		Here we do not use Bernstein's inequality to obtain a lower bound for $|\hat{p}_{k_1,k_2,l_1,l_2}-B_{l_1,l_2}|$ similarly as in the traditional unipartite and fixed $K$ case, since that we cannot ensure that $\rho|\mathcal{T}_{k_1,k_2,l_1,l_2}|=\Omega(1)$ for any $k_1,k_2,l_1,l_2$.
	\end{remark}
	\begin{propositionA}\label{prop_overfit1}
		When $K_1'<K_1$ and $K_2'>2K_1K_2$, we have $$\mathrm{Pr}\left[\ell_{K_1',K_2'}(A,\mathcal{E}^c)-\ell_{0}(A,\mathcal{E}^c)\leq-3K_1K_2K_1'K_2'\log n \right]\leq\frac{2K_1^2K_2}{n^5}.$$
	\end{propositionA}
	\begin{proof}[Proof of Proposition \ref{prop_overfit1}]
		Just follow the lines of the estimation of $\mathcal{III}$ in the above lemma, changing the bound for $K_2'$ to $n$.
	\end{proof}
	
	\begin{propositionA}\label{prop_probcontrol1}
		As long as $n\rho K=\Omega(K_2\log n)$, $\nmin\rho w\beta\gg K^2$, $\min\{n_1,n_2\}w\beta\gg K^3$, $\nmin wd^2\rho\beta\gg K^5K_2^2$ and $\nmin wd^2\beta\gg K^6K_2^2$, we have
		$$\sum_{K_1'<K_1}\sum_{K_2'\in[n]}\mathrm{Pr}\left[L_{K_1,K_2}(A,\mathcal{E}^c)>L_{K_1',K_2'}(A,\mathcal{E}^c)\right]\to0$$
	\end{propositionA}
	\begin{proof}[Proof of Proposition \ref{prop_probcontrol1}]
		Denote the event
		$$\Omega:=\left\{\frac{n_1n_2(1-w)}{2}\leq|\mathcal{E}^c|\leq2n_1n_2(1-w)\right\}.$$
		It is easy to show that cardinality $|\mathcal{E}^c|$ follows a Binomial distribution $\mathcal B(n_1n_2,1-w)$, thus by Hoeffding's inequality, we have
		$$\mathrm{Pr}\left[|\mathcal{E}^c|\leq\frac{n_1n_2(1-w)}{2}\right]\leq\exp\left(-\frac{n_1n_2(1-w)^2}{2}\right).$$
		Similarly, we have
		$$\mathrm{Pr}\left[|\mathcal{E}^c|\geq2n_1n_2(1-w)\right]\leq\exp\left(-2n_1n_2(1-w)^2\right).$$
		Therefore, as $n(1-w)\to\infty,$ $\mathrm{Pr}(\Omega)\to1$. From now on, we shall analyze under the condition that event $\Omega$ holds.
		
		Therefore, by Proposition \ref{prop_sbmupper} and \ref{prop_underfit1}, we have for $K_1'<K_1,K_2'\leq2K_1K_2$,
		\begin{align*}
			&\mathrm{Pr}\left[L_{K_1,K_2}(A,\mathcal{E}^c)>L_{K_1',K_2'}(A,\mathcal{E}^c)\right]\\
			\leq&\mathrm{Pr}\left[\left(\ell_{K_1',K_2'}(A,\mathcal{E}^c)-\ell_{0}(A,\mathcal{E}^c)\right)-\left(\ell_{K_1,K_2}(A,\mathcal{E}^c)-\ell_{0}(A,\mathcal{E}^c)\right)\right.\\
			&\left.<1/2(K_1K_2-K_1'K_2')n_1n_2(1-w)\lambda_{n_1,n_2}\right]\\
			\leq&4K_1^2K_2^2n^{-\gamma}+\exp(-\alpha n_1n_2(1-w)^2/K_1^3K_2^2)+2K_1^3K_2^2n^{-6}+\Pr\left[\lambda_{n_1,n_2}>\frac{\rho^2}{K_1^4K_2^3}\right]
		\end{align*}
		as long as $\nmin wd^2\rho\beta\gg K^5K_2^2$ and $\nmin wd^2\beta\gg K^6K_2^2$. Therefore, as long as $K_1^2K_2\Pr[\lambda_{n_1,n_2}>\rho^2/K_1^4K_2^3]\to0$, we have
		\begin{equation}\label{eqn:control1}
			\sum_{K_1'<K_1,K_2'\leq2K_1K_2}\mathrm{Pr}\left[L_{K_1,K_2}(A,\mathcal{E}^c)>L_{K_1',K_2'}(A,\mathcal{E}^c)\right]\to 0.
		\end{equation}
		
		Moreover, when $K_2'>2K_1K_2$, we have
		\begin{align*}
			&\mathrm{Pr}\left[L_{K_1,K_2}(A,\mathcal{E}^c)>L_{K_1',K_2'}(A,\mathcal{E}^c)\right]\\
			\leq&\mathrm{Pr}\left[\left(\ell_{K_1,K_2}(A,\mathcal{E}^c)-\ell_{0}(A,\mathcal{E}^c)\right)-\left(\ell_{K_1',K_2'}(A,\mathcal{E}^c)-\ell_{0}(A,\mathcal{E}^c)\right)\right.\\
			&\left.>1/2(K_1'K_2'-K_1K_2)n_1n_2(1-w)\lambda_{n_1,n_2}\right]\\
			\leq&\mathrm{Pr}\left[\left(\ell_{K_1,K_2}(A,\mathcal{E}^c)-\ell_{0}(A,\mathcal{E}^c)\right)-\left(\ell_{K_1',K_2'}(A,\mathcal{E}^c)-\ell_{0}(A,\mathcal{E}^c)\right)\right.\\
			&\left.>1/4K_1'K_2'n_1n_2(1-w)\lambda_{n_1,n_2}\right]\\
			\overset{(*)}{\leq}&\mathrm{Pr}\left[C\frac{nK^2\rho(1-w)}{w\beta}\maxrho+3KK_2K_1'K_2'\log n>1/4K_1'K_2'n_1n_2(1-w)\lambda_{n_1,n_2}\right]\\
			&+3K_1^2K_2^2n^{-\gamma}+\frac{2K_1^2K_2}{n^3}\\
			\leq&\mathrm{Pr}\left[C\frac{nK^2\rho(1-w)}{w\beta}\maxrho>1/8K_1K_2n_1n_2(1-w)\lambda_{n_1,n_2}\right]+3K_1^2K_2^2n^{-\gamma}\\
			&+\mathrm{Pr}\left[3K_1K_2K_1'K_2'\log n>1/8K_1'K_2'n_1n_2(1-w)\lambda_{n_1,n_2}\right]+\frac{2K_1^2K_2}{n^3},
		\end{align*}
		where (*) comes from Proposition \ref{prop_overfit1}. Notice that for any $C>0$, if we have $\mathrm{Pr}[\lambda_{n_1,n_2}<C{\rho_{n_1,n_2}K\maxrho}/{w\beta\min\{n_1,n_2\}K_2}]\ll1/n^2,$ then as $(n_1,n_2)\to\infty$, $$\mathrm{Pr}\left[C\frac{nK^2\rho(1-w)}{w\beta}\maxrho>1/8K_1K_2n_1n_2(1-w)\lambda_{n_1,n_2}\right]\ll\frac{1}{n^2},$$
		thus
		$$\sum_{K_1'<K_1,K_2'>2K_1K_2}\mathrm{Pr}\left[C\frac{n\rho(1-w)}{w\beta}>1/8K_1K_2n_1n_2(1-w)\lambda_{n_1,n_2}\right]=o(1).$$
		Similarly, since for any $C>0$, we have $\mathrm{Pr}\left[\lambda_{n_1,n_2}<{CK_1K_2\log n}/{n_1n_2(1-w)}\right]\ll1/n^2$, then as $(n_1,n_2)\to\infty$,
		$$\sum_{K_1'<K_1,K_2'>2K_1K_2}\mathrm{Pr}\left[3K_1K_2K_1'K_2'\log n>1/8K_1'K_2'n_1n_2(1-w)\lambda_{n_1,n_2}\right]=o(1).$$
		Therefore, as $(n_1,n_2)\to\infty$,
		\begin{align*}
			&\sum_{K_1'<K_1,K_2'>2K_1K_2}\mathrm{Pr}\left[L_{K_1,K_2}(A,\mathcal{E}^c)>L_{K_1',K_2'}(A,\mathcal{E}^c)\right]\\
			\leq& o(1)+\sum_{K_2'>2K_1K_2}K_1\frac{2K_1^2K_2}{n^5}+\sum_{K_2'>2K_1K_2}3K_1^3K_2^2n^{-\gamma}\\
			\leq& o(1)+\frac{2K_1^3K_2}{n^4}+3K_1^3K_2^2n^{-\gamma+1}\to0,
		\end{align*}
		as long as we choose an appropriately large $\gamma$. Together with \eqref{eqn:control1}, we can get the conclusion.
	\end{proof}

	\subsubsection{Case $K_1'=K_1$}\label{proof:region2}
	We use a similar argument as in the previous part.
	\begin{lemmaA}\label{lemma_underfitindex2}
		Assume $K_2^2\log K\ll\nmin(1-w)$. Suppose we cluster the first side into $K_1$ communities that satisfies Proposition \ref{prop_labelcons} and the second side into $K_2'< K_2$ communities. Define $\mathcal I_{k_1 k_2} = (\mathcal G_{k_1} \times \mathcal H_{k_2})\, \cap\, \mathcal{E}^c$ and $\hat{\mathcal I}_{k_1 k_2} = (\hat{\mathcal G}_{k_1} \times \hat{\mathcal H}_{k_2})\, \cap\, \mathcal{E}^c$. Then, for some $\alpha>0$ depending on $\pi_0$, with probability larger than $1-\exp(-\alpha n_1n_2(1-w)^2/K_1K_2^2)$, there must exist $l_1, l_2\in [K_2]$ and $l_0 \in [K_2']$, such that for all $k\in[K_1]$,
		\begin{enumerate}
			\setlength{\itemsep}{0.05cm}
			\item $|\hat{\mathcal I}_{k l_0} \cap \mathcal I_{k l_1}| \geq cn_1n_2(1-w)/K_1K_2^2$
			\item $|\hat{\mathcal I}_{k l_0} \cap \mathcal I_{k l_2}| \geq cn_1n_2(1-w)/K_1K_2^2$.
		\end{enumerate}
	\end{lemmaA}
	\begin{proof}[Proof of Lemma \ref{lemma_underfitindex2}]
		We first prove a uniform bound on the test sample size as follows. Consider two subsets $S_i \subset \mathcal{G}_i$ and $T_j \subset \mathcal{H}_j$ with $|S_i| = n_{S_i} \geq \frac{\pi_0n_1}{2K_1}$ and $|T_j| = n_{T_j} \geq \frac{\pi_0 n_2}{K_2^2}$, where $\pi_0$ is the constant in Assumption 1. We know that the cardinality of the test set within $S_i \times T_j$, given by $|(S_i \times T_j) \cap \mathcal{E}^c|$, is Binomial$(n_{S_i} n_{T_j}, 1 - w)$.
		
		Thus by Hoeffding's inequality, we have
		\begin{align*}
			\mathrm{Pr}\left( |(S_i \times T_j) \cap \mathcal{E}^c| \geq \frac{\pi_0^2 n_1n_2 (1 - w)}{4K_1K_2^2} \right) 
			&\geq \mathrm{Pr}\left( |(S_i \times T_j) \cap \mathcal{E}^c| \geq \frac{n_{S_i} n_{T_j} (1 - w)}{2} \right) \\
			&\geq 1 - 2 \exp\left( - \frac{1}{2} n_{S_i} n_{T_j} (1 - w)^2 \right) \\
			&\geq 1 - 2 \exp\left( - \frac{ \pi_0^2 n_1n_2 (1 - w)^2}{4K_1K_2^2} \right)
		\end{align*}
		Still taking $c = \frac{\pi_0^2}{4}$ gives
		\[
		\mathrm{Pr}\left( |(S_i \times T_j) \cap \mathcal{E}^c| \geq c \frac{n_1n_2 (1 - w)}{4K_1K_2^2} \right) \geq 1 - 2 \exp\left(-c\frac{n_1n_2 (1 - w)}{4K_1K_2^2}\right)
		\]
		
		To obtain a uniform bound for all $i\in [K_1],j\in[K_2]$ and all subsets satisfying $|S_i| = n_{S_i} \geq \frac{\pi_0 n_1}{2K_1}$ and $|T_j| = n_{T_j} \geq \frac{\pi_0 n_2}{K_2^2}$, we apply the union bound across all such sets. The number of terms in the sum is bounded above by
		\[
		K_1K_2 \sum_{m_1\in [\frac{\pi_0 n_1}{2K_1}, (1 - \pi_0)n_1]} \sum_{m_2\in [\frac{\pi_0 n_2}{K_2^2}, (1 - \pi_0)n_2]}\binom{n_1}{m_1} \binom{n_2}{m_2}.
		\]
		
		By Sterling's approximation, $\log\binom{n_1}{m_1} = \Theta(\ell_{dev}(m_1/n_1)n_1)$ where $\ell_{dev}$ is the binomial deviance. Notice that $\ell_{dev}(\pi_0/K_1)\asymp\log K_1/K_1$, thus the number of terms can be bounded by $$KK_2 n_1n_2 \exp\left(c' \left(\frac{n_1\log K_1}{K_1}+\frac{n_2\log K_2}{K_2^2}\right)\right),$$ where $c'$ is a constant depending on $\pi_0$. Thus by the union bound, as long as $K_2^2\log K\ll\min\{n_1,n_2\}(1-w)$, with probability at least
		\begin{align*}
			&1 - 2K_1K_2 n_1n_2 \exp\left(c' \left(\frac{n_1\log K_1}{K_1}+\frac{n_2\log K_2}{K_2^2}\right)-c\frac{n_1n_2(1-w)}{K_1K_2^2}\right)\\
			\leq&1-\exp\left(-\gamma\frac{n_1n_2(1-w)}{K_1K_2^2}\right),
		\end{align*}
		for any $i\in [K],j\in[K_2]$ and any $S_i\subset\mathcal{G}_i$ and $T_j\subset\mathcal{H}_j$ with $|S_i| = n_{S_i} \geq \frac{\pi_0 n_1}{2}$ and $|T_j| = n_{T_j} \geq \frac{\pi_0 n_2}{K_2}$, we have
		$$|(S_i \times T_j) \cap \mathcal{E}^c|\geq c\frac{n_1n_2(1-w)}{K_1K_2^2}.$$
		Now, for each $l\in[K_2]$, there must exist $\hat{l}\in[K_2']$ such that
		$$|\hat{\mathcal{H}}_{\hat{l}}\,\cap\,\mathcal{H}_l|\geq\frac{|\mathcal{H}_l|}{K_2'}\geq\frac{\pi_0n_2}{K_2^2},$$
		and by Proposition \ref{prop_labelcons}, with probability tending to 1, we have for all $k\in[K]$,
		$$|\hat{\mathcal{G}}_{k}\,\cap\,\mathcal{G}_k|\geq\frac{\pi_0n_1}{2K_1}$$.
		By the pigeonhole principle, there must exist $l_1$ and $l_2$ such that $\hat{l}_1=\hat{l}_2:=l_0$. Then such $l_0,l_1,l_2$ satisfy the requirement.
	\end{proof}
	
	\begin{propositionA}\label{prop_underfit2}
		When $K_2'<K_2$, we have with probability larger than $1-n^{-\gamma}-\exp(-\alpha n_1n_2(1-w)^2/K_1K_2^2)-2K_1^2K_2^2n^{-6}$, $$\ell_{K_1,K_2'}(A,\mathcal{E}^c)-\ell_{0}(A,\mathcal{E}^c)\geq C\left(\frac{n_1n_2d^2\rho^2(1-w)}{K_1K_2^2}\right),$$
		as long as $n_1n_2d^2\rho^2(1-w)\gg K^3K_2^4\log n$ and $n_1n_2d^4\rho(1-w)\gg KK_2^4$.
	\end{propositionA}
	
	\begin{proof}[Proof of Proposition \ref{prop_underfit2}]
		Without loss of generalization, we assume that in Lemma \ref{lemma_underfitindex2}, $l_0=1,l_1=1$ and $l_2=2$. Then there exists $k\in[K]$, such that $|B_{0,k1}- B_{0,k2}|\geq d$. Also, without loss of generalization, we can assume such $k=1$. Then similarly as the proof of Proposition \ref{prop_underfit1}, using the same notation, we have  
		\begin{align*}
			&\ell_{K_1,K_2'}(A,\mathcal{E}^c)-\ell_{0}(A,\mathcal{E}^c) \\
			\geq& \sum_{(i,j) \in \mathcal{T}_{1,1,1,1}} \left[(A_{ij}-\hat{p})^2- (A_{ij}-B_{11})^2 \right]+\sum_{(i,j) \in \mathcal{T}_{1,1,1,2}} \left[ (A_{ij}-\hat{p})^2-(A_{ij}-B_{12})^2 \right]\\
			&+ \sum_{{(k_1, k_2, l_1, l_2) \notin \{(1,1,1,1), (1,1,1,2)\}}} 
			\sum_{(i,j) \in \mathcal{T}_{k_1, k_2, l_1, l_2}} 
			\left[(A_{ij}-\hat{p}_{k_1,k_2,l_1,l_2})^2-(A_{ij}-B_{l_1 l_2})^2\right]\\
			:=& \mathcal{I} + \mathcal{II} + \mathcal{III},
		\end{align*}
		where $\hat{p}$ is the average of $A_{ij}$ over $\mathcal{T}_{1,1,1,1} \cup \mathcal{T}_{1,1,1,2}$ and $\hat{p}_{k_1,k_2,l_1,l_2}$ is the average of $A_{ij}$ over $\mathcal{T}_{k_1,k_2,l_1,l_2}$. Here $\hat{p}=t\hat{p}_1+(1-t)\hat{p}_2$, where $\hat{p}_1=\hat{p}_{1,1,1,1}$ and $\hat{p}_2=\hat{p}_{1,1,1,2}$, with $t=\frac{|\mathcal{T}_{1,1,1,1}|}{|\mathcal{T}_{1,1,1,1}|+|\mathcal{T}_{1,1,1,2}|}.$ We have $|\mathcal{T}_{1,1,1,1}|\sim|\mathcal{T}_{1,1,1,2}|$, so $t$ and $1-t$ are of constant order. Denote correspondingly $p_1=B_{11}$, $p_2=B_{12}$, and denote $h(p):=\sum_{(i,j)\in\mathcal{T}_{1,1,1,1}}(A_{ij}-p)^2$. We have
		$$\mathcal{I}=h(t\hat{p}_1+(1-t)\hat{p}_2)-h(p_1)\geq|\mathcal{T}_{1,1,1,1}|(1-t)^2|\hat{p}_1-\hat{p}_2|^2+h(\hat{p}_1)-h(p_1).$$
		By Lemma \ref{lemma_underfitindex2}, $|\mathcal{T}_{1,1,1,1}|,|\mathcal{T}_{1,1,1,2}|\geq cn_1n_2(1-w)/KK_2^2$, so by Bernstein inequality, with probability larger than $1-n^{-\gamma}$,
		$$|\hat{p}_1-\hat{p}_2|\geq|p_1-p_2|-|\hat{p}_1-p_1|-|\hat{p}_2-p_2|\geq|p_1-p_2|-2C\left(\sqrt{\frac{\rho KK_2^2\log n}{n_1n_2(1-w)}}\right)\geq Cd\rho$$
		for some constant $C$ depending on $\gamma$ as long as $n_1n_2d^2\rho(1-w)\gg KK_2^2\log n$. Similarly as the proof of Proposition \ref{prop_sbmupper} but with bound of $|\hat{p}_1-p_1|$ becomes $O_{\mathrm{P}}\left(\sqrt{\frac{\rho KK_2^2\log n}{n_1n_2(1-w)}}\right)$, we get
		$$|h(\hat{p}_1)-h(p_1)|=O_{\mathrm{P}}\left(\sqrt{\frac{n_1n_2(1-w)}{K}}\rho^{3/2}\right).$$
		Thus combining the above formulas, we have as long as $n_1n_2d^4\rho(1-w)\gg KK_2^4$, with probability larger than $1-n^{-\gamma}-\exp(-\alpha n_1n_2(1-w)^2/K_1K_2^2)$,
		$$\mathcal{I}\geq C\left(\frac{n_1n_2d^2\rho^2(1-w)}{KK_2^2}\right).$$
		Similarly Similarly with probability larger than $1-n^{-\gamma}-\exp(-\alpha n_1n_2(1-w)^2/K_1K_2^2)$, $\mathcal{II}$ has the same lower bound. Moreover, using a similar argument as in Proposition \ref{prop_underfit1} or refer to the proof as in Proposition \ref{prop_overfit3}, we have
		$$\mathcal{III}=-\sum_{{(k_1, k_2, l_1, l_2) \notin \{(1,1,1,1), (1,1,1,2)\}}}|\mathcal{T}_{k_1, k_2, l_1, l_2}|(\hat{p}_{k_1, k_2, l_1, l_2}-B_{l_1l_2})^2$$
		and 
		$$\mathrm{Pr}\left[\mathcal{III}\leq-3K_1^2K_2K_2'\log n \right]\leq\frac{2K_1^2K_2^2}{n^6}.$$
		Combining all the results, we have with probability larger than $1-n^{-\gamma}-\exp(-\alpha n_1n_2(1-w)^2/K_1K_2^2)-2K_1^2K_2^2n^{-6}$, $$\ell_{K_1,K_2'}(A,\mathcal{E}^c)-\ell_{0}(A,\mathcal{E}^c)\geq C\left(\frac{n_1n_2d^2\rho^2(1-w)}{K_1K_2^2}\right),$$
		as long as $n_1n_2d^2\rho^2(1-w)\gg K^3K_2^4\log n$.
	\end{proof}
	
	Finally we control the overfitting case.
	\begin{propositionA}\label{prop_overfit3}
		When $K_2'>K_2$, we have $$\mathrm{Pr}\left[\ell_{K_1,K_2'}(A,\mathcal{E}^c)-\ell_{0}(A,\mathcal{E}^c)\leq-3K_1^2K_2K_2'\log n \right]\leq\frac{2K_1^2K_2}{n^5}.$$
	\end{propositionA}
	\begin{proof}[Proof of Proposition \ref{prop_overfit3}]
		We use the same notation as in the proof of Proposition \ref{prop_underfit2}, and we also have
		\begin{align*}
			\ell_{K_1,K_2'}(A,\mathcal{E}^c)-\ell_{0}(A,\mathcal{E}^c)=&\sum_{k_1,k_2,l_1,l_2}\sum_{(i,j)\in\mathcal{T}_{k_1,k_2,l_1,l_2}}\left[\left(A_{ij}-\hat{P}^{(K')}_{ij}\right)^2-\left(A_{ij}-B_{l_1l_2}\right)^2\right]\\
			\geq&\sum_{k_1,k_2,l_1,l_2}\sum_{(i,j)\in\mathcal{T}_{k_1,k_2,l_1,l_2}}\left[\left(A_{ij}-\hat{p}_{k_1,k_2,l_1,l_2}\right)^2-\left(A_{ij}-B_{l_1l_2}\right)^2\right]\\
			=&-\sum_{k_1,k_2,l_1,l_2}|\mathcal{T}_{k_1,k_2,l_1,l_2}|\left(\hat{p}_{k_1,k_2l_1,l_2}-B_{l_1l_2}\right)^2.
		\end{align*}
		
		We have
		\begin{align*}
			&\mathrm{Pr}\left[\ell_{K_1,K_2'}(A,\mathcal{E}^c)-\ell_{0}(A,\mathcal{E}^c)\leq-3K_1^2K_2K_2'\log n\right]\\
			\leq&\mathrm{Pr}\left[\sum_{k_1,k_2,l_1,l_2}|\mathcal{T}_{k_1,k_2,l_1,l_2}|\left(\hat{p}_{k_1,k_2,l_1,l_2}-B_{l_1l_2}\right)^2\geq 3K_1^2K_2K_2'\log n\right]
		\end{align*}
		Now we shall estimate the bound of $\sum_{k_1,k_2,l_1,l_2}|\mathcal{T}_{k_1,k_2,l_1,l_2}|\left(\hat{p}_{k_1,k_2,l_1,l_2}-B_{l_1l_2}\right)^2$. For each $\{k_1,k_2,l_1,l_2\}$, by Hoeffding's inequality, we have
		$$\mathrm{Pr}\left[|\hat{p}_{k_1,k_2,l_1,l_2}-B_{l_1l_2}|\geq \sqrt{\frac{3\log n}{|\mathcal{T}_{k_1,k_2,l_1,l_2}|}}\right]
		\leq 2\exp\left(-\frac{6\log n|\mathcal{T}_{k_1,k_2,l_1,l_2}|}{|\mathcal{T}_{k_1,k_2,l_1,l_2}|}\right)\leq 2n^{-6},$$
		i.e.,
		$$\mathrm{Pr}\left[|\mathcal{T}_{k_1,k_2,l_1,l_2}|\left(\hat{p}_{k_1,k_2,l_1,l_2}-B_{l_1l_2}\right)^2\geq 2\log n\right]\leq 2n^{-4}$$
		Therefore, by pigeonhole principle, if $\sum_{k_1,k_2,l_1,l_2}|\mathcal{T}_{k_1,k_2,l_1,l_2}|\left(\hat{p}_{k_1,k_2,l_1,l_2}-B_{l_1l_2}\right)^2\geq 3K_1^2K_2K_2'\log n$, then there must exists some $\{k_1,k_2,l_1,l_2\}$, such that $|\mathcal{T}_{k_1,k_2,l_1,l_2}|\left(\hat{p}_{k_1,k_2,l_1,l_2}-B_{l_1l_2}\right)^2\geq 3\log n$. Thus, 
		\begin{align*}
			&\mathrm{Pr}\left[\sum_{k_1,k_2,l_1,l_2}|\mathcal{T}_{k_1,k_2,l_1,l_2}|\left(\hat{p}_{k_1,k_2,l_1,l_2}-B_{l_1l_2}\right)^2\geq3K_1^2K_2K_2'\log n\right]\\
			\le & \sum_{k_1,k_2,l_1,l_2}\mathrm{Pr}\left[|\mathcal{T}_{k_1,k_2,l_1,l_2}|\left(\hat{p}_{k_1,k_2,l_1,l_2}-B_{l_1l_2}\right)^2\geq3\log n\right]\\
			\leq & 2K_1^2K_2K_2'n^{-6}\leq 2K_1^2K_2n^{-5}.
		\end{align*}
	\end{proof}
	
	\begin{propositionA}\label{prop_probcontrol2}
		As long as $n\rho K=\Omega(K_2\log n)$, $\nmin\rho w\beta\gg K^2$, $\min\{n_1,n_2\}w\beta\gg K^3$, $\nmin wd^2\rho\beta\gg K^3K_2^2$ and $\nmin wd^2\beta\gg K^4K_2^2$, we have
		$$\sum_{K_2'\neq K_2}\mathrm{Pr}\left[L_{K_1,K_2}(A,\mathcal{E}^c)>L_{K_1,K_2'}(A,\mathcal{E}^c)\right]\to0$$
	\end{propositionA}
	\begin{proof}[Proof of Proposition \ref{prop_probcontrol2}]
		Recall the event
		$$\Omega=\left\{\frac{n_1n_2(1-w)}{2}\leq|\mathcal{E}^c|\leq2n_1n_2(1-w)\right\}$$
		defined in the proof of Proposition \ref{prop_probcontrol1}. Similarly, recall Proposition \ref{prop_sbmupper} gives with probability larger than $1-3K^2K_2^2n^{-\gamma}$,
		\begin{equation*}
			\ell_{K_1,K_2}(A, \mathcal{E}^c) - \ell_0(A, \mathcal{E}^c) \leq C\frac{nK^2\rho(1-w)}{w\beta}\maxrho.
		\end{equation*}
		
		We first prove the consistency when $K_2'\leq K_2$. Notice that
		\begin{align*}
			&\mathrm{Pr}\left[L_{K_1,K_2}(A,\mathcal{E}^c)>L_{K_1,K_2'}(A,\mathcal{E}^c)\right]\\
			\leq&\mathrm{Pr}\left[\left(\ell_{K_1,K_2'}(A,\mathcal{E}^c)-\ell_{0}(A,\mathcal{E}^c)\right)-\left(\ell_{K_1,K_2}(A,\mathcal{E}^c)-\ell_{0}(A,\mathcal{E}^c)\right)\right.\\
			&\left.<1/2K_1(K_2-K_2')n_1n_2(1-w)\lambda_{n_1,n_2}\right]\\
			\leq&4K_1^2K_2^2n^{-\gamma}+\exp(-\alpha n_1n_2(1-w)^2/K_1K_2^2)+2K_1^2K_2^2n^{-6}+\Pr\left[\lambda_{n_1,n_2}>\frac{\rho^2}{K_1^2K_2^3}\right]
		\end{align*}
		by Proposition \ref{prop_sbmupper} and \ref{prop_underfit2}, as long as $\nmin wd^2\rho\beta\gg K^3K_2^2$ and $\nmin wd^2\beta\gg K^4K_2^2$. Therefore, as long as $K_2\Pr[\lambda_{n_1,n_2}>\rho^2/K_1^2K_2^3]\to0$, we have
		\begin{equation*}
			\sum_{K_2'<K_2}\mathrm{Pr}\left[L_{K_1,K_2}(A,\mathcal{E}^c)>L_{K_1,K_2'}(A,\mathcal{E}^c)\right]\to 0.
		\end{equation*}
		
		Now we prove conditioned on the event $\Omega$, we have 
		\begin{equation}\label{eqn:control2}
			\sum_{K_2'>K_2}\mathrm{Pr}\left[L_{K_1,K_2}(A,\mathcal{E}^c)>L_{K_1,K_2'}(A,\mathcal{E}^c)\right]\to0.
		\end{equation}
		When $K_2'>2K_2$, we have
		\begin{align*}
			&\mathrm{Pr}\left[L_{K_1,K_2}(A,\mathcal{E}^c)>L_{K_1,K_2'}(A,\mathcal{E}^c)\right]\\
			\leq&\mathrm{Pr}\left[\left(\ell_{K_1,K_2}(A,\mathcal{E}^c)-\ell_{0}(A,\mathcal{E}^c)\right)-\left(\ell_{K_1,K_2'}(A,\mathcal{E}^c)-\ell_{0}(A,\mathcal{E}^c)\right)\right.\\
			&\left.>1/2K_1(K_2'-K_2)n_1n_2(1-w)\gamma_{n_1,n_2}\right]\\
			\leq&\mathrm{Pr}\left[\left(\ell_{K_1,K_2}(A,\mathcal{E}^c)-\ell_{0}(A,\mathcal{E}^c)\right)-\left(\ell_{K_1,K_2'}(A,\mathcal{E}^c)-\ell_{0}(A,\mathcal{E}^c)\right)\right.\\
			&\left.>1/4K_1K_2'n_1n_2(1-w)\lambda_{n_1,n_2}\right]\\
			\overset{(*)}{\leq}&\mathrm{Pr}\left[C\frac{nK^2\rho(1-w)}{w\beta}\maxrho+3K_1^2K_2K_2'\log n>1/4K_1K_2'n_1n_2(1-w)\lambda_{n_1,n_2}\right]\\
			&+3K_1^2K_2^2n^{-\gamma}+\frac{2K_1^2K_2}{n^5}\\
			\leq&\mathrm{Pr}\left[C\frac{nK^2\rho(1-w)}{w\beta}\maxrho>1/8K_1K_2n_1n_2(1-w)\lambda_{n_1,n_2}\right]+3K_1^2K_2^2n^{-\gamma}\\
			&+\mathrm{Pr}\left[3K_1^2K_2K_2'\log n>1/8K_1K_2'n_1n_2(1-w)\lambda_{n_1,n_2}\right]+\frac{2K_1^2K_2}{n^5},
		\end{align*}
		where (*) comes from Proposition \ref{prop_sbmupper} and \ref{prop_overfit3}. Notice that for any $C>0$, if we have $\mathrm{Pr}[\lambda_{n_1,n_2}<C{\rho_{n_1,n_2}K\maxrho}/{w\beta\min\{n_1,n_2\}K_2}]\ll1/n^2,$ then as $(n_1,n_2)\to\infty$, $$\mathrm{Pr}\left[C\frac{nK^2\rho(1-w)}{w\beta}\maxrho>1/8K_1K_2n_1n_2(1-w)\lambda_{n_1,n_2}\right]\ll\frac{1}{n^2},$$
		thus
		$$\sum_{K_2'>2K_2}\mathrm{Pr}\left[C\frac{nK^2\rho(1-w)}{w\beta}\maxrho>1/8K_1K_2n_1n_2(1-w)\lambda_{n_1,n_2}\right]=o(1).$$
		Similarly, since for any $C>0$, we have $\mathrm{Pr}\left[\lambda_{n_1,n_2}<{CK_1K_2\log n}/{n_1n_2(1-w)}\right]\ll1/n^2$, then as $(n_1,n_2)\to\infty$,
		$$\sum_{K_2'>2K_2}\mathrm{Pr}\left[3K_1^2K_2K_2'\log n>1/8K_1K_2'n_1n_2(1-w)\lambda_{n_1,n_2}\right]=o(1).$$
		Therefore, as $(n_1,n_2)\to\infty$,
		\begin{align*}
			&\sum_{K_2'>2K_2}\mathrm{Pr}\left[L_{K_1,K_2}(A,\mathcal{E}^c)>L_{K_1,K_2'}(A,\mathcal{E}^c)\right]\\
			\leq& o(1)+\sum_{K_2'>2K_2}\frac{2K_1^2K_2}{n^5}+\sum_{K_2'>2K_2}3K_1^2K_2^2n^{-\gamma}\\
			\leq& o(1)+\frac{2K_1^2K_2}{n^4}+3K_1^2K_2^2n^{-\gamma+1}\to0,
		\end{align*}
		as long as we choose an appropriately large $\gamma$. 
		
		Finally, we consider for $K_2<K_2'\leq 2K_2$. We have
		\begin{align*}
			&\mathrm{Pr}\left[L_{K_1,K_2}(A,\mathcal{E}^c)>L_{K_1,K_2'}(A,\mathcal{E}^c)\right]\\
			\leq&\mathrm{Pr}\left[\left(\ell_{K_1,K_2}(A,\mathcal{E}^c)-\ell_{0}(A,\mathcal{E}^c)\right)-\left(\ell_{K_1,K_2'}(A,\mathcal{E}^c)-\ell_{0}(A,\mathcal{E}^c)\right)\right.\\
			&\left.>1/2K_1(K_2'-K_2)n_1n_2(1-w)\lambda_{n_1,n_2}\right]\\
			{\leq}&\mathrm{Pr}\left[C\frac{nK^2\rho(1-w)}{w\beta}\maxrho+3K_1^2K_2K_2'\log n>1/2K_1n_1n_2(1-w)\lambda_{n_1,n_2}\right]\\
			&+3K_1^2K_2^2n^{-\gamma}+\frac{2K_1^2K_2}{n^5}\\
			\leq&\mathrm{Pr}\left[C\frac{nK^2\rho(1-w)}{w\beta}\maxrho>1/8K_1n_1n_2(1-w)\lambda_{n_1,n_2}\right]+3K_1^2K_2^2n^{-\gamma}\\
			&+\mathrm{Pr}\left[6K_1^2K_2^2\log n>1/8K_1n_1n_2(1-w)\lambda_{n_1,n_2}\right]+\frac{2K_1^2K_2}{n^5}
		\end{align*}
		by Proposition \ref{prop_sbmupper},and \ref{prop_overfit3}. Therefore, similarly as before, as long as $K_2\Pr[\lambda_{n_1,n_2}<C\rho K\maxrho/w\beta\nmin]\to0$ and $K_2\mathrm{Pr}\left[\lambda_{n_1,n_2}<{CK_1K_2^2\log n}/{n_1n_2(1-w)}\right]\to 0$, we have
		\begin{equation*}
			\sum_{K_2<K_2'\leq 2K_2}\mathrm{Pr}\left[L_{K_1,K_2}(A,\mathcal{E}^c)>L_{K_1,K_2'}(A,\mathcal{E}^c)\right]\to 0.
		\end{equation*}
		Thus \eqref{eqn:control2} is proved.
	\end{proof}
	
	\subsubsection{Case $K_1'>K_1,K_2'<K_2$}\label{proof:region3}
	Notice that we did not use the fact that $K_1=K<K_2$ in the first part of our proof (Section \ref{proof:region1}). Therefore, all the proofs in Section \ref{proof:region1} can be used in this case by exchanging the role for $K_1,K_2$ and $K_1',K_2'$. We list the symmetric version of the lemmas and propositions used in Section \ref{proof:region1}, with proofs omitted.
	\begin{lemmaA}\label{lemma_underfitindex3}
		Assume $K^2K_2\log K_2\ll\nmin(1-w)$. Suppose we cluster the first side into $K_1'\leq 2K_1K_2$ communities and cluster the second side into $K_2'<K_2$ communities. Define $\mathcal I_{k_1 k_2} = (\mathcal G_{k_1} \times \mathcal H_{k_2})\, \cap\, \mathcal{E}^c$ and $\hat{\mathcal I}_{k_1 k_2} = (\hat{\mathcal G}_{k_1} \times \hat{\mathcal H}_{k_2})\, \cap\, \mathcal{E}^c$. Then, for some $\alpha>0$ depending on $\pi_0$, with probability larger than $1-\exp(-\alpha n_1n_2(1-w)^2/K_1^2K_2^3)$, there must exist $l_2, l_2'\in [K_2]$ and $k_2 \in [K_2']$, with $l_1\in[K_2]$ and $k_1\in[K_1']$ and a constant $c$ that only depends on $\pi_0$, such that
		\begin{enumerate}
			\setlength{\itemsep}{0.05cm}
			\item $|\hat{\mathcal I}_{k_1 k_2} \cap \mathcal I_{l_1l_2}| \geq cn_1n_2(1-w)/K_1^2K_2^3$
			\item $|\hat{\mathcal I}_{k_1k_2} \cap \mathcal I_{l_1l_2'}| \geq cn_1n_2(1-w)/K_1^2K_2^3$.
			\item $|B_{0,l_1l_2}-B_{0,l_1l_2'}|\geq d$.
		\end{enumerate}
	\end{lemmaA}
	
	\begin{propositionA}\label{prop_underfit3}
		For all $K_1<K_1'\leq2K_1K_2$ and $K_2'<K_2$, we have with probability larger than $1-n^{-\gamma}-\exp(-\alpha n_1n_2(1-w)^2/K_1^2K_2^3)-2K_1^2K_2^3n^{-6}$, $$\ell_{K_1',K_2'}(A,\mathcal{E}^c)-\ell_{0}(A,\mathcal{E}^c)\geq C\left(\frac{n_1n_2d^2\rho^2(1-w)}{K_1^2K_2^3}\right),$$
		as long as $n_1n_2d^2\rho^2(1-w)\gg K^4K_2^6\log n$ and $n_1n_2d^4\rho(1-w)\gg K^4K_2^7$.
	\end{propositionA}
	
	\begin{propositionA}\label{prop_overfit4}
		When $K_1'>2K_1K_2$ and $K_2'<K_2$, we have $$\mathrm{Pr}\left[\ell_{K_1',K_2'}(A,\mathcal{E}^c)-\ell_{0}(A,\mathcal{E}^c)\leq-3K_1K_2K_1'K_2'\log n \right]\leq\frac{2K_1K_2^2}{n^5}.$$
	\end{propositionA}
	
	\begin{propositionA}\label{prop_probcontrol3}
		As long as $\nmin wd^2\rho\beta\gg K^4K_2^3$ and $\nmin wd^2\beta\gg K^5K_2^3$, we have
		$$\sum_{K_1'>K_1}\sum_{K_2'<K_2}\mathrm{Pr}\left[L_{K_1,K_2}(A,\mathcal{E}^c)>L_{K_1',K_2'}(A,\mathcal{E}^c)\right]\to0$$
	\end{propositionA}
	
	\subsubsection{Case $K_1'>K_1,K_2'=K_2$}\label{proof:region4}
	Similarly it's a symmetric version of the proof in Section \ref{proof:region2}. We also list the symmetric versions of the lemmas and propositions used in Section \ref{proof:region2} and omit the proofs.
	
	\begin{propositionA}\label{prop_overfit6}
		When $K_1'>K_1$, we have $$\mathrm{Pr}\left[\ell_{K_1',K_2}(A,\mathcal{E}^c)-\ell_{0}(A,\mathcal{E}^c)\leq-3K_1K_2^2K_1'\log n \right]\leq\frac{2K_1K_2^2}{n^5}.$$
	\end{propositionA}
	\begin{propositionA}\label{prop_probcontrol4}
		As long as $n\rho K=\Omega(K_2\log n)$, $\nmin\rho w\beta\gg K^2$ and $\min\{n_1,n_2\}w\beta\gg K^3$, we have
		$$\sum_{K_1'> K_1}\mathrm{Pr}\left[L_{K_1,K_2}(A,\mathcal{E}^c)>L_{K_1',K_2}(A,\mathcal{E}^c)\right]\to0.$$
		Here an additional assumption for $\lambda_n$ will write as $$K\Pr\left[\lambda_{n_1,n_2}<C\frac{\rho K^2}{w\beta\nmin K_2}\maxrho\right]\to0$$
		and
		$$K\Pr\left[\lambda_{n_1,n_2}<C\frac{K^2K_2\log n}{n_1n_2(1-w)}\right]\to0$$
	\end{propositionA}

	\subsubsection{Case $K_1'>K_1,K_2'>K_2$}\label{proof:region5}
	We use exactly the same argument as in Proposition \ref{prop_overfit1} to get the following proposition.
	\begin{propositionA}\label{prop_overfit8}
		When $K_1'>K_1$ or $K_2'>K_2$, we have $$\mathrm{Pr}\left[\ell_{K_1',K_2'}(A,\mathcal{E}^c)-\ell_{0}(A,\mathcal{E}^c)\leq-3K_1K_2K_1'K_2'\log n \right]\leq\frac{2K_1K_2}{n^4}.$$
	\end{propositionA}
	\begin{proof}[Proof of Proposition \ref{prop_overfit8}]
		We use the same notation as in Proposition \ref{prop_overfit3}. Notice that
		\begin{align*}
			\ell_{K_1',K_2'}(A,\mathcal{E}^c)-\ell_{0}(A,\mathcal{E}^c)=&\sum_{k_1,k_2,l_1,l_2}\sum_{(i,j)\in\mathcal{T}_{k_1,k_2,l_1,l_2}}\left[\left(A_{ij}-\hat{P}^{(K_1',K_2')}_{ij}\right)^2-\left(A_{ij}-B_{l_1l_2}\right)^2\right]\\
			\geq&\sum_{k_1,k_2,l_1,l_2}\sum_{(i,j)\in\mathcal{T}_{k_1,k_2,l_1,l_2}}\left[\left(A_{ij}-\hat{p}_{k_1,k_2,l_1,l_2}\right)^2-\left(A_{ij}-B_{l_1l_2}\right)^2\right]\\
			=&-\sum_{k_1,k_2,l_1,l_2}|\mathcal{T}_{k_1,k_2,l_1,l_2}|\left(\hat{p}_{k_1,k_2l_1,l_2}-B_{l_1l_2}\right)^2.
		\end{align*}
		
		We have
		\begin{align*}
			&\mathrm{Pr}\left[\ell_{K_1',K_2'}(A,\mathcal{E}^c)-\ell_{0}(A,\mathcal{E}^c)\leq-3K_1K_2K_1'K_2'\log n\right]\\
			\leq&\mathrm{Pr}\left[\sum_{k_1,k_2,l_1,l_2}|\mathcal{T}_{k_1,k_2,l_1,l_2}|\left(\hat{p}_{k_1,k_2,l_1,l_2}-B_{l_1l_2}\right)^2\geq 3K_1K_2K_1'K_2'\log n\right]
		\end{align*}
		Now we shall estimate the bound of $\sum_{k_1,k_2,l_1,l_2}|\mathcal{T}_{k_1,k_2,l_1,l_2}|\left(\hat{p}_{k_1,k_2,l_1,l_2}-B_{l_1l_2}\right)^2$. For each $\{k_1,k_2,l_1,l_2\}$, by Hoeffding's inequality, we have
		$$\mathrm{Pr}\left[|\hat{p}_{k_1,k_2,l_1,l_2}-B_{l_1l_2}|\geq \sqrt{\frac{3\log n}{|\mathcal{T}_{k_1,k_2,l_1,l_2}|}}\right]
		\leq 2\exp\left(-\frac{6\log n|\mathcal{T}_{k_1,k_2,l_1,l_2}|}{|\mathcal{T}_{k_1,k_2,l_1,l_2}|}\right)\leq 2n^{-6},$$
		i.e.,
		$$\mathrm{Pr}\left[|\mathcal{T}_{k_1,k_2,l_1,l_2}|\left(\hat{p}_{k_1,k_2,l_1,l_2}-B_{l_1l_2}\right)^2\geq 2\log n\right]\leq 2n^{-6}$$
		Therefore, by pigeonhole principle, if $\sum_{k_1,k_2,l_1,l_2}|\mathcal{T}_{k_1,k_2,l_1,l_2}|\left(\hat{p}_{k_1,k_2,l_1,l_2}-B_{l_1l_2}\right)^2\geq 3K_1K_2K_1'K_2'$, then there must exists some $\{k_1,k_2,l_1,l_2\}$, such that $$|\mathcal{T}_{k_1,k_2,l_1,l_2}|\left(\hat{p}_{k_1,k_2,l_1,l_2}-B_{l_1l_2}\right)^2\geq 3\log n.$$ Thus, 
		\begin{align*}
			&\mathrm{Pr}\left[\sum_{k_1,k_2,l_1,l_2}|\mathcal{T}_{k_1,k_2,l_1,l_2}|\left(\hat{p}_{k_1,k_2,l_1,l_2}-B_{l_1l_2}\right)^2\geq3K^2K_2K'\log n\right]\\
			\le & \sum_{k_1,k_2,l_1,l_2}\mathrm{Pr}\left[|\mathcal{T}_{k_1,k_2,l_1,l_2}|\left(\hat{p}_{k_1,k_2,l_1,l_2}-B_{l_1l_2}\right)^2\geq3\log n\right]\\
			\leq & 2K_1K_2K_1'K_2'n^{-6}\leq 2K_1K_2n^{-4}.
		\end{align*}
	\end{proof}
	Using the above proposition, we have
	\begin{propositionA}\label{prop_probcontrol5}
		As long as $n\rho K=\Omega(K_2\log n)$, $\nmin\rho w\beta\gg K^2$ and $\min\{n_1,n_2\}w\beta\gg K^3$, we have
		$$\sum_{K_1'>K_1}\sum_{K_2'>K_2}\mathrm{Pr}\left[L_{K_1,K_2}(A,\mathcal{E}^c)>L_{K_1',K_2'}(A,\mathcal{E}^c)\right]\to0$$
	\end{propositionA}
	\begin{proof}[Proof of Proposition \ref{prop_probcontrol5}]
		Recall the event
		$$\Omega=\left\{\frac{n_1n_2(1-w)}{2}\leq|\mathcal{E}^c|\leq2n_1n_2(1-w)\right\}$$
		defined in the proof of Proposition \ref{prop_probcontrol1}. Similarly, Similarly, recall Proposition \ref{prop_sbmupper} gives with probability larger than $1-3K^2K_2^2n^{-\gamma}$,
		\begin{equation*}
			\ell_{K_1,K_2}(A, \mathcal{E}^c) - \ell_0(A, \mathcal{E}^c) \leq C\frac{nK^2\rho(1-w)}{w\beta}\maxrho.
		\end{equation*}
		
		We first prove the consistency when $K_1<K_1'\leq 2K_1$ and $K_2<K_2'\leq2K_2$. Notice that $K_1'K_2'-K_1K_2\geq K_1$ in this scenario. Thus we have
		\begin{align*}
			&\mathrm{Pr}\left[L_{K_1,K_2}(A,\mathcal{E}^c)>L_{K_1,K_2'}(A,\mathcal{E}^c)\right]\\
			\leq&\mathrm{Pr}\left[\left(\ell_{K_1,K_2}(A,\mathcal{E}^c)-\ell_{0}(A,\mathcal{E}^c)\right)-\left(\ell_{K_1,K_2'}(A,\mathcal{E}^c)-\ell_{0}(A,\mathcal{E}^c)\right)\right.\\
			&\left.>1/2(K_1'K_2'-K_1K_2)n_1n_2(1-w)\lambda_{n_1,n_2}\right]\\
			{\leq}&\mathrm{Pr}\left[C\frac{nK^2\rho(1-w)}{w\beta}\maxrho+3K_1K_2K_1'K_2'\log n>1/2K_1n_1n_2(1-w)\lambda_{n_1,n_2}\right]\\
			&+3K_1^2K_2^2n^{-\gamma}+\frac{2K_1K_2}{n^4}\\
			\leq&\mathrm{Pr}\left[C\frac{nK^2\rho(1-w)}{w\beta}\maxrho>1/8K_1n_1n_2(1-w)\lambda_{n_1,n_2}\right]+3K_1^2K_2^2n^{-\gamma}\\
			&+\mathrm{Pr}\left[12K_1^2K_2^2\log n>1/8K_1n_1n_2(1-w)\lambda_{n_1,n_2}\right]+\frac{2K_1K_2}{n^4}
		\end{align*}
		by Proposition \ref{prop_sbmupper}, \ref{prop_overfit8}. Therefore, as long as $$K_1K_2\Pr\left[\lambda_{n_1,n_2}<C\frac{\rho K}{w\beta\nmin}\maxrho\right]\to0$$ and $$K_1K_2\mathrm{Pr}\left[\lambda_{n_1,n_2}<C\frac{K_1K_2^2\log n}{n_1n_2(1-w)}\right]\to 0$$, we have
		\begin{equation*}
			\sum_{K_1<K_1'\leq 2K_1,K_2<K_2'\leq 2K_2}\mathrm{Pr}\left[L_{K_1,K_2}(A,\mathcal{E}^c)>L_{K_1',K_2'}(A,\mathcal{E}^c)\right]\to 0.
		\end{equation*}
		
		Now we prove conditioned on the event $\Omega$, we have 
		\begin{equation}\label{eqn:control3}
			\sum_{K_1'>2K_1\text{ or }K_2'>K_2}\mathrm{Pr}\left[L_{K_1,K_2}(A,\mathcal{E}^c)>L_{K_1',K_2'}(A,\mathcal{E}^c)\right]\to0.
		\end{equation}
		When $K_1'>2K_1$ or $K_2'>K_2$, we have
		\begin{align*}
			&\mathrm{Pr}\left[L_{K_1,K_2}(A,\mathcal{E}^c)>L_{K_1,K_2'}(A,\mathcal{E}^c)\right]\\
			\leq&\mathrm{Pr}\left[\left(\ell_{K_1,K_2}(A,\mathcal{E}^c)-\ell_{0}(A,\mathcal{E}^c)\right)-\left(\ell_{K_1,K_2'}(A,\mathcal{E}^c)-\ell_{0}(A,\mathcal{E}^c)\right)\right.\\
			&\left.>1/2(K_1'K_2'-K_1K_2)n_1n_2(1-w)\gamma_{n_1,n_2}\right]\\
			\leq&\mathrm{Pr}\left[\left(\ell_{K_1,K_2}(A,\mathcal{E}^c)-\ell_{0}(A,\mathcal{E}^c)\right)-\left(\ell_{K_1,K_2'}(A,\mathcal{E}^c)-\ell_{0}(A,\mathcal{E}^c)\right)\right.\\
			&\left.>1/4K_1'K_2'n_1n_2(1-w)\lambda_{n_1,n_2}\right]\\
			\overset{(*)}{\leq}&\mathrm{Pr}\left[C\frac{nK^2\rho(1-w)}{w\beta}\maxrho+3K_1K_2K_1'K_2'\log n>1/4K_1'K_2'n_1n_2(1-w)\lambda_{n_1,n_2}\right]\\
			&+3K_1^2K_2^2n^{-\gamma}+\frac{2K_1K_2}{n^4}\\
			\leq&\mathrm{Pr}\left[C\frac{nK^2\rho(1-w)}{w\beta}\maxrho>1/8K_1K_2n_1n_2(1-w)\lambda_{n_1,n_2}\right]+3K_1^2K_2^2n^{-\gamma}\\
			&+\mathrm{Pr}\left[3K_1K_2K_1'K_2'\log n>1/8K_1'K_2'n_1n_2(1-w)\lambda_{n_1,n_2}\right]+\frac{2K_1K_2}{n^4},
		\end{align*}
		where (*) comes from Proposition \ref{prop_sbmupper} and \ref{prop_overfit8}. Notice that for any $C>0$, if we have $\mathrm{Pr}[\lambda_{n_1,n_2}<C{\rho_{n_1,n_2}K\maxrho}/{w\beta\min\{n_1,n_2\}K_2}]\ll1/n^2,$ then as $(n_1,n_2)\to\infty$, $$\mathrm{Pr}\left[C\frac{nK^2\rho(1-w)}{w\beta}\maxrho>1/8K_1K_2n_1n_2(1-w)\lambda_{n_1,n_2}\right]\ll\frac{1}{n^2},$$
		thus
		$$\sum_{K_1'>2K_1\text{ or }K_2'>K_2}\mathrm{Pr}\left[C\frac{nK^2\rho(1-w)}{w\beta}\maxrho>1/8K_1K_2n_1n_2(1-w)\lambda_{n_1,n_2}\right]=o(1).$$
		Similarly, since for any $C>0$, we have $\mathrm{Pr}\left[\lambda_{n_1,n_2}<{CK_1K_2\log n}/{n_1n_2(1-w)}\right]\ll1/n^2$, then as $(n_1,n_2)\to\infty$,
		$$\sum_{K_1'>2K_1\text{ or }K_2'>K_2}\mathrm{Pr}\left[3K_1K_2K_1'K_2'\log n>1/8K_1'K_2'n_1n_2(1-w)\lambda_{n_1,n_2}\right]=o(1).$$
		Therefore, as $(n_1,n_2)\to\infty$,
		\begin{align*}
			&\sum_{K_1'>2K_1\text{ or }K_2'>K_2}\mathrm{Pr}\left[L_{K_1,K_2}(A,\mathcal{E}^c)>L_{K_1,K_2'}(A,\mathcal{E}^c)\right]\\
			\leq& o(1)+\sum_{K_1'>2K_1\text{ or }K_2'>K_2}\frac{2K_1K_2}{n^4}+\sum_{K_1'>2K_1\text{ or }K_2'>K_2}3K_1^2K_2^2n^{-\gamma}\\
			\leq& o(1)+\frac{2K_1K_2}{n^2}+3K_1^2K_2^2n^{-\gamma+2}\to0,
		\end{align*}
		as long as we choose an appropriately large $\gamma$. Thus \eqref{eqn:control3} is proved.
	\end{proof}
	
	\begin{proof}[Proof of Theorem 1]
		Combining Propositions \ref{prop_probcontrol1}, \ref{prop_probcontrol2}, \ref{prop_probcontrol3}, \ref{prop_probcontrol4} and \ref{prop_probcontrol5}, it's clear that all the conditions there can be derived from the conditions in Assumption 4. Thus, one can directly obtain
		$$\sum_{(K_1',K_2')\neq(K_1,K_2)}\mathrm{Pr}\left[L_{K_1,K_2}(A,\mathcal{E}^c)>L_{K_1',K_2'}(A,\mathcal{E}^c)\right]\to 0,\quad (n_1,n_2)\to\infty,$$
		which directly implies the desired result.
	\end{proof}

	\subsection{Proof of Corollary 1}
	We first consider the balanced case. Here, we have $n_1\asymp n_2\asymp n$. Therefore, the three conditions in Assumption 4 will write as
	\begin{enumerate}
		\item $n\rho=\Omega((K_2/K)\log n)$.
		\item $n^2\rho^2\min\{1/\log n,1/(K_2\rho)\}\gg K^4K_2^6$.
		\item $n\rho\beta\min\{1,1/(K\rho)\}\gg K^4K_2^3$.
	\end{enumerate}
	The first condition implies that $\rho=\Omega(\log n/n)$, which covers the condition in Proposition \ref{prop_con}. Therefore, this is the sparsity threshlod rate in the balanced case. In order to let this rate holds, plug in $\rho_{n_1,n_2}\asymp \log n/n$ and $\beta=\Omega(K/K_2)$ into the three conditions, we will get
	$$K_2/K=O(1),\qquad \log n\gg K^4K_2^6,\qquad \log n\gg K^3K_2^4,$$
	respectively. The first formula says $K_1\asymp K_2$, and thus the last two formulas will give $K=o((\log n)^{1/10})$.
	
	Then we consider the highly imbalanced case, where $n_1\sim n_2^a$. Then, the three conditions in Assumption 4 will write as
	\begin{enumerate}
		\item $n\rho=\Omega((K_2/K)\log n)$.
		\item $n_2^{a+1}\rho^2\min\{1/\log n,1/(K_2\rho)\}\gg K^4K_2^6$.
		\item $n_2\rho\beta\min\{1,1/(K\rho)\}\gg K^4K_2^3$.
	\end{enumerate}
	Here, the last formula implies that $n_2\rho\gg K^4K_2^3\gg 1$, thus we must have $\rho_{n_1,n_2}\gg 1/n_2\sim n^{-1/a}\gg \log n/n$. Thus, $\rho_{n_1,n_2}\gg 1/n_2$ is the sparsity threshold rate in this regime, and in order this rate can be achieved, the third formula tells us that we must have $K_1\asymp K_2\asymp 1$.
	
	\subsection{Proof of Corollary 2}
	Denote $\bar{K}=\bar{K}_1=\min\{\bar{K}_1,\bar{K}_2\}$. Since we assume that $K_1\leq\bar{K}_1$, $K_2\leq\bar{K}_2$, it is obvious that the three assumptions for $\lambda_{n_1,n_2}$ proposed in Theorem 1 holds if the following three conditions hold:
	\begin{enumerate}
		\setlength{\itemsep}{0.05cm}
		\item $\bar{K}^2\bar{K}_2\Pr[\lambda_{n_1,n_2}>\rho^2/\bar{K}^3\bar{K}_2^4]\to 0.$
		\item For any constant $C>0$, $$\Pr\left[\lambda_{n_1,n_2}<C\frac{\rho_{n_1,n_2} \bar{K}}{\beta_{n_1,n_2}\nmin}\maxbarrho\right]\ll \frac{1}{n^2}.$$ 
		\item For any constant $C>0$, $$\Pr\left[\lambda_{n_1,n_2}<C_2\frac{\bar{K}\bar{K}_2^2\log n}{n_1n_2}\right]\ll \frac{1}{n^2}.$$
	\end{enumerate}
	Then we shall prove that our proposed $\lambda_{n_1,n_2}$ satisfies the above three conditions.
	
	We first prove a concentration result on $\hat{\rho}_{n_1,n_2}$, which we will denote simply as $\hat{\rho}$. By Bernstein's inequality, we have
	$$\Pr\left(\left|\sum_{i,j}A_{ij}-\mathbb{E}\left[\sum_{i,j}A_{ij}\right]\right|\right)\leq 2\exp\left(-\frac{t^2/2}{\sum_{ij}P_{ij}+t/3}\right)\leq\left(-\frac{t^2/2}{n_1n_2\rho+t/3}\right).$$
	Choose $t=\sqrt{12n_1n_2\rho\log n}$, one can prove that 
	$$\Pr\left(\left|\hat{\rho}-\frac{\sum_{kl}n_{1k}n_{2l}B_{kl}}{n_1n_2}\rho\right|\right)\leq 2n^{-3}.$$
	Notice that $C\leq\sum_{kl}n_{1k}n_{2l}B_{kl}/n_1n_2\leq 1$ where $C$ is the constant in Assumption 3, thus there exist constants $0<C_1<C_2$, such that with probability larger than $1-2n^{-3}$, we have
	$$C_1\rho\leq\hat{\rho}\leq C_2\rho.$$
	
	Now we verify the three conditions. For the first one, we have
	\begin{align*}
		&\Pr[\lambda_{n_1,n_2}>\rho^2/\bar{K}^3\bar{K}_2^4]\\
		\leq&\Pr\left[C_1^{2}\frac{\rho^{3/2}}{\bar{K}\bar{K}_2^2\sqrt{\nmin}}\sqrt{\maxbarrho}>\rho^2/\bar{K}^3\bar{K}_2^4\right]\\
		&+\Pr\left[C_1\frac{\rho}{\bar{K}\bar{K}_2}\sqrt{\frac{\log n}{n_1n_2}}>\rho^2/\bar{K}^3\bar{K}_2^4\right]+2n^{-3}\\
		=&\Pr\left[\nmin\rho\min\left\{1,1/\bar{K}\rho\right\}<\bar{K}^4\bar{K}_2^4\right]+\Pr\left[n_1n_2\rho^2<\bar{K}^4\bar{K}_2^6\right]+2n^{-3}\\
		\leq& 2n^{-3},
	\end{align*}
	by the second and third condition in Assumption 4. Thus $$\bar{K}^2\bar{K}_2\Pr\left[\lambda_{n_1,n_2}>\frac{\rho^2}{\bar{K}^3\bar{K}_2^4}\right]\leq 2\bar{K}^2\bar{K}_2n^{-3}\to 0$$, 
	
	For the second one, we have for any $C>0$,
	\begin{align*}
		&\Pr\left[\lambda_{n_1,n_2}<C\frac{\rho \bar{K}}{\beta_{n_1,n_2}\nmin}\maxbarrho\right]\\
		\leq&\Pr\left[C_2^{2}\frac{\rho^{3/2}}{\bar{K}\bar{K}_2^2\sqrt{\nmin}}\sqrt{\maxbarrho}<C\frac{\rho \bar{K}}{\beta_{n_1,n_2}\nmin}\maxbarrho\right]+2n^{-3}\\
		\leq&\Pr\left[\nmin\rho\min\{1,1/\bar{K}\rho\}<\frac{C^2}{C_2^4}\bar{K}^4\bar{K}_2^4\right]+2n^{-3}\\
		<&2n^{-3}\ll n^{-2},
	\end{align*}
	by the third condition in Assumption 4. 
	
	For the last one, we have for any $C>0$,
	\begin{align*}
		&\Pr\left[\lambda_{n_1,n_2}<C\frac{\bar{K}\bar{K}_2^2\log n}{n_1n_2}\right]\\
		\leq&\Pr\left[C_2\frac{\rho}{\bar{K}\bar{K}_2}\sqrt{\frac{\log n}{n_1n_2}}<C\frac{\bar{K}\bar{K}_2^2\log n}{n_1n_2}\right]+2n^{-3}\\
		\leq&\Pr\left[n_1n_2\rho^2<\frac{C}{C_2}\bar{K}^4\bar{K}_2^6\log n\right]+2n^{-3}\\
		<&2n^{-3}\ll n^{-2},
	\end{align*}
	by the second condition in Assumption 4.
	
	Finally, when $K_1$, $K_2$ are fixed (correspondingly, $\bar{K}_1$ and $\bar{K}_2$ are of constant order), since $\hat{\rho}\asymp \rho$ still holds, we have $\max\{1,\bar{K}_1\hat{\rho}\}=1$, and
	$$\frac{\log n}{n_1n_2}=O\left( \frac{\hat{\rho}}{\nmin}\right),$$
	since $\rho=\Omega(\log n/n)$. Therefore, the form of $\lambda_{n_1,n_2}$ in (2) in the main article directly reduces to
	$$\lambda_{n_1, n_2}=C\frac{\hat{\rho}_{n_1,n_2}^{3/2}}{\sqrt{\nmin}}.$$
	
\end{document}